\documentclass[a4paper]{article}
\usepackage{mathtools,amsmath,amssymb,amsthm,amsfonts,graphicx,fullpage}
\usepackage{authblk}

\usepackage{color}
\usepackage{enumerate}
\usepackage[colorlinks=true,citecolor=vert]{hyperref}
\usepackage[utf8]{inputenc}
\usepackage[scientific-notation=true]{siunitx}
\usepackage{bbding}   %PAB

\usepackage{stmaryrd}
\usepackage{boites}

\def\Nset{\mathbb{N}}
\def\Rset{\mathbb{R}}

\def\dps{\displaystyle}
\def\bfo{{\mathbf 1}}

\def\by{{\bar y}}

\def\cB{{\cal B}}

\def\cF{{\cal F}}
\def\cI{{\cal I}}
\def\cJ{{\cal J}}

\def\cM{{\cal M}}

\def\cP{{\cal P}}
\def\cQ{{\cal Q}}
\def\cR{{\cal R}}

\def\cV{{\cal V}}
\def\cX{{\cal X}}

\def\hcJ{\hat{\cJ}}

\newcommand{\comment}[1]{}
\renewcommand{\t}{^{\mbox{\tiny\sf T}}}
\def\diag{{\mathrm {\tt diag}}}

\def\eq{{\mathrm{eq}}}

\def\Xbf{\mathbf{X}}

\def\tx{\tilde{x}}

\definecolor{vert}{rgb}{0.1, 0.7, 0.3}   % PAB
\definecolor{vert_old}{rgb}{0.06, 0.7, 0.6}   % PAB
\definecolor{rouge_old}{rgb}{0.85, 0.26, 0}   % PAB
\definecolor{rouge}{rgb}{0.75, 0.06, 0.3}   % PAB
\definecolor{mauve}{rgb}{0.5, 0.1, 0.99}   % PAB
\definecolor{bleu}{rgb}{0, 0.25, 1}   % PAB
\newtheorem{theorem}{Theorem}
\newtheorem{lemma}{Lemma}
\newtheorem{example}{Example}
\newtheorem{remark}{Remark}
\newtheorem{definition}{Definition}
\newtheorem{assumption}{Assumption}

\newtheorem{corollary}[theorem]{Corollary}
\newtheorem{proposition}[theorem]{Proposition}

{
{
\begin{document}

\title{A framework for the modelling and the analysis of epidemiological spread in commuting populations}
\author[1]{Pierre-Alexandre Bliman\,\Envelope\,\thanks{Corresponding author, email: \href{mailto:pierre-alexandre.bliman@inria.fr}{pierre-alexandre.bliman@inria.fr}}}
\author[2]{Boureima Sangar\'e\thanks{Email: \href{mailto:mazou1979@yahoo.fr}{mazou1979@yahoo.fr}}}
\author[1,2]{Assane Savadogo\thanks{Email: \href{mailto:savadogoass302@yahoo.fr}{savadogoass302@yahoo.fr}}}
\affil[1]{Sorbonne Universit\'e, Inria, CNRS, Universit\'e Paris Cit\'e, Laboratoire Jacques-Louis Lions UMR 7598, Equipe MUSCLEES, Paris, France\vspace{.1cm} }
\affil[2]{Laboratoire de Math\'ematiques, Informatique et Applications, Universit\'e Nazi Boni, Bobo-Dioulasso, Burkina Faso}

\maketitle
\tableofcontents
\medskip

\begin{abstract}
In the present paper, our goal is to establish a framework for the mathematical modelling and the analysis of the spread of an epidemic in a large population commuting regularly, typically along a time-periodic pattern, as is roughly speaking the case in populous urban center.
We consider a large number of distinct {\em homogeneous} groups of individuals of various sizes, called {\em subpopulations},
and focus on the modelling of the changing conditions of their mixing along time and of the induced disease transmission.
We propose a general class of models in which the `force of infection' plays a central role, which attempts to `reconcile' the classical modelling approaches in mathematical epidemiology, based on compartmental models, with some widely used analysis results (including those by P.~van den Driessche and J.~Watmough in 2002), established for apparently less structured systems of nonlinear ordinary-differential equations.
We take special care in explaining the modelling approach in details, and provide analysis results that allow to compute or estimate the value of the basic reproduction number for such general periodic epidemic systems.
%In this \PABD preliminary work, the modelling approach is displayed in details, and first analysis results are provided and demonstrated.
\end{abstract}

\section{Introduction}
\label{se0}

Generally speaking, spatial or geographic heterogeneity plays an important role in the transmission process of many infectious diseases, as well as the temporal variations of the spreading conditions.
%Modelling, simulation and analysis a challenge
The diversity of the corresponding (time and space) scales makes it a challenge to model, simulate and analyze the evolution of such phenomena.
In the present paper, our goal is to establish a framework for the mathematical modelling and the analysis of the spread of an epidemic in a large population commuting regularly, typically along a time-periodic pattern, as is roughly speaking the case in populous urban center.
The latter account for host displacements of varied nature according to the context, and typically distinguish the current location of the individuals, in order to define the possible infection paths.

%\cite{Eubank:2004aa,Arino:2005aa,Wang:2004aa,Zhao:2003aa}

A powerful framework to model the spread of infectious diseases among populations that are naturally partitioned into spatial sub-units, is based on {\em metapopulation models}, see \cite{Hanski:1997aa,Keeling:2008aa} and~\cite{Arino:2006aa,Arino:2009aa} for a quite complete overview.
Practically, `a metapopulation model involves explicit movements of the individuals between distinct locations'~\cite{Arino:2009aa}, called {\em patches}.
The fluxes between the different patches, seen as the vertices of an associated graph, are usually described through a discrete Laplacian matrix and take place continuously.
The transfers between patches are generally instantaneous, but some contributions modelled specifically the infection occurring in transportation systems, see~\cite{Knipl:2013aa,Nakata:2015aa} and references therein.
%\cite{Liu:2008aa,Nakata:2011aa,Knipl:2012aa,Knipl:2013aa,Nakata:2015aa,Knipl:2016aa}
At the price of an increased number of variables, some models keep track of the `origin' of each individual (e.g.~through a {\em residency patch}), but most of them just account for the number of individuals that are present at each location at the current time, making them otherwise indistinguishable.
This omission is inadequate to describe precisely the complex, but quite regular, displacement schemes that shape urban commutations.
In addition, the `residency patch' is not sufficient to characterize univocally the history of the encounters, which are as many opportunities of inter-infection.

On the other hand, much work has been made to integrate various heterogeneity traits relevant to describe the spread of an infection within a population.
Such approaches consider structured populations subdivided into {\em subpopulations} that differ from each other~\cite{Hethcote:1987aa}.
As written by H.W.~Hethcote and J.W.~van Ark~\cite{Hethcote:1987aa}, the latter  `can be determined not only on the basis of disease-related factors such as mode of transmission, latent period, infectious period, and genetic susceptibility or resistance, but also on the basis of social, cultural, economic, demographic, and geographic factors.
In a spatially heterogeneous population the subpopulations can be different schools, neighborhoods, cities, states, countries, or continents.'
This also includes the heterogeneity of the contact rates and of the behaviours that impact the disease spread.
Attempts have been made to distinguish between `social' encounters, occurring between individuals of the same subpopulation (on a somehow `regular' or `predictable' basis), and random encounters between individuals of possibly different subpopulations~\cite{Sattenspiel:1988aa}.
Also, some structured models accommodate geographical aspects, in the way of metapopulation~\cite{Sattenspiel:1995aa}.

An important characteristic of the `subpopulations' considered in this type of models is their {\em homogeneity}:
%n a sense, for these contributions, the `subpopulations' constitute the elementary constituents of the population: their
the members of each subpopulation are perfectly mixed, submitted to the same `encounters' with other subpopulations, and they cannot usually change from one subpopulation to the other (so that in particular, age classes are usually not treated as subpopulations).
In other words, within each subpopulation, the individuals are indistinguishable from one another from the point of view of the epidemic dynamics.

Generally speaking, in these settings, the movements are rather thought of as permanent and stationary, rather than time-varying, phenomena.
On the other hand, the evolution of disease transmission is usually subject to time fluctuations, and in particular periodic fluctuations~\cite{Wang:2008aa}.
The latter may result from seasonal variations induced by the climatic conditions, but also from human activities at various scales (daily commuting, opening and closing of schools, vaccination programs, etc.).
Efforts have been made to analyze the effects of such forcing conditions on infection spread~\cite{Bacaer:2007aa,Bacaer:2007ab,Bacaer:2012aa}.
Due to the difficulty to analyze the behaviour of periodic nonlinear systems in general, several simplifications have been made, amounting to analyze {\em stationary} systems obtained from a (usually informal) averaging of the periodic effects.
However, examples in~\cite{Wang:2008aa} have established that such simplifications may produce overly optimistic or overly pessimistic estimation of the occurrence of an epidemic (overestimation or underestimation of the basic reproduction number).
Modelling and analysis of time-periodic epidemic models therefore remain of great interest, compared to the use of `residence time' stationary models.
This point is already true for the determination of the basic reproduction number, deduced from the model behaviour in the vicinity of a disease-free equilibrium; but is even more necessary for important issues outside the linearity range of the infection, such as e.g.~the determination of the epidemic final size~\cite{Miller:2012aa,Bacaer:2009ab}.\\

Commuting involves complex interactions between population groups through complicated, but mostly repetitive, spatial and temporal patterns of encounters.
In the present paper, in order to account for such situations, we shift the paradigm from a {\em spatial} `Eulerian' description of the population and its displacements, to a `Lagrangian' description in terms of {\em meeting and mixing} of homogeneous groups followed along their evolution.
To carry out this program, we introduce (a large number of) distinct {\em subpopulations}, which are homogeneous groups of individuals of various sizes that constitute a partition of the total population.
To fix ideas, here a subpopulation represents typically the group of all individuals living in a same location A that go at the same time to another location B (which represents e.g.~school, work, shops, but also means of public transportation\dots), and then to a third location C at the same moment, etc.
A key point is that all individuals of a given subpopulation have the same `history'; that is, all are in contact at every moment with exactly the same other subpopulations.
This is similar in spirit with the homogeneity assumption in the subpopulation models mentioned above.
The example of a subpopulation with five subpopulations is provided in Figure~\ref{fi1}, and two different partitions are shown in Figure~\ref{fi2}.

\begin{figure}
\begin{center}
\includegraphics[scale=0.6]{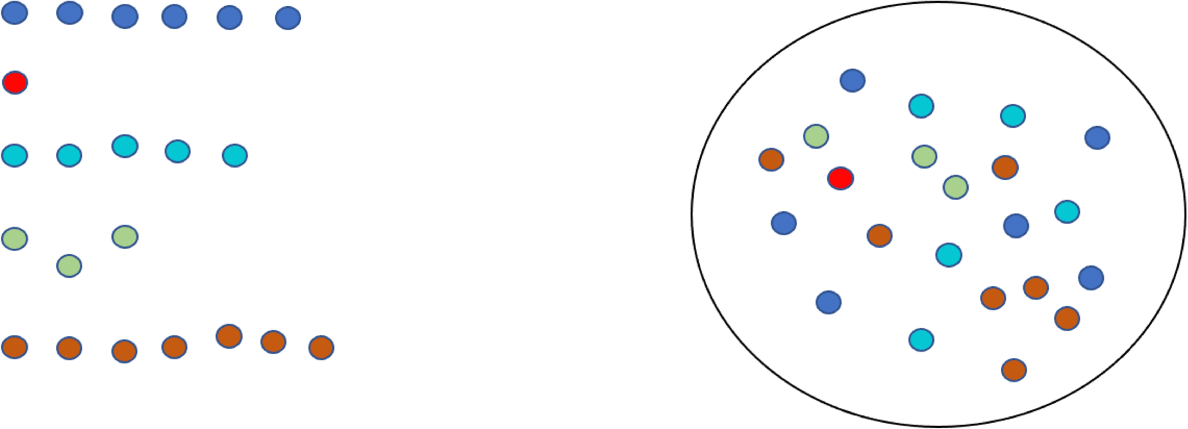}
\caption{A population with 22 individuals, divided into 5 subpopulations $p_i$, $i=1,\dots,5$.
Their cardinals are respectively $|p_1|=6$, $|p_2|=1$, $|p_3|=5$, $|p_4|=3$, $|p_5|=7$.
The set of subpopulations is $\cP=\{p_1, p_2, p_3, p_4, p_5\}$, of cardinal $|\cP|=5$.}
\label{fi1}
\end{center}
\end{figure}
   
\begin{figure}
\begin{center}
\includegraphics[scale=0.6]{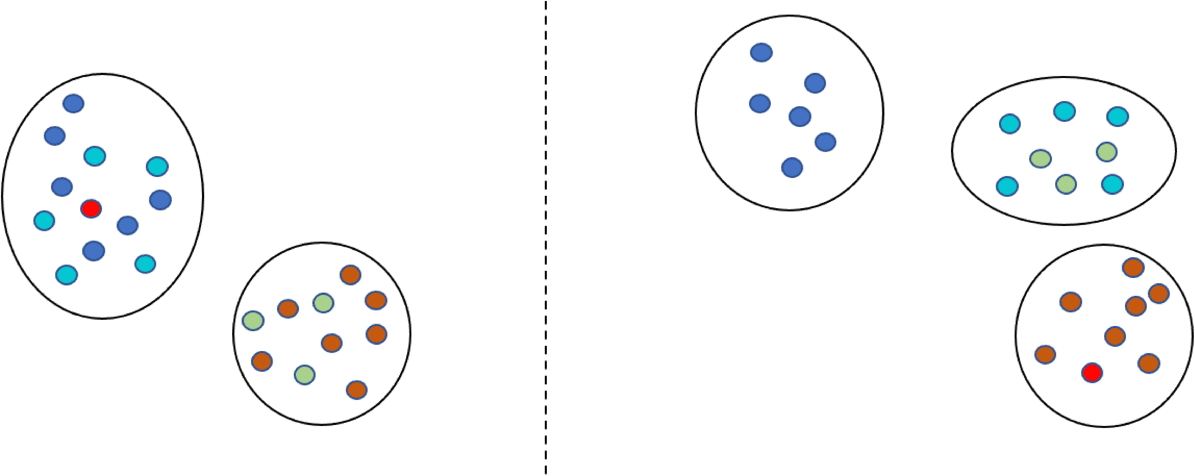}
\caption{Two different partitions $\cR_1, \cR_2$ of the set of subpopulations $\cP$ shown in Figure~\ref{fi1}.
On the left, the set of classes $\cQ_1 := \cP\setminus\cR_1 = \{(p_1, p_2, p_3), (p_4, p_5)\}$ contains $|\cQ_1|=2$ elements, namely the class $q_1:=\{p_1, p_2, p_3\}$ with $|q_1|=3$ elements and the class $q_2:=\{p_4, p_5\}$ with $|q_2|=2$ elements.
On the right, $\cQ_2:= \cP\setminus\cR_2 = \{(p_1), (p_2, p_5), (p_3, p_4)\}$ contains $|\cQ_2|=3$ elements, namely the classes $q_1:=\{p_1\}$ with $|q_1|=1$ elements, $q_2:=\{p_2, p_5\}$  with $|q_2|=2$ elements and $q_3:=\{p_3, p_4\}$  with $|q_3|=2$ elements.
In both configurations all individuals from a same subpopulation are together in the same class, and one may check that the balance formula~\eqref{eq376} given later holds in both cases: $|\cP|=|q_1|+|q_2|+|q_3|+|q_4|+|q_5|$.
This is a consequence of the fact that in a partition, every subpopulation pertains exactly to one class.}
%   \caption{Partition in force when $t\in\cI_1$: $\cQ(\cI_1) = \{(p_1, p_2, p_3), (p_4, p_5)\}$}
\label{fi2}
\end{center}
\end{figure}

In a complementary way, we assume that in each location, the different subpopulations are {\em perfectly mixed}.
This assumption will allow to describe the epidemiological evolution at each fixed location where a fixed mixture of subpopulations through a {\em unique} aggregate compartmental model.
Of course, this evolution is
%conditions of disease transmission are
a priori different in different locations, due to different population densities and different behaviors modifying the contagion rate (e.g.~at home, in public transportation\dots), and also due to the meeting of different subpopulations (e.g.~at work, at school, during seasonal holidays\dots).
This framework is quite distinct from most models that use subpopulations to represent heterogeneity traits: here the main difference between the subpopulations lies in the history of their contacts with other subpopulations, and in the local conditions of infection transmission, with otherwise identical characteristics.
On the other hand, the modelling of heterogeneous traits may be achieved as usual through pertinent choice of the state variable.

%This does not differ in spirit with the subpopulation models mentioned above.

% Cui:2008aa 422
% Capasso:1978aa 1363
% Liu:1987aa 991
% Collinson:2014aa 105
% Cui:2008aa 423
% Li:2008aa 20
% Liu:2012aa 131
% Sun:2011aa 213
% Wang:2005aa 102
% Xiao:2007aa 631

The implementation of this program requires two steps: first to model the alternating mixture of different subpopulations; and second to describe how the epidemic spreads in a homogeneous melange of subpopulations, in a given location during a given time interval.
On the one hand, the commuting and the induced changes in the contacts between subpopulations are modelled as piecewise-constant, periodic, mixing of the latter.
We abandon any reference to `physical locations', in favour of a vision in terms of {\em partitions} of the set of the different subpopulations.
In this sense, trains or metros are also `locations', and `physical locations' that are actually empty (schools at night\dots) do not have to appear.
On the other hand, to describe how the disease spreads locally, we propose an original class of compartmental models in which the {\em force of infection} $\lambda$ is explicitly distinguished and infection occurs by application of the latter to some (usually susceptible, but not only) compartments.
Generally speaking, the force of infection is defined as `the {\em per capita} rate at which susceptible individuals contract the infection'~\cite{Keeling:2008aa}.
While the latter is usually taken as proportional either to the number of infected or to the proportion of infected, more complex models have been used, e.g.~in order to take into account protective measures and intervention policies, or to reflect the necessity of multiple contacts to transmit the disease, see~\cite{Capasso:1978aa,Liu:1987aa,Wang:2005aa,Xiao:2007aa,Cui:2008aa,Cui:2008aa,Li:2008aa,Sun:2011aa,Liu:2012aa,Collinson:2014aa}.
We capture explicitly the force of infection in the class of models proposed below.
This induces a special shape of the right-hand sides of the ODE models, which are jointly {\em linear with respect to the state variable $x$} and {\em affine with respect to the force of infection $\lambda$}.
Notice that we consider in general several infective forces, first of all according to every location, but also to every susceptible compartment in the case of `multi-group population', like done e.g.~in {\rm\cite[Section 23.1]{Thieme:2003aa}}.
This leads, to a {\em vector-valued} force of infection.

Of course, the value of the force of infection is as usual a (nonlinear) function $\lambda(x)$ of the state variable $x$.
The framework adopted here offers the advantage that, when two subpopulations described respectively by state vectors $x_1, x_2$ come into contact (i.e.~pertain to the same class in the current partition), the resulting force of infection is simply $\lambda(x_1+x_2)$ and applies identically to each of them, with in fact no other interaction between the two subgroups.

This feature has an important consequence.
Many models of mathematical epidemiology describe the evolution of the epidemiological status of a given population in terms of {\em proportions} (of susceptible individuals, of infected individuals, etc.), allowing in particular to assume constant recruitment terms.
Here, due to our modelling choice, it is necessary to have the ability to {\em sum up} the values at each compartment of the different populations.
As a result, the models introduced here manipulate exclusively {\em extensive quantities}, and it is usually necessary to substitute constant recruitment terms by expressions proportional to the size of each subpopulation, in order to allocate realistically the demographic effects.

Up to our knowledge, the use of periodic epidemiological model structured in subpopulations with the purpose of accounting for commuting is original.
In contrast with the models mentioned above, here all subpopulations will have essentially the same behaviour with respect to the infection spread: they differ only through the history of their displacements and the local definition and intensity of the force of infection, and subpopulations that pertain to the same equivalence class of the partition at a certain time participate jointly and identically %homogeneously
to the spreading.

We display in the present paper the proposed modelling framework and provide some key analysis results extending the approach in~\cite{Driessche:2002aa}.
This has to be seen as a preliminary step towards more complex results that will fully exploit the modelling capabilities.
% to deduce analysis insights.}
%\PABD basic results. This has to be seen as a preliminary step towards more complex results that will fully exploit the modelling capabilities to deduce analysis insights.}
The modelling approach based on compartmental models has produced numerous variants in mathematical epidemiology~\cite{Capasso:1993aa,Diekmann:2000aa,Keeling:2008aa,Brauer:2012aa}.
Attempts have been made to obtain general classes regrouping related models, liable to common general analysis results.
See for example {\em multi-group models}~\cite{Thieme:2003aa,Guo:2012aa}, {\em metapopulation models}~\cite{Wang:2004aa,Arino:2005aa,Arino:2009aa}, {\em models of differential susceptibility and infectivity}~\cite{Iggidr:2007aa,Bonzi:2011aa}.
%\cite{Eubank:2004aa,Arino:2005aa,Wang:2004aa,Zhao:2003aa}

%A mon avis il s’agit d’un travail qui permet de mettre sous une  forme unique un grand nombre de modèles : les multi-groupes (au sens de Thieme, les modèles de Michael Li), les metapop ( Juien Arino) et les modèles DIDS (differential susceptibility et Infectivity) 
In a sense, the present work constitutes an attempt to unify them, and to `reconcile' the classical modelling approaches
% in mathematical epidemiology~\cite{Capasso:1993aa,Diekmann:2000aa,Keeling:2008aa,Brauer:2012aa}, based on compartmental models,
with some analysis results~\cite{Driessche:2002aa,Wang:2008aa} established for apparently less structured and less specific systems of equations.
It is an effort to contribute to rationalizing and systematizing the modelling and analysis of epidemics, in line with the work~\cite{Hethcote:1994aa}.\\
%Unfortunately, due to space limitation, we are not able to show here how the proposed class accommodates most compartmental models previously shown in the literature\footnote{As a matter of fact, we have not been able to isolate an infection mechanism that couldn't fit into this class.}.
%This will be exemplified extensively in a forthcoming complete version of the paper.\\
%This will be done in a complete version of the paper with various examples.}\\

%the vision of epidemic models provided through in classical monographs

The paper is organised as follows.
We display first in Section \ref{se1} the representation of the changing melange at the scale of the global population.
We then propose in Section \ref{se2} a class of epidemiological models that describe locally the infectious processes, in which the force of infection is explicitly depicted.
We show by ample examples that it includes various compartmental models previously shown in the literature.
These two scales of description are built in a consistent way, and are put together in Section \ref{se3}, which summarizes the key modelling contribution of the paper, consisting of an original class of models of epidemics in commuting populations.
% {\PAB (Definition~\ref{de4}).}
Adequate assumptions (numbered~\ref{as1},~\ref{as3},~\ref{as4}) are presented, and it is shown how they allow to recover essentially the framework developed by van den Driessche and Watmough~\cite{Driessche:2002aa} and to compute the basic reproduction number $\cR_0$ of an epidemic in a given, unique, subpopulation under stationary conditions
(Theorem~\ref{th1}).
The general analysis results are then given in Section~\ref{se6}, indicating how to compute the basic reproduction number for the proposed general class under periodic commutations and providing some applications
%general epidemiological models in commuting population
(Theorem~\ref{th5}, Corollary~\ref{co0} and Lemma~\ref{le87}).
Concluding remarks are provided in Section~\ref{se7}.

\paragraph{Notations}
A list of the notations specific to the present paper is provided in Table~\ref{fi3}.
We recall here some usual mathematical notions.

$\bullet$
By definition let the vector $\bfo_n\in\Rset^n$, $n\in\Nset\setminus\{0\}$ be such that $\bfo_n\t :=\begin{pmatrix} 1 & \dots & 1 \end{pmatrix}$.
$\bullet$
The {\em componentwise} order relation on any Cartesian power of the set $\Rset$ is denoted as usual $\geq$.
The relation $>$ means `larger but not equal', and the notation $\succ$ is used to denote the strict ordering of every component ($0 \prec x\in\Rset^n$ means $0 < x_i$, $i=1,\dots,n$).
$\bullet$
The {\em spectral radius} of a real square-matrix $M$ (that is the maximum of the moduli  of its eigenvalues) is denoted by $\rho(M)$; its {\em stability modulus} (that is the maximum real part of its eigenvalues) is denoted $s(M)$.

\begin{table}
\label{fi3}
\begin{center}
\begin{tabular}{|l|l|}
\hline
$x$ & State variable of the epidemiological model\\
$x^I$ & Infected compartments of the state variable\\
$x^U$ & Non-infected compartments of the state variable\\
$\lambda$ & Vector of the forces of infection\\
\hline
$n$ & Dimension of the state variable $x$\\
$n_I$ & Dimension of the infected compartment state $x^I$\\
$n_U$ & Dimension of the non-infected compartment state $x^U$\\
$n_\lambda$ & Dimension of the vector of the forces of infection $\lambda$\\
\hline
$\Xbf_s$ &
Set of all disease-free states ($\Xbf_s \subset \Rset_+^n$)\\
$\Xbf_\eq$ &
Set of all disease-free equilibrium points ($\Xbf_\eq \subset \Xbf_s$)\\
\hline
\hline
$\cP$ & Set of the subpopulations\\
$\cR$ & Partition of $\cP$ (possibly time-dependent)\\
$\cQ$ & Quotient set $\cP\setminus\cR$ of the equivalence classes (`loci') of subpopulations by $\cR$\\
$\cX$ & Global state variable (of dimension $|\cP|n$)\\
\hline
$p$ & Subpopulation name\\
$[p], [p]_\cR$ & Class of the subpopulation $p$\\
$q$ & Class name\\
\hline
$i$ & Subpopulation index ($i\in\{1,\dots, |\cP|\}$) or compartment index ($i\in\{1,\dots, n\}$) \\
$j$ & Class index\\
\hline
\hline
$m$ & Interswitching interval number (in periodic evolution)\\
$k$ & Interswitching interval index \\
\hline
\end{tabular}
\caption{Nomenclature.
The cardinals of the sets $\cP, \cQ$ are written $|\cP|, |\cQ|$, and the cardinal of a class $q\in\cQ$ is written $|q|$.
Without creating confusion, the same letter $i$ is used to index the subpopulations and the compartments in the state vector of each subpopulation.
The notations and explanations related to the model itself are to be found in the text.}
\end{center}
\end{table}

\section{Representation of commuting populations}
\label{se1}

In order to represent commuting populations, we first introduce in Section~\ref{se11} a notion of {\em partitioned-population model}, with several subpopulations evolving in parallel, in contact or out of contact.
We then introduce in Section~\ref{se115} a way to represent the commutations.
This allows to define {\em models of commuting populations} in Section~\ref{se12}, which are basically partitioned-population models where the contact pattern between subpopulations is varying with time.
%; and in Section~\ref{se13} a variant of the latter called {\em models of commuting populations with random raw}.
% according to a so-called {\em (admissible) switching signal}.

\subsection{Partitioned-population models}
%\section{Lagrangian models}
\label{se11}

A basic compartmental model describing the epidemiological evolution within an isolated, homogeneous and perfectly mixed population, is given~\cite{Capasso:1993aa,Diekmann:2000aa,Keeling:2008aa,Brauer:2012aa} as a dynamical system on $\Rset_+^n$ defined by a system of ordinary differential equations
\begin{equation}
\label{eq0}
\dot x = f(x(t),\lambda(t)),\qquad \lambda(t) := \lambda(x(t)),
\end{equation}
for  $f\ :\ \Rset_+^n\times \Rset_+^{n_\lambda} \to \Rset^n$, $\lambda\ :\ \Rset_+^n \to \Rset_+^{n_\lambda}$.
%Model \eqref{eq0} describes the epidemiological evolution within an isolated, homogeneous and perfectly mixed population.
The vector $x$ bears the {\em numbers of individuals} in each of the model compartments, and the vector $\lambda$ gathers {\em forces of infection} exerted on the components of $x$.
{\em Frequency-dependent} transmission mechanisms are typically rendered by a function $\lambda$ positively homogeneous of degree $0$, and {\em density-dependent} transmission mechanisms by function positively homogeneous of positive degree\footnote{Recall that (in an adequate setting) a function $g$ is called positively homogeneous of degree $d$ if $g(\alpha x) = \alpha^d g(x)$ for any positive scalar $\alpha$ and any $x$ in its domain.}, but as mentioned in Section~\ref{se0}, other more complex expressions may be found in the literature.
%Due to the heterogeneity of the population, t
The dimension $n_\lambda$ of the latter may be larger than 1, see examples in Section \ref{se2}.

%We define, as in van den Driessche and Watmough \cite{}, the set \mathbf{X}_\text{s} :=

%\pab{Natural assumptions on the model \eqref{eq0} involve Lipschitz property for $f$ and $\lambda$, in order to have well-posedness of the Cauchy problem; uniform asymptotic boundedness of the solutions, and existence of a unique Disease-Free Equilibrium (DFE).}
%
%
%\subsubsection{Definitions}
%
%We call model of partitioned population\dots

\begin{definition}[Solution of a partioned-population model]
\label{de1}
Let $\cP$ be a {\em finite} set, called the {\em set of subpopulations,} and $\cR\subset\cP\times\cP$ be an equivalence relation on $\cP$.
For any subpopulation $p\in\cP$, the corresponding equivalence class of $p$ in the quotient set $\cQ := \cP\backslash\cR$ is denoted $[p]_\cR$, simply abbreviated $[p]$ when no confusion is possible.
We call
{\em $(\cP,\cR)$-solution of model \eqref{eq0}} any continuous function 
 $\cX\ :\ \Rset_+\to\Rset_+^{n|\cP|}$, $\cX:=(x_1,\dots,x_{|\cP|})$, continuously differentiable on $\Rset_+$ such that
%{\em Lagrangian Metapopulation model }
%the stationary dynamical system defined on $\Rset_+^{n|\cP|}$ by
\begin{equation}
\label{eq1}
\dot x_p = f(x_p(t),\lambda_{[p]}(t)),\qquad \lambda_{[p]}(t) := \lambda_{[p]} \left(
x_{[p]} (t)
\right),\qquad
x_{[p]} := \sum_{p'\in [p]} x_{p'},
\qquad p\in\cP.
\end{equation}
\end{definition}

%\begin{remark}
%If necessary, one could also adopt the possibility to have $f$ depending upon the class, using the notation $f_q$\dots
%\end{remark}

$\cP$ represents the index of different subpopulations, while the equivalence relation $\cR$ defines which of them may infect each other.
Therefore, one may consider each equivalence class as a {\em location}, namely the place where are present every subpopulation pertaining to this class.
However, what is done here is not exactly to define the {\em place} where is located a given population at a given time, but rather which are the other populations with which it interacts.
In particular, this perspective suppresses the necessity of any absolute notion of (geographical) location: after a switch in the repartition of the populations, the new partition is a priori not related to the previous one, and the equivalence classes change and become incomparable without supplementary information.

In \eqref{eq1}, the conditions of transmission of the infection at time $t$ in a certain location $[p]$ depend upon the total population present therein at that time, which is exactly what is written $x_{[p]}(t)$.
These conditions may be different from one place to the other: this is rendered by the dependence of $\lambda_{[p]}$ upon $[p]$.
On the other hand, the {\em mechanisms of the infection} themselves are the same in every location, and for this reason the function $f$ is the same at every location.
Recall that, as mentioned in Section \ref{se0}, the definition of $x_{[p]}$ in \eqref{eq1} makes it essential to place in the components of $x$ {\em extensive quantities} such as numbers of individuals, and not relative proportions or densities.

Another point worth noting is that the coupling of two different subpopulations $p,p'\in\cP$ that are present together in the same place (i.e.~$[p]=[p']$) is reflected exclusively through the shared force of infection.

Summing up the contributions of every population within a given class yields:
\begin{equation}
\label{eq10}
\dot x_{[p]} = \sum_{p'\in [p]} f(x_{p'}(t),\lambda_{[p]}(t)),\qquad \lambda_{[p]}(t) = \lambda_{[p]} \left(
x_{[p]} (t)
\right),\qquad
p\in\cP.
\end{equation}
If $f$ is linear with respect to its first argument (an assumption that will be explained and adopted in Section~\ref{se2}), the effective number of equations in \eqref{eq10} is the number of components of the partition $\cR$, that is the cardinal $|\cQ|$ %$|\cP\backslash\cR|$
(rather than $|\cP|$), times the state dimension $n$ of the model~\eqref{eq0}.
As a matter of fact, one then obtains by adding all equations within a given equivalence class
\begin{equation*}
%\label{eq33}
\dot x_{[p]} = f(x_{[p]}(t),\lambda_{[p]}(t)),\qquad \lambda_{[p]}(t) = \lambda_{[p]} \left(
x_{[p]} (t)
\right),\qquad p\in\cP,
\end{equation*}
or more simply:
\begin{equation}
\label{eq3}
\dot x_q = f(x_q(t),\lambda_q(t)),\qquad \lambda_q(t) = \lambda_q \left(
x_q (t)
\right),\qquad q\in \cQ.
%\cP\backslash\cR.
\end{equation}
This describes the evolution in any `location' $q$.

Let $q\in \cQ$ be given and assume $x_p(0)\in\Rset_+^n$ given, for any $p\in\cP$ such that $[p]=q$.
From the solution of the nonlinear problem \eqref{eq3}, one may compute the solution of the different interacting subpopulations $p$ such that $[p]=q$, as follows.
\begin{enumerate}
\item
Compute
\begin{equation}
\label{eq4}
x_q(0) = \sum\{x_p(0)\ :\ p\in\cP,\ [p]=q \}.
\end{equation}
\item
Solve \eqref{eq3} with the initial condition \eqref{eq4}.
This provides $x_q(t)$, and as a by product, this also provides $\lambda_q(t) = \lambda_q \left(
x_q (t)
\right)$.

\item
Then, for any $p\in\cP$ such that $[p]=q$, one has
\[
\dot x_p = f(x_p(t),\lambda_q(t)),\qquad x_p(0) \text{ known.}
\]
Solving the Cauchy problem for this {\em time-varying} ordinary differential equation, one obtains the solution $x_p$ for any $p\in q$.
More specifically, as $f$ is linear with respect to its first argument, the solution yields the following algebraic relation
\begin{equation}
\label{eq44}
x_p(t) = \Phi_q(t) x_p(0),\qquad t\geq 0,
\end{equation}
where $\Phi_q(\cdot)$ is the fundamental matrix of the {\em linear, time-varying,} equation
\[
\dot y = f(y,\lambda_q(t)),\qquad \lambda_q(t) = \lambda_q \left(
x_q (t)
\right),
\]
that is by definition
\begin{equation}
\label{eq454}
\dot\Phi_q = f(\Phi_q(t),\lambda_q(t)),\quad t\geq 0,\qquad
\Phi_q(0) = I_n.
\end{equation}
By convention, the result of applying $f$ to the matrix $\Phi_q(t)$ in the right-hand side of~\eqref{eq454}, is the matrix obtained as the concatenation of the vectors that result from applying $f$ to each of the column vectors of $\Phi_q(t)$.
In other words, the $i$-th column of $\Phi_q(t)$ contains the value at time $t$ of the trajectory initiated on the $i$-th vector of the canonical basis at time $0$.

%Notice that, as underlined by the notation in~\eqref{eq44}, $\Phi_q$ is identical for all $p\in q$.
\end{enumerate}
The previous consideration is quite important from the point of view of numerical simulation.
As a matter of fact, solving~\eqref{eq1} as a system of $|q|$ coupled equations engages $|q|n$ scalar variables for any subpopulation class $q$; while computing jointly the solution of~\eqref{eq3} with initial condition~\eqref{eq4} amounts to solve a system of ODEs with $n$ scalar variables, and solving~\eqref{eq454} (in sequence or in parallel with the previous one) necessitates $n^2$ scalar variables, independently of the size of the class $q$.
This number is usually much smaller than the previous one.
The same observation extends directly to the piecewise autonomous models presented below in Section~\ref{se12}.\\ %and~\ref{se13}.\\

Before going further, we provide now explicit expressions of $f$ and $\lambda$ for an elementary example.
%\subsubsection{A}
%\label{se12}
%By using the compartmental model in epidemiology given by \eqref{eq0},

\begin{example}
\label{ex1}

We consider the simple case where \eqref{eq0} is an SIR model with no induced mortality, that is
\begin{equation}
\label{eq555}
\dot S = \mu N - \beta S\frac{I}{N} -\mu S,\qquad \dot I = \beta S\frac{I}{N} -(\gamma+\mu) I,\qquad \dot R = \gamma I -\mu R,
\end{equation}
where  $\beta, \mu, \gamma$ are positive parameters, and $N := S+I+R$ is the total population value, invariant along every trajectory due to the fact that $\dot N \equiv 0$.
In order to comply with the analysis framework constructed by van den Driessche and Watmough~{\rm\cite{Driessche:2002aa}}, we put first in the state vector $x$ the (unique) {\em infected} compartment $I$, and take 
% Assuming also that there is no induced mortality caused by the disease, and that the population is initially at its equilibrium level, this means that one has here in \eqref{eq0}:
\begin{equation}
\label{eq5}
x:= \begin{pmatrix} I \\ S \\ R \end{pmatrix},\quad
f(x,\lambda) := \begin{pmatrix}
\lambda S - (\gamma+\mu) I \\ - \lambda S + \mu (I+R) \\ \gamma I - \mu R \end{pmatrix}
= \begin{pmatrix}
- (\gamma+\mu) & \lambda & 0 \\ \mu & -\lambda & \mu \\ \gamma & 0 & - \mu  \end{pmatrix} x
,\quad
\lambda := \beta\frac{I}{S+I+R}.
\end{equation}
The previous notation showcases the linear dependance of the function $f$ with respect to its first argument $x$ and its affine dependance with respect to $\lambda$.
These properties are one of the key elements of the framework introduced later in Section~{\rm\ref{se2}}.

The situation corresponding to system \eqref{eq1} (or \eqref{eq3}) is retrieved by introducing the function $f$ in \eqref{eq5} together with the definition of $\lambda_q$, in such a way that, for any $p\in\cP,$ the evolution of disease analogous to system \eqref{eq1} is  given by:
	\begin{equation}
	\label{eq6}
	\hspace{-.1cm}
	x_p:= \begin{pmatrix} I_p \\ S_p \\ R_p \end{pmatrix},\quad
	\dot x_p
	= \begin{pmatrix}
- (\gamma+\mu) & \lambda_{[p]} & 0 \\ \mu & -\lambda_{[p]} & \mu \\ \gamma & 0 & - \mu  \end{pmatrix} x_p
%	- \lambda_{[p]} & \mu & \mu \\ \lambda_{[p]} & - (\gamma+\mu) & 0 \\ 0 & \gamma & - \mu  \end{pmatrix} x_p
	,\quad
	\lambda_{[p]} = %F_{[p]}(x) :=
	\beta_{[p]}\frac{I_{[p]}}{S_{[p]}+I_{[p]}+R_{[p]}},\quad t\geq 0,\quad p\in\cP,
	\end{equation}
for given positive parameters $\beta_q$, $q\in \cQ$.
The latter may be identical (homogeneous transmission conditions), or not (heterogeneous transmission conditions).
In \eqref{eq6}, in agreement with equation \eqref{eq1}, one has $I_{[p]} := \sum_{p'\in [p]} I_{p'}$, and similarly for $S_{[p]}$ and $R_{[p]}$.
\end{example}

%See Section \ref{se12} below.

\subsection{Admissible switching signals}
\label{se115}

In order to introduce the notion of commuting population models in the next Section~\ref{se12}, we first define the key notion of {\em admissible switching signals.}

%\subsubsection{Admissible switching signals}

\begin{definition}[Admissible switching signals]
\label{de2}
Let $\cP$ be a finite set of subpopulations.
We call {\em admissible switching signal} any function $\cR(\cdot)$ which associates to any $t\geq 0$ an equivalence relation (or equivalently a partition) on $\cP$, denoted $\cR(t)$, and such that $\cR(\cdot)$ {\em is constant on any bounded set of $\Rset^+$, except on a finite number of points} (called {\em switches}).

For any admissible switching signal $\cR$, we denote $\cI$ the (open) set of time instants on which $\cR$ is continuous.
Otherwise said, $\Rset_+\setminus\cI$ is the set of {\em switching times}.
\end{definition}

For any admissible switching signal, there exists a finite or denumerable number of switches, and the set $\cI$ is a (finite or denumerable) union of successive open intervals sharing common endpoints.

In order to represent commutations, our interest lies especially in switching signals showing repeating pattern.
%\paragraph{Representation of periodic switching signal}
%\subsubsection{Periodic admissible switching signals}
For this, we will consider in the sequel admissible switching signals $\cR(t)$ which are {\em $T$-periodic} for some $T>0$.
Denoting $m\in\Nset$ the number of commutations of such a signal $\cR(\cdot)$ within each period, we will represent the corresponding switching pattern by an element of some set $\Lambda_m$, $m>0$,
%of inter-switch durations during a period,
defined as
\begin{equation}
\label{eq58}
\Lambda_m=
\left\{
(\tau_0,\tau_1,\dots,\tau_m) : 0 = \tau_0 < \tau_1 < \dots < \tau_m =T
%(d_1,\dots,d_m) : d_k>0, \dps\sum_{i=1}^md_k=T
\right\}.
\end{equation}
For convenience, we define also the $m$ unions $\cI_k$ of open intervals, $k=1,\dots, m$, by
\[
\forall t\geq 0,\qquad
t\in\cI_k \ \Leftrightarrow\ t \in \left(
\tau_{k-1},\tau_k
\right) \quad \text{mod}\ T.
\]
One then has $\cI = \bigcup_{k=1}^m \cI_k$, for $\cI$ defined in Definition \ref{de2}.

As no commutation occurs during the intervals contained in $\cI_k$, we may denote unambiguously
\begin{itemize}
\item
$\cR(\cI_k)$ the partition in force at any $t\in\cI_k$, instead of $\cR(t)$;
\item
$[p]_{\cR(\cI_k)}$ the corresponding class of a subpopulation $p\in\cP$, instead of $[p]_{\cR(t)}$;
%$\cR(\cI_k)$, instead of $\cR(t)$, the partition in force at any $t\in\cI_k$;
\item
and $\cQ(\cI_k) = \cP\setminus\cR(\cI_k)$ the set of subpopulation classes on the set $\cI_k$, instead of $\cQ(t)$, $t\in\cI_k$.
\end{itemize}

\subsection{Models of commuting populations}
\label{se12}

%\subsubsection{Models of commuting populations}

We are finally in position to introduce the models of commuting populations.
Consider the following system of equations
\begin{subequations}
\label{eq2}
\begin{gather}
\label{eq2a}
\hspace{-.4cm}
\dot x_p = f(x_p(t),\lambda_{[p]_{\cR(t)}}(t)),\quad
\lambda_{[p]_{\cR(t)}}(t) = \lambda_{[p]_{\cR(t)}} \left(
x_{[p]_{\cR(t)}}(t)
\right),\quad 
x_{[p]_{\cR(t)}} := \hspace{-.2cm} \sum_{p'\in [p]_{\cR(t)}} \hspace{-.2cm} x_{p'},
\quad t\in\cI,\quad p\in\cP,\\
\label{eq2b}
x_p(t^+) = x_p(t^-),
\quad t\in\Rset_+\setminus\cI,
\quad p\in\cP.
\end{gather}
\end{subequations}
%
%\begin{definition}[Epidemiological model in commuting population]
%\label{de4}
We call {\em epidemiological model in commuting population} the system~\eqref{eq2}.
%\eqref{eq250}-\eqref{eq251}.
%\end{definition}

\begin{definition}[Solution of a model of commuting population]
\label{de3}
Let $\cP$ be a finite set of subpopulations, and $\cR(\cdot)$ an admissible switching signal.
%define a {\em switching signal} $\cR(\cdot)$ as a function which associates to any $t\geq 0$ an equivalence relation on $\cP$, denoted $\cR(t)$, {\color{bleu} such that $\cR(\cdot)$ has finite number of switches on any bounded set of $\Rset_+$.}
%For any $p\in\cP$, the corresponding equivalence of $p$ in the quotient set $\cP\backslash\cR$ is denoted $[p]$.
We call {\em $(\cP,\cR(\cdot))$-solution of model \eqref{eq0}} any continuous function $\cX\ :\ \Rset_+ \to \Rset_+^{n|\cP|}$, $t\mapsto (x_1(t),\dots, x_{|\cP|}(t))$, continuously differentiable on $\cI$ that fulfils \eqref{eq2}.
\end{definition}
In \eqref{eq2}, $[p]_{\cR(t)}$ denotes for any subpopulation $p\in\cP$ the equivalence class to which $p$ belongs at time $t$, that is the `place' where it is located, together with the other members of the same equivalence class $[p]_{\cR(t)}$.
The difference between \eqref{eq1} and \eqref{eq2} is only in the fact that the partition $\cR$ in \eqref{eq2} is now a function of time.

Due to the non-accumulation of the number of switches (which comes from the admissibility of $\cR$), it is easy to show that there exists a unique $(\cP,\cR(\cdot))$-solution of model \eqref{eq0} passing through a given initial condition $\cX(0)\in\Rset_+^{n|\cP|}$.

If $f$ is linear with respect to its first argument, writing together the evolution of the populations within each class yields the following equation, which is the analogue of \eqref{eq3} for system \eqref{eq2}:
%\begin{subequations}
%\label{eq88}
\begin{equation*}
%\label{eq8}
\dot x_q = f(x_q(t),\lambda_q(t)),\qquad \lambda_q(t) = \lambda_q \left(
x_q (t)
\right),\qquad t\in\cI,\quad
q\in \cQ(t).
\end{equation*}
Here by definition we put
\begin{equation*}
\cQ(t) := \cP\backslash\cR(t),\qquad t\in \cI.
\end{equation*}
%\end{subequations}

The considerations on simulation issues made at the end of Section~\ref{se11} apply directly to system~\eqref{eq2}, due to the continuity of the state variable at the switching instants.

%\addtocounter{example}{-1}
%\begin{example}[continued]
\begin{example}[continuation of Example~\ref{ex1}]
\label{ex2}

The example of partitioned system shown in Example {\em\ref{ex1}} extends straightforwardly to a model of commuting population, using the definition of $f$ in \eqref{eq5}, and defining, for a given admissible switching signal, the force of infection at every (`location') $q\in\cP\backslash\cR(t)$ at any given time $t\in\cI$ by
\begin{equation}
\label{eq9}
\lambda_q(x) := \beta_q \frac{I}{S+I+R},\qquad t\in\cI,\quad q \in \cQ(t).
%\cP\backslash\cR(t).
\end{equation}

\end{example}

\section{Epidemiological dynamics}
\label{se2}

So far, we have introduced a description of the commutations and `mixing' of different subpopulations.
In the present section, we will first focus in Section~\ref{se50} on the representation of the different steps of the epidemiological processes (infection, exposition, recovery, immunity\dots), following the framework introduced by van den Driessche and Watmough~\cite{Driessche:2002aa}, further extended by Wang and Zhao~\cite{Wang:2008aa} to periodic problems.
As an illustration and in order to show the interest of the obtained class of models, we show in Sections~\ref{se51} to~\ref{se55} that it contains in particular the five relevant examples studied in~\cite{Driessche:2002aa}.

\subsection{A class of compartmental disease transmission models}
\label{se50}

To adapt the framework of~\cite{Driessche:2002aa}, a first formality is to introduce a distinction between the `infected' and `non-infected' compartments (respectively with exponents $I$ and $U$) and
%, in accordance with~ \cite{Driessche:2002aa},
to separate accordingly the terms in the equations which correspond to infections.
More precisely we write the local model \eqref{eq0} under the more precise form:
\begin{subequations}
\label{eq3017}
\begin{gather}
\label{eq30b}
x :=\begin{pmatrix} x^I \\ x^U  \end{pmatrix},\qquad
f(x,\lambda) := \begin{pmatrix} \cF^I(x,\lambda) + \cV^I(x,\lambda) \\ \cV^U(x,\lambda) \end{pmatrix},\\
n_I := \dim x^I,\qquad n_U := \dim x^U,\qquad n_I+n_U = n.
%f(x,\lambda) :=\begin{pmatrix} f^I (x,\lambda)\\ f^U (x,\lambda) \end{pmatrix},
%\cF(x,\lambda) :=\begin{pmatrix} \cF^I (x,\lambda)\\ \cF^U (x,\lambda), \end{pmatrix}.
\end{gather}
%with
%\begin{equation}
%\label{eq45}
%f^I(x,\lambda)=\cF^I(x,\lambda) + \cV^I(x,\lambda).
%\end{equation}
%% first write the local model \eqref{eq0} in a more precise form:
%\label{eq30}
%\begin{equation}
%\label{eq30a}
%\dot x^I = f^I(x,\lambda),\qquad
%\dot x^U = f^U(x,\lambda),\qquad \lambda = F(x).
%\end{equation}
%in order to benefit both from the structure developed at the beginning and from that of the framework of van den Driessche and Watmough.
%For the systems \eqref{eq2} and \eqref{eq11}, the functions $f^I, f^U$ are periodic in time {\em through the commutations appearing in $F$}, and we can use the computational framework of $\cR_0$ from {\em Threshold Dynamics for Compartmental Epidemic Models in Periodic Environments}, by Wang and Zhao (2008)).
%In all the sequel we will write
The components $\cF^I(x,\lambda), \cV^I(x,\lambda), \cV^U(x,\lambda)$
are chosen according to the prescriptions made by van den Driessche and Watmough\footnote{The right-hand side decomposition is denoted $f(x) = \cF(x) - \cV(x)$ in~\cite{Driessche:2002aa}. For simplicity, we choose here the opposite sign for the term $\cV$.
This has no major impact, including on the subsequent definition of the basic reproduction number in~\eqref{eq355a}, as the spectral radius of a matrix is an even function.}:
$\cF^I(x,\lambda)$ gathers solely the rates of appearance of new infections in infectious compartments, and $\cV^I(x,\lambda)$ the rates of transfer of individuals in and out of the infected compartments by any other means; while $\cV^U(x,\lambda)$ expresses the rates of transfer of individuals in and out of the non-infected compartments.
%All components $\cF^I(x,\lambda), \cV^I_{+}(x,\lambda), \cV^I_{-}(x,\lambda)$ take on nonnegative values.

In fact, more structure is present, as will be demonstrated in Sections~\ref{se51} to~\ref{se55} by the study of examples.
%{\PAB A key feature of the classical models is that the nonlinearity of the system enters either through the force of infection $\lambda(x)$, or from some nonlinear combination in the infection function $\cF^I(x,\lambda)$.
%Moreover, in view of the models used in the literature, a realistic assumption is that the nonlinear terms present in the right-hand side, which are consequences of the transfers from a compartment to the other caused by the infection, are also {\em\color{red} linear upon} $\lambda$, and consequently bilinear in $(x,\lambda)$.
%Overall the function $f$ will thus be chosen {\em affine with respect to} $\lambda$.
A key feature of most classical models is that the force of infection, which measures the per capita probability of acquiring the infection~\cite{Anderson:1991aa}, enters {\em linearly in the infection terms}; while the other terms present in the equations are linear with respect to the numbers of individuals in some compartments.
In line with this observation, we will assume that the function $f(x,\lambda)$ is {\em linear} with respect to $x$ and {\em affine} with respect to $\lambda$, and therefore that $\cF^I(x,\lambda)$, $\cV^I(x,\lambda)$ and $\cV^U(x,\lambda)$ share the same properties.
%{\color{red} Given the objective of our study and in view of the models existing for this in the literature, we assume that  the function $f(x,\lambda)$ is {\em linear} with respect to $x$.}
%Most importantly, the terms $\cV^I(x,\lambda)$ do {\em not} depend upon the infection force $\lambda$: all infection terms are counted in $\cF^I$.
%In line with the framework built in Section \ref{se1},
%we assume that all the components of $\cF^I(x,\lambda)$, $\cV^I(x,\lambda)$, $\cV^U(x,\lambda)$ are {\em\color{red} linear in} $x$.
In particular, the nonlinear character of the system arises from the force of infection $\lambda(x)$, and from  cross products between components of $\lambda(x)$ and $x$.

%nonlinear combination in the infection function $\cF^I(x,\lambda)$.
Due to the fact that infections appear only in presence of infected, we assume that the function $\cF^I(x,\lambda)$ which bears the corresponding terms is {\em linear upon} $\lambda$ (and not simply affine), and consequently bilinear in $(x,\lambda)$.
The terms contained in $\cF^I$ corresponds to the appearance of new infected individuals, therefore they always take on nonnegative values.
They are compensated by the disparition of other (non-infected, but also infected in case of super-infections) individuals in corresponding compartments.
%Moreover, in view of the models used in the literature, a realistic assumption is that the nonlinear terms present in the right-hand side, which are consequences of the transfers from a compartment to the other caused by the infection, are also {\em\color{red} linear upon} $\lambda$, and consequently bilinear in $(x,\lambda)$.
%Overall the function $f$ will thus be chosen {\em affine with respect to} $\lambda$.
We also assume that the term $\cV^I(x,\lambda)$ is linear in $x^I$ only, meaning that it does not depend upon $x^U$.
This formalizes the fact that every transfer to an infected compartments coming from a non-infected one, occurs through an infection process and therefore is counted within $\cF^I(x,\lambda)$.
%Last, we assume that the functions $\cV^I(x,\lambda), \cV^U(x,\lambda)$ are {\em affine} functions of $\lambda$, meaning that the terms that do depend upon the force of infection are in fact linear with respect to the latter.

Consequently we denote more specifically
\begin{equation}
\label{eq17}
\cF^I(x,\lambda) := F^I_+ (\lambda) x^I + F_+ (\lambda) x^U,\quad
\cV^I(x,\lambda) := V^I x^I -  F^I_- (\lambda) x^I,\quad
\cV^U(x,\lambda) := V^U x - F_- (\lambda) x^U
\end{equation}
\end{subequations}
where $V^I \in \Rset^{n_I\times n_I}$, $V^U \in \Rset^{n_U\times n}$, and where the functions $F^I_\pm (\cdot) \ :\ \Rset^{n_\lambda} \to \Rset^{n_I\times n_I}$, $F_+ (\cdot) \ :\ \Rset^{n_\lambda} \to \Rset^{n_I\times n_U}$, $F_- (\cdot) \ :\ \Rset^{n_\lambda} \to \Rset^{n_U\times n_U}$
%\[
%V^I \in \Rset^{n_I\times n_I},\qquad 
%V^U \in \Rset^{n_U\times n},
%\]
%and the functions
%\[
%F^I_\pm (\lambda) \ :\ \Rset_+^p \to \Rset_+^{n_I\times n_I},\quad F_+ (\lambda) \ :\ \Rset_+^p \to \Rset_+^{n_I\times n_U}, \quad F_- (\lambda) \ :\ \Rset_+^p \to \Rset_+^{n_U\times n_U}
%\]
are {\em linear}.
In~\eqref{eq17}, the new infections are decomposed into infection of susceptible individuals and infection of already infected ones, respectively by the nonnegative terms $F_+ (\lambda) x^U$ and $F^I_+ (\lambda) x^I$.
They result in the {\em remotion} of corresponding infected individuals.
This appears through the nonpositive term $- F^I_- (\lambda) x^I$ in $\cV^I(x,\lambda)$ in case of infection of an already infected, and through the nonpositive term $- F_- (\lambda) x^U$ in $\cV^U(x,\lambda)$ in case of infection of a non-infected.
Notice also that these two matrices have to be {\em diagonal}, while in absence of self-infections, the diagonal of
%the terms $F_+ (\lambda) x^U$ and $F^I_+ (\lambda) x^I$ have to be zero.
$F^I_+ (\lambda)$ has to be zero.
%Overall, the infections are responsible for the intakes in the $\cF^I$ map and the outtakes in the $\cV^I$ and $\cV^U$ maps.

Denoting
\begin{equation}
\label{eq781}
F_0(\cdot) := F^I_+ (\cdot) - F^I_- (\cdot),
\end{equation}
one is finally led to adopt the general form
\begin{equation}
\label{eq41}
f(x,\lambda) := \begin{pmatrix} (F_0 (\lambda) + V^I) x^I + F_+ (\lambda) x^U\\ V^U x - F_- (\lambda) x^U \end{pmatrix}.
\end{equation}
This decomposition is indeed the most general form that integrates {\em linear dependence with respect to $x$ and affine dependence with respect to $\lambda$} within the framework based on the {\em distinction of the new infection terms} introduced in~\cite{Driessche:2002aa}.
To summarize the comments made below, it is reasonable to assume that $F_+(\cdot), F_-(\cdot)$ take on nonnegative values with $F_-(\cdot)$ is diagonal, and that $F_0(\cdot)$ is a {\em Metzler} matrix: this will be the content of Assumption~\ref{as1} below.

%In accordance with what was explained in the preceding paragraph, the linear functions $F_\pm$ map the domain positive cone to the codomain positive cone: the infections result in {\em intakes} in the $\cF^I$ map and {\em outtakes} in the $\cV^I$ and $\cV^U$ maps.}

\begin{subequations}
\label{eq18}
Compartmental models describe transfers of individuals between different compartments.
Except for the processes of birth and death which may add or subtract individuals, the corresponding transitions are essentially lossless.
As a consequence, mass balance equations are usually fulfilled in absence of specific, disease-induced, mortality, namely:
\begin{gather}
\label{eq18a}
\bfo_{n_I}\t F_0 (\cdot) = 0_{n_I}\t,\qquad
\bfo_{n_I}\t F_+ (\cdot) = \bfo_{n_U}\t F_-(\cdot),\\
\label{eq18b}
\bfo_{n_I}\t V^I + \bfo_{n_U}\t V^U  \begin{pmatrix}
I_{n_I} \\ 0_{n_U \times n_I}
\end{pmatrix}
= 0_{1\times n_I},\\
\label{eq18d}
\bfo_{n_U}\t V^U \begin{pmatrix}
0_{n_I \times n_U} \\ I_{n_U}
\end{pmatrix}
= 0_{1\times n_U}.
%\begin{pmatrix} \bfo_{n_I}\t V^I & 0_{1\times n_U} \end{pmatrix} + \bfo_{n_U}\t V^U = 0_{1\times n}.
\end{gather}
These formulas express the conservation of individuals during the transfers from one compartment to others during infections, for any value of the force of infection $\lambda$ (formula~\eqref{eq18a}); and during the non-infective transfers: formula~\eqref{eq18b} concerns the transfers of infected individuals
%\footnote{\PAB More precisely: of individuals infected before the beginning of the transfer.}
and~\eqref{eq18d} the transfers of non-infected individuals.
As a complementary explanation, consider the total population for system \eqref{eq0}-\eqref{eq3017}, which is equal to $\bfo_n\t x = \bfo_{n_I}\t x^I + \bfo_{n_U}\t x^U$.
The derivative of this quantity with respect to time is
\begin{eqnarray*}
\bfo_{n_I}\t \dot x^I + \bfo_{n_U}\t \dot x^U
& = &
\bfo_{n_I}\t \left(
F_0 (\lambda) + V^I
\right)x^I + (\bfo_{n_I}\t F_+ (\lambda)- \bfo_{n_U}\t F_- (\lambda)) x^U + \bfo_{n_U}\t V^U x\\
& = &
\left[
\bfo_{n_I}\t \left(
F_0 (\lambda) + V^I
\right)
+ \bfo_{n_U}\t V^U  \begin{pmatrix}
I_{n_I} \\ 0_{n_U \times n_I}
\end{pmatrix}
\right]x^I\\
& &
+ \left[
\bfo_{n_I}\t F_+ (\lambda)- \bfo_{n_U}\t F_- (\lambda)
+ \bfo_{n_U}\t V^U \begin{pmatrix}
0_{n_I \times n_U} \\ I_{n_U}
\end{pmatrix}
\right] x^U,
\end{eqnarray*}
so that $\bfo_{n_I}\t \dot x^I + \bfo_{n_U}\t \dot x^U\equiv 0_n$ when all identities in \eqref{eq18a}, \eqref{eq18b}, \eqref{eq18d} are fulfilled, and the total population is then constant along the evolution.

In presence of specific, disease-induced, mortality, these formulas have to be modified to account for this effect, which is limited to the infected compartments.
Equations \eqref{eq18a}, \eqref{eq18d} remain valid, while \eqref{eq18b} has to be replaced by
\begin{equation}
\label{eq18c}
\bfo_{n_I}\t V^I + \bfo_{n_U}\t V^U  \begin{pmatrix}
I_{n_I} \\ 0_{n_U \times n_I}
\end{pmatrix}
<\ 0_{1\times n_I}.
%\qquad
%\bfo_{n_U}\t V^U \begin{pmatrix}
%0_{n_I \times n_U} \\ I_{n_U}
%\end{pmatrix}
%= 0_{1\times n_U}.
\end{equation}
\end{subequations}
Recall that in the previous formula, $<$ means `less or equal, but not equal'.
As a matter of fact, the previous computation shows that one has in this case
\[
\bfo_{n_I}\t \dot x^I + \bfo_{n_U}\t \dot x^U
=
\left(
\bfo_{n_I}\t V^I
+ \bfo_{n_U}\t V^U  \begin{pmatrix}
I_{n_I} \\ 0_{n_U \times n_I}
\end{pmatrix}
\right) x^I \leq 0.
\]
In any case, \eqref{eq18d} alone ensures that any {\em disease-free} evolution occurs under constant total population.

The inequalities presented above are included below in Assumption~\ref{as3}.
Other properties are needed, which are quite natural in this setting.
First, the square matrices $V^I$ and $V^U \begin{pmatrix} 0_{n_I \times n_U} \\ I_{n_U} \end{pmatrix}$ should be Metzler matrices, as they represent intakes from other compartments for the off diagonal terms, and outtakes to other compartments for the diagonal ones.
The latter matrix accounts for the internal exchanges between non-infected compartments, governed by the equation $\dot x^U = V^U \begin{pmatrix} 0_{n_I} \\ x^U \end{pmatrix}$ in absence of infected.
As the balance of these exchanges should be zero in order to sustain constant population levels, it is thus natural to assume that the matrix $V^U \begin{pmatrix} 0_{n_I \times n_U} \\ I_{n_U} \end{pmatrix}$ admits $0$ as an eigenvalue, associated to the left eigenvector $\bfo_{n_U}\t$.

Also, every non-infected compartment is assumed to receive intakes from some infected compartments, either `after infection' or `before' in the shape of births proportional to this compartment: otherwise it is not needed for the description of the disease spread.
For this it is added in Assumption~\ref{as3} that $\bfo_{n_U}\t V^U  \begin{pmatrix}
I_{n_I} \\ 0_{n_U \times n_I}
\end{pmatrix} \succ 0_{1\times n_I}$.
Recall that the symbol $\succ$ has been defined in the paragraph Notations at the end of Section~\ref{se0}: $0 \prec x\in\Rset^n$ means $0 < x_i$, $i=1,\dots,n$.\\

%This being said, these formulas do not seem to be necessary for our derivations.\\

In order to exemplify the ideas behind the previous modelling choices, we first come back to the SIR model already commented.
%Before going further, we exemplify the ideas behind the previous modelling choices, based on the study of five epidemiological models extracted from the celebrated work by van den Driessche and Watmough \cite{Driessche:2002aa}.
Afterwards, we study in details in Sections~\ref{se51} to~\ref{se55} the five epidemiological models shown as examples in~\cite[Section 4]{Driessche:2002aa}, with some slight adaptations.
We then come back in Section~\ref{se3} to the formal presentation of the proposed framework.

\begin{remark}
We find in~{\rm\cite[equation (1)]{Bonzi:2011aa}} a general class of models with differential susceptibility and infectivity, and in {\rm\cite[Section 23.1]{Thieme:2003aa}} a general class of models with multiple groups or populations (including e.g.~multi-stage infection).
Both show evolution governed by ODEs that are at the same time linear with respect to state variables, and affine with respect to forces of infection (which are themselves nonlinear functions of the state).
With some slight modifications (in particular, making the recruitment terms proportional to the population), both classes can be embedded in the framework~\eqref{eq3017} by adequate choice of the $x$ variable.
\end{remark}

%\addtocounter{example}{-1}
\begin{example}[end of Examples~\ref{ex1} and \ref{ex2}]
% and~\ref{ex3}]
\label{ex4}
The model~\eqref{eq5} is embedded in the framework~\eqref{eq3017} with $x^I := I$, $x^U := \begin{pmatrix} S & R \end{pmatrix}\t$, so that $n=3$, $n_I=1$, $n_U=2$, and one sees that $n_\lambda=1$.
The function $f$ defined in~\eqref{eq5} is decomposed as in~\eqref{eq30b} in its two, infected and non-infected, components, yielding
\[
\cF^I(x,\lambda):= \lambda S,\qquad
\cV^I(x,\lambda):= -(\gamma+\mu)I,\qquad
\cV^U(x,\lambda) := \begin{pmatrix}
\mu (I+R) - \lambda S\\
\gamma I - \mu R
\end{pmatrix},
\]
and the decomposition~\eqref{eq17} with the following elements:
\begin{subequations}
\label{eq613}
\begin{gather}
\label{eq613a}
F_0 (\lambda) := 0,\qquad
F_+ (\lambda) := \lambda\begin{pmatrix} 1 & 0 \end{pmatrix},\qquad
V^I := -(\gamma+\mu),\\
\label{eq613b}
V^U := \begin{pmatrix} \mu & 0 & \mu \\ \gamma & 0 & -\mu \end{pmatrix},\qquad
F_- (\lambda) := \lambda \begin{pmatrix} 1 & 0 \\ 0 & 0\end{pmatrix}.
\end{gather}
One sees easily that the identities in~\eqref{eq18a} and~\eqref{eq18d} are verified, and as well~\eqref{eq18b}, because
\[
\bfo_{n_I}\t V^I + \bfo_{n_U}\t V^U  \begin{pmatrix}
I_{n_I} \\ 0_{n_U \times n_I}
\end{pmatrix}
%= V^I + \bfo_2\t V^U  \begin{pmatrix}
%1 \\ 0_{2\times 1}
%\end{pmatrix}
= -(\gamma+\mu) + \begin{pmatrix}
1 & 1
\end{pmatrix}
\begin{pmatrix}
\mu \\ \gamma
%0_{2\times 1}
\end{pmatrix}
= 0,
\]
\end{subequations}

\end{example}

\subsection{Treatment model}
\label{se51}

The first example considered in~\cite{Driessche:2002aa} is a tuberculosis model with treatment proposed in~\cite{Blower:1996aa,Castillo-Chavez:1997aa}.
It possesses four compartments, namely, individuals susceptible to tuberculosis (S), exposed individuals (E), infectious individuals (I) and treated individuals (T).
The forces of infection applied to susceptible and treated individuals have different rates, respectively $\beta_1\dps\frac{I}{N}$ and $\beta_2\dps\frac{I}{N}$, where $N := E + I + S + T$.
The exposed individuals enter the infectious compartment at a rate $\nu$.
All newborns are susceptible, and the mortality is uniform and equal to $d > 0$.
The treatment rates are denoted $r_1$, resp.~$r_2$, for the exposed individuals, resp.~the infectious individuals.
A fraction $q$ of the treatments of infectious individuals are successful, and the unsuccessfully treated infectious individuals reenter the exposed compartment.
The epidemiological model is given by the following system.
\begin{equation}\label{W1}
\left \lbrace\begin{array}{lll}
\dot{E}&=& \beta_1 S\dps\frac{I}{N}+\beta_2 T\dps\frac{I}{N}-(d+\nu+r_1)E+(1-q)r_2I,\\
\dot{I}&= & \nu E-(d+r_2)I,\\
\dot{S}&= & b(N) -dS-\beta_1 S\dps\frac{I}{N},\\
\dot{T}&= & -dT + r_1 E+qr_2I-\beta_2T\dps\frac{I}{N}.
\end{array} \right.
\end{equation}

In the case where the natality term is linear $b(N)=bN$,
% is linear and equal to $dN$ (i.e.~the natality and mortality rates are equal),
this decomposes as in \eqref{eq41},
%\eqref{eq3017},
with state dimension $n=4$ and the following dynamics.
\begin{subequations}
\label{eq25}
\begin{gather}
\label{eq25e}
%x :=\begin{pmatrix} x^I \\ x^U  \end{pmatrix}, \qquad
x^I:=\begin{pmatrix} E \\ I  \end{pmatrix}, \qquad
x^U:=\begin{pmatrix} S \\ T  \end{pmatrix}, \qquad
n_I = n_U := 2,\qquad
\lambda := \begin{pmatrix} \lambda_1 \\ \lambda_2  \end{pmatrix},\qquad
n_\lambda :=2,\\
\label{eq25a}
F_0(\cdot) \equiv 0_{2\times 2},\qquad
F_+(\lambda) = F_-(\lambda) := \lambda_1 \begin{pmatrix}
1 & 0  \\ 0 &  0  \end{pmatrix}+ \lambda_2 \begin{pmatrix}
0& 0 \\ 0 &  1  \end{pmatrix},\\
%F_+(\lambda) = F_-(\lambda) := \lambda_1 \begin{pmatrix}
%1 & 0  \\ 0 &  0  \end{pmatrix}x^U+ \lambda_2 \begin{pmatrix}
%0& 0 \\ 0 &  1  \end{pmatrix}x^U,\\
%
%\cF^I(x,\lambda)
%:= \lambda_1 \begin{pmatrix}
%1 & 0  \\ 0 &  0  \end{pmatrix}x^U+ \lambda_2 \begin{pmatrix}
%0& 0 \\ 0 &  1  \end{pmatrix}x^U,\qquad
%\cV^I(x,\lambda)
%:= V^Ix^I,\\
%\label{eq25b}
%\cV^U(x,\lambda)
%:= V^Ux
%- \lambda_1 \begin{pmatrix}
%1 & 0  \\ 0 &  0  \end{pmatrix}x^U- \lambda_2 \begin{pmatrix}
%0& 0 \\ 0 &  1  \end{pmatrix}x^U,\\
\label{eq25c}
V^I :=
\begin{pmatrix}
-(d+\nu+r_1) & (1-q)r_2  \\  \nu &  -(d+r_2) \end{pmatrix}, \qquad
V^U
:=\begin{pmatrix}
b & b & b-d & b \\ r_1 & qr_2 & 0 & -d \end{pmatrix},\\
\label{eq25d}
\lambda_1(x)=\beta_1 \dps\frac{I}{E+I+S+T},\qquad \lambda_2(x) =\beta_2 \dps\frac{I}{E+I+S+T}.
\end{gather}
\end{subequations}
One checks that $V^I$ is a nonsingular Metzler matrix, %that $F^I_+$ and $F^I_-$ are null
that the identity in \eqref{eq18a} is fulfilled and that~\eqref{eq18d} holds when  the natality and mortality rates $b$ and $d$ are equal.
Whenever $b=d$,~\eqref{eq18b} (resp.~\eqref{eq18c})  is fulfilled when the specific mortality $d$ is zero (resp.~positive).

%Here, $\cF^I(x,\lambda)$ and $\cV^I(x,\lambda)$ can be written as :
%\[
%\cF^I(x,\lambda)= \lambda_1 \begin{pmatrix}
%1 & 0  \\ 0 &  0  \end{pmatrix}x^U+ \lambda_2 \begin{pmatrix}
%0& 0 \\ 0 &  1  \end{pmatrix}x^U,\qquad
%\cV^I(x,\lambda)=-V^Ix^I,
%\]
%where
% \[
%\cV^I_{-}(x,\lambda)= \begin{pmatrix}
%(d+\nu+r_1)E \\ (d+r_2) I \end{pmatrix} = 
%\begin{pmatrix}
%d+\nu+r_1 & 0  \\ 0 & d+r_2 \end{pmatrix}
%x^I,\quad
%\cV^I_{+}(x,\lambda)= \begin{pmatrix}
% (1-q)r_2 I \\ \nu E  \end{pmatrix}
% = \begin{pmatrix}
%0 & (1-q)r_2  \\ \nu & 0 \end{pmatrix} x^I
%\]

%The infected compartments are $E$ and $I$.
%Note 
%The forces of infection are
%\[
%\lambda_1(x)=\beta_1 \dps\frac{I}{N} \quad \text{ and}  \quad \lambda_2(x)=\beta_2 \dps\frac{I}{N}.\]
%Let us note
%\[
%\lambda(x)= \begin{pmatrix}
%\lambda_1(x) \\\lambda_2(x) \end{pmatrix}
%\]
%The model \eqref{W1} can be written as follow : 
%
%\begin{equation}\label{W2}
%\dot x = f(x,\lambda(x)),\ \lambda(x) = F(x)\Leftrightarrow\left \lbrace\begin{array}{lll}
%\dot{x}^I&=& \begin{pmatrix}
%1 & 0  \\ 0 &  0  \end{pmatrix}x^U.\lambda_1(x)+\begin{pmatrix}
%0 & 1  \\ 0 &  0  \end{pmatrix}x^U.\lambda_2(x)+V^I.x,\\
%\dot{x}^U&= &V^U.x-\begin{pmatrix}
%1 & 0  \\ 0 &  0  \end{pmatrix}x^U.\lambda_1(x)-\begin{pmatrix}
%0 & 0  \\ 0 &  1  \end{pmatrix}x^U.\lambda_2(x)+b(N).e_1.
%\end{array} \right.
%\end{equation}  where 
%\begin{equation*}
%V^I =\begin{pmatrix}
%-(d+\nu+r_1) & pr_2  & 0  & 0 \\ \nu & -(d+r_2)& 0 & 0 \end{pmatrix},\quad V^U =\begin{pmatrix}
%0 & 0  & -d  & 0 \\ r_1 & qr_2 & 0 & -d \end{pmatrix} ,\quad
%e_1 = \begin{pmatrix}
%1 \\ 0 \end{pmatrix}
%\end{equation*}

\subsection{Multigroup model}
\label{se52}

The second example in~\cite{Driessche:2002aa} is an $m$-group SIRS-vaccination model of Hethcote~\cite{Hethcote:1987aa,Hethcote:1978aa}.
The model is slightly adapted here.
Using numbers of individuals instead of proportions, the model equations are as follows :
\begin{equation}\label{W01}
\left \lbrace\begin{array}{lll}
\dot{I_i}&=& \dps\sum_{j=1}^{m}\beta_{ij}(x) S_i \frac{I_j}{I_j+S_j+R_j}-(d_i+\gamma_i+\epsilon_i)I_i,\\
\dot{S_i}&= & (1-p_i)b_i (I_i+S_i+R_i) -(d_i+\theta_i)S_i+\sigma_i R_i-\dps\sum_{j=1}^{m}\beta_{ij}(x) S_i \frac{I_j}{I_j+S_j+R_j},\\
\dot{R_i}&= & p_i b_i (I_i+S_i+R_i) +\gamma_i I_i+\theta_i S_i-(d_i+\sigma_i)R_i,\\
\end{array} \right.
\end{equation}
for $i=1,\dots,m$ where $x=\begin{pmatrix}
I_1 & \dots & I_m & S_1 & \dots & S_m & R_1 & \dots & R_m
\end{pmatrix}\t$.
The natural death rate is $d_i$, and $\epsilon_i$ the specific death rate.
Recovering of the infected individuals occurs at rate $\gamma_i$, and immunity wanes at a rate $\sigma_i$.
All newborns are susceptible, with group $i$ having a birth rate equal to $b_i$.
 Here, the dimension of the state variable is $n = 3m$, where $m$ is the number of distinct groups.
 %; and $b_i = d_i(I_i+S_i+R_i)$.
Vaccination of the newborns, resp.~of the susceptibles, is achieved with a rate $p_i$, resp.~$\theta_i$.
The incidence rates $\beta_{ij}(x)$ depend on individual behaviour and describes the amount of mixing between groups~\cite{Jacquez:1988aa}.
%The infected compartments are $I_1,\dots,I_m$.

The decomposition \eqref{eq41}
%\eqref{eq3017}
is obtained here with
\begin{subequations}
	\label{eq26}
	\begin{gather}
	\label{eq26e}
	x^I:=
	\begin{pmatrix} I_1 & \dots & I_m  \end{pmatrix}\t, \quad
	x^U:=
	\begin{pmatrix} S_1 & \dots & S_m & R_1 & \dots \dots R_m  \end{pmatrix}\t,\quad
	n_I = n_\lambda := m,\ n_U:=2m,\\
	\label{eq26a}
F_0(\cdot) \equiv 0_{m\times m},\qquad
F_+(\lambda) := \begin{pmatrix} \Lambda & 0_{m\times m} \end{pmatrix},\qquad
F_-(\lambda) := \diag \{ \Lambda; 0_{m\times m} \},\\
%	\cF^I(x,\lambda)
%	:= \begin{pmatrix} \Lambda & 0_{m\times m} \end{pmatrix} x^U,\qquad
%	\cV^I(x,\lambda) := V^Ix^I, \\
%%	\lambda_1 B_1 x^U
%%	+\dots+ \lambda_m B_m x^U,\\
%	\label{eq26b}
%	\cV^U(x,\lambda)
%	:= V^Ux
%	- \diag \{ \Lambda; 0_{m\times m} \} x^U,\\
%	-\lambda_1 C_1 x^U-\dots
%	-\lambda_m C_m x^U,\\
\label{eq26c}
V^I := \diag\{-(d_i+\gamma_i+\epsilon_i) \},\\
\label{eq26g}
V^U := \begin{pmatrix}
\diag\{(1-p_i)b_i\} &  \diag\{(1-p_i)b_i-d_i-\theta_i\} & \diag\{(1-p_i)b_i +\sigma_i\}  \\
\diag\{p_ib_i+\gamma_i\} & \diag\{p_ib_i+\theta_i\} & \diag\{p_ib_i-d_i -\sigma_i\} 
\end{pmatrix},\\
%	\\
\label{eq26d}
\Lambda(x) := \diag\{ \lambda_i (x)\} \in\Rset^{m\times m},\qquad
\lambda_i(x) := \sum_{j=1}^m \beta_{ij}(x) \frac{I_j}{I_j+S_j+R_j},\  i=1,\dots,m.
\end{gather}
\end{subequations}

Again, $V^I$ is a nonsingular Metzler matrix.
The two identities in \eqref{eq18a} are clearly fulfilled.
On the other hand, here
\begin{gather*}
\bfo_{n_I}\t V^I + \bfo_{n_U}\t V^U  \begin{pmatrix}
I_{n_I} \\ 0_{n_U \times n_I}
\end{pmatrix}
= \begin{pmatrix} b_1-d_1- \epsilon_1 & \dots & b_m-d_m - \epsilon_m \end{pmatrix},\\
\bfo_{n_U}\t V^U \begin{pmatrix}
0_{n_I \times n_U} \\ I_{n_U}
\end{pmatrix}
= \begin{pmatrix} b_1-d_1 & \dots & b_m-d_m & b_1-d_1 & \dots & b_m-d_m \end{pmatrix},
\end{gather*}
so that \eqref{eq18d} holds when $b_i = d_i, i=1,\dots,m$; and \eqref{eq18b}, resp.~\eqref{eq18c} also holds when $\epsilon_i=0, i=1,\dots,m$, resp.~$\begin{pmatrix} \epsilon_1 & \dots & \epsilon_m \end{pmatrix} > \begin{pmatrix} 0 & \dots & 0 \end{pmatrix}$.
%, with $F^I_+$ and $F^I_-$ equal to zero.

\comment{
Hence, we have 
\[\cF^I(x,\lambda)=  \begin{pmatrix} \Lambda & 0_{m\times m} \end{pmatrix} x^U\]

\[\cV^I(x,\lambda)=-V^Ix^I,\] where 
\[
\cV^I_{-}(x,\lambda)
	= -V^I x^I,\quad
	\cV^I_{+}(x,\lambda)=0_m.
	\]
}

\subsection{Staged progression model}
\label{se53}

The third example shown in~\cite{Driessche:2002aa} has a single susceptible compartment, but the infected individuals pass through $m$ different stages of disease~\cite{Hyman:1999aa}.
The model has dimension $m+1$, with the following equations.
\begin{equation*}
%\label{W03}
\left \lbrace
\begin{array}{lll}
\dot{I}_1
&=&
S\dps\sum_{k=1}^{m-1}\dps\frac{\beta_{k} I_k}{N}-(\nu_1+d_1)I_1,\\
\dot{I}_i
&= &
\nu_{i-1}I_{i-1}-(\nu_i+d_i)I_i,\quad i=2,\dots,m-1,\\
\dot{I}_m
&= &
\nu_{m-1}I_{m-1}-d_mI_m,\\
\dot{S}
&= &
bN -bS-S\dps\sum_{k=1}^{m-1}\dps\frac{\beta_{k} I_k}{N},
\end{array} \right.
\end{equation*}
where the total population is $N:=S+I_1+\dots+I_m$.

Denoting $e_1$ the first vector of the canonical basis of $\Rset^m$,
% :=  \begin{pmatrix}
%1 & 0 & \dots & 0
%\end{pmatrix}\t$,
this model writes as \eqref{eq41}
%\eqref{eq3017}
with
\begin{subequations}
\begin{gather}
x^I:=\begin{pmatrix} I_1 & \dots & I_m  \end{pmatrix}\t, \qquad
 x^U=S, \qquad
 n_I = n_\lambda := m,\ n_U := 1,\\
% x=\begin{pmatrix} x^I \\ x^U  \end{pmatrix}, \quad
F_0(\cdot) \equiv 0,\qquad
F_+(\lambda) := \lambda e_1,\qquad F_-(\lambda) := \lambda,\\
%\cF^I (x,\lambda):=
% \lambda x^U e_1,\qquad
%\cV^I (x,\lambda):=
%V^Ix^I ,\qquad
%\cV^U (x,\lambda) :=  V^Ux-\lambda x^U,\\
V^I := \begin{pmatrix}
-(\nu_1+d_1) & 0 & \cdots &\cdots & 0 \\
\nu_1 & -(\nu_2+d_2)  & 0 &  \cdots & 0 \\
0 &  \nu_2 & \ddots &  &  \vdots\\
\vdots & \vdots &  & -(\nu_{m-1}+d_{m-1}) & 0\\
0 &  0  & \dots  &  \nu_{m-1} &  -d_m
\end{pmatrix},\qquad 
V^U
:= \begin{pmatrix} b \bfo_m\t & 0 \end{pmatrix},\\
\lambda(x):=\sum_{k=1}^{m-1}\dps\frac{\beta_{k} I_k}{N}.
\end{gather}
\end{subequations}
%\[
%D := \begin{pmatrix}
%1 & 0 & \cdots & 0  \\
%0 &  0 &   \cdots & 0   \\
%\vdots &  \vdots & \ddots &  \vdots \\
%0 & 0 & \cdots &  0 
%\end{pmatrix}, 
%%\quad f(x,\lambda) :=\begin{pmatrix} f^I (x,\lambda)\\ f^U (x,\lambda) \end{pmatrix}
%\]
The matrix $V^I$ is a nonsingular Metzler matrix.
Identities \eqref{eq18a} and \eqref{eq18d} are fulfilled, and \eqref{eq18b} as well when $d_1=\dots=d_m=b$ (i.e.~when all mortality rates are identical and equal to the natality rate).
Alternatively, \eqref{eq18c} is fulfilled when $d_i \geq b, i=1,\dots, m$ (every mortality rate at most equal to the natality rate).
%, with $F^I_+$ and $F^I_-$ equal to zero.

\comment{
Hence, we have 
\[
\cF^I(x,\lambda)= \lambda x^U e_1,\qquad
\cV^I(x,\lambda)=-V^Ix^I,
\]
where 
\begin{gather*}
\cV^I_{-}(x)= \diag\{ \nu_1+d_1,\dots,\nu_{m-1}+d_{m-1}, d_m \}x^I,\\
	\cV^I_{+}(x)
	= \begin{pmatrix}
0& 0 & \cdots &\cdots & 0 \\
\nu_1 & 0  & 0 &  \cdots & 0 \\
0 &  \nu_2 & \ddots &  &  \vdots\\
\vdots & \vdots &  & 0 & 0\\
0 &  0  & \dots  & \nu_{m-1} &  0
\end{pmatrix}
x^I
\end{gather*}
}

\subsection{Multistrain model}
\label{se54}

The fourth example from~\cite{Driessche:2002aa} is a multistrain model adapted from~\cite{Castillo-Chavez:1997aa,Feng:1997aa}.
The model has a unique susceptible compartment, but two infectious compartments corresponding to the two infectious agents, with possible super-infections by strain one of individuals already infected by strain two, giving rise to a new infection in compartment $I_1$.
The model is the following 3-dimensional model, where the total population is $N:=S+I_1+I_2$.
\begin{equation*}
%\label{W04}
\left \lbrace
\begin{array}{lll}
\dot{I_1}
&=&
\dps \beta_1\frac{I_1}{N} S-(b+\gamma_1)I_1+\nu \frac{I_1}{N} I_2,\\
\dot{I_2}
&= &
\dps \beta_2 \frac{I_2}{N} S-(b+\gamma_2)I_2-\nu \frac{I_1}{N} I_2,\\
\dot{S}
&= &
\dps bN -bS+\gamma_1I_1+\gamma_2 I_2- \frac{\beta_1I_1+\beta_2I_2}{N}S.
\end{array} \right.
\end{equation*}

The setting in \eqref{eq41} is
%\eqref{eq3017}
recovered by putting
\begin{subequations}
\begin{gather}
x^I:=\begin{pmatrix} I_1 \\ I_2  \end{pmatrix}, \qquad
 x^U=S,\qquad
 n_I = n_\lambda := 2,\ n_U :=1,\qquad
\lambda(x) = \begin{pmatrix} \lambda_1(x) \\ \lambda_2(x) \end{pmatrix},\\
F_0(\lambda) := \frac{\nu}{\beta_1} \lambda_1\begin{pmatrix}
0 & 1 \\ 0 & -1 
\end{pmatrix},\qquad
%F^I_+(\lambda) := \frac{\nu}{\beta_1} \lambda_1\begin{pmatrix}
%0 & 1 \\ 0 & 0 
%\end{pmatrix},\qquad
%F^I_-(\lambda) \equiv  := \frac{\nu}{\beta_1} \lambda_1\begin{pmatrix}
%0 & 0 \\ 0 & 1 
%\end{pmatrix},\\
F_+(\lambda) := \lambda_1 \begin{pmatrix}
	1 \\  0  \end{pmatrix}
	+ \lambda_2  \begin{pmatrix}
	0 \\ 1 \end{pmatrix},\qquad
F_-(\lambda) := \lambda_1 +\lambda_2,\\
%\cF^I(x,\lambda)
%	:=
%	\lambda_1 \begin{pmatrix}
%	1 \\  0  \end{pmatrix}x^U
%	+ \lambda_2  \begin{pmatrix}
%	0 \\ 1 \end{pmatrix}x^U
%	 + \frac{\nu}{\beta_1} \lambda_1\begin{pmatrix}
%0 & 1 \\ 0 & -1 
%\end{pmatrix} x^I,\qquad
%\cV^I(x,\lambda)
%	:= V^Ix^I,\\
%\cV^U(x,\lambda)
%	:= V^U x
%	-(\lambda_1 +\lambda_2  )x^U,\\
	V^I :=
	\begin{pmatrix}
	-(b+\gamma_1) & 0  \\ 0 & -(b+\gamma_2) \end{pmatrix},
	\qquad
	V^U
	:=\begin{pmatrix}
	b+\gamma_1 & b+\gamma_2  & 0  \end{pmatrix},\\
\lambda_{1}(x)=\beta_1 \frac{I_1}{S+I_1+I_2},\qquad \lambda_{2}(x)= \beta_2 \frac{I_2}{S+I_1+I_2}.
\end{gather}
\end{subequations}

As before, $V^I$ is a nonsingular Metzler matrix and \eqref{eq18a}, \eqref{eq18b} and \eqref{eq18d} are fulfilled.
Notice that here the functions $F^I_+$ and $F^I_-$ that appear in the decomposition \eqref{eq17} are {\em not} zero, so that $F_0$ is not identically zero.
This is due to the fact that infections may occur for already infected individuals.

\comment{
Hence, we have
\begin{gather*}
\cF^I(x,\lambda)= \lambda_1 \begin{pmatrix}
	1 & 0  \\ 0 &  0  \end{pmatrix}x^U
	+ \lambda_2 \begin{pmatrix}
	0 & 0  \\ 0 & 1 \end{pmatrix}x^U
+\frac{\nu}{\beta_1} \lambda_1\begin{pmatrix}
0 & 1 \\ 0 & -1 
\end{pmatrix} x^I,\\
\cV^I(x,\lambda)=-V^Ix^I,
\end{gather*}
where \[
\cV^I_{-}(x,\lambda) = -V^Ix^I,\quad
	\cV^I_{+}(x,\lambda)= 0_2.
	\]
}

\subsection{Vector-host model}
\label{se55}

The last example from~\cite{Driessche:2002aa} is a simplified version of a vector–host model borrowed from~\cite{Feng:1997aa}.
The latter couples a simple SIS model for the hosts with an SI model for the vectors.
It has four compartments: infected hosts ($I$) and vectors ($V$), and susceptible hosts ($S$) and vectors ($M$).
Inter-species transmission occurs in both directions through infective contacts between vectors and hosts.
The frequency of the (infectious) contacts is proportional to the number of vectors, with different rates $\beta_s$ and $\beta_m$, depending on the direction.
The equations are as follows.
\begin{equation}
\label{eq38}
\left \lbrace\begin{array}{lll}
\dot{I}&=& \dps \beta_{s} \frac{S}{S+I}V-(b+\gamma)I,\\
\dot{V}&= & \dps \beta_{m} \frac{I}{S+I}M -cV,\\
\dot{S}&= & \dps b(S+I)-bS+\gamma I-\beta_{s} \frac{S}{S+I}V,\\
\dot{M}&= & \dps c (M+V) -cM-\beta_{m} \frac{I}{S+I}M,
\end{array} \right.
\end{equation}
%here we consider the case where the total of host population equal $S+I$ and the vector population $M+V.$
%Denoting $x:=\begin{pmatrix} I &V &S& M\end{pmatrix}\t$, 
Model \eqref{eq38} is 4-dimensional and rewrites
\begin{subequations}
	\label{eq28}
	\begin{gather}
	\label{eq28e}
	x^I:=\begin{pmatrix} I \\ V  \end{pmatrix}, \qquad
	x^U:=\begin{pmatrix} S \\ M  \end{pmatrix},\qquad
	n_I = n_U = n_\lambda := 2,\qquad
	\lambda := \begin{pmatrix} \lambda_1 \\ \lambda_2 \end{pmatrix},\\
	\label{eq28a}
F_0(\cdot) \equiv 0_{2\times 2},\qquad
F_+(\lambda) = F_-(\lambda)
:= \lambda_1 \begin{pmatrix}
	1 & 0  \\ 0 &  0  \end{pmatrix}
	+ \lambda_2 \begin{pmatrix}
	0 & 0  \\ 0 &  1  \end{pmatrix},\qquad\\
%	\cF^I(x,\lambda)
%	:= \lambda_1 \begin{pmatrix}
%	1 & 0  \\ 0 &  0  \end{pmatrix}x^U
%	+ \lambda_2 \begin{pmatrix}
%	0 & 0  \\ 0 &  1  \end{pmatrix}x^U,\qquad
%	\cV^I(x,\lambda)
%	:= V^Ix^I,\\
%	\label{eq28b}
%	\cV^U(x,\lambda)
%	:= V^Ux
%	-\lambda_1 \begin{pmatrix}
%	1 & 0  \\ 0 &  0  \end{pmatrix}x^U
%	-\lambda_2 \begin{pmatrix}
%	0 & 0  \\ 0 &  1  \end{pmatrix}x^U,\\
	\label{eq28c}
	 V^I :=
	\begin{pmatrix}
	-(b+\gamma) & 0  \\ 0 & -c \end{pmatrix},
	\qquad
	V^U
	:=\begin{pmatrix}
	b+\gamma & 0 & 0  & 0 \\ 0 & c & 0 & 0 \end{pmatrix} ,\\
	\label{eq28d}
%\lambda(x) = (\lambda_1,\lambda_2)\t, \qquad
\lambda_{1}(x)=\beta_s \frac{V}{S+I},\qquad \lambda_{2}(x)=\beta_m \frac{I}{S+I}.
%	x :=\begin{pmatrix} x^I \\ x^U  \end{pmatrix},\qquad
%	f(x,\lambda) :=\begin{pmatrix} f^I (x,\lambda)\\ f^U (x,\lambda) \end{pmatrix}.
	\end{gather}
\end{subequations}

For this example too, one shows easily that $V^I$ is a nonsingular Metzler matrix and \eqref{eq18a}, \eqref{eq18b} and \eqref{eq18d} are fulfilled.

\section{Models of epidemics in commuting populations}
% and hypotheses}
\label{se3}

For clarity, we first recall and summarize here the proposed class of models (Section \ref{se31}).
The disease-free states and in particular the disease-free equilibria are studied in Section \ref{se315}.
In Section \ref{se32}, we state a list of adequate assumptions, which basically allows to retrieve the framework put in place by van den Driessche and Watmough \cite{Driessche:2002aa}.
It is then shown in Section \ref{se34} how these assumptions allow to find the basic reproduction number in the case of a single population (i.e.~for $|\cP|=1$), which is the case treated in the latter paper.

\subsection{A class of epidemiological models in commuting populations}
\label{se31}

For sake of clarity, we repeat and summarize here in a compact way the ingredients of the proposed epidemiological models in commuting population.\\

Assume be given matrices
\begin{subequations}
\label{eq300}
\begin{equation}
\label{eq300a}
V^I \in \Rset^{n_I\times n_I},\qquad 
V^U \in \Rset^{n_U\times n},
\end{equation}
linear maps
\begin{equation}
\label{eq300b}
F_0 (\lambda) \ :\ \Rset^{n_\lambda} \to \Rset^{n_I\times n_I},\quad F_+ (\lambda) \ :\ \Rset^{n_\lambda} \to \Rset^{n_I\times n_U}, \quad F_- (\lambda) \ :\ \Rset^{n_\lambda} \to \Rset^{n_U\times n_U},
\end{equation}
and a function $\lambda\ :\ \Rset_+^n \to \Rset_+^{n_\lambda}$, define 
\begin{equation}
\label{eq300c}
f(x,\lambda) :=
\begin{pmatrix} (F_0 (\lambda) + V^I) x^I + F_+ (\lambda) x^U\\ V^U x - F_- (\lambda) x^U \end{pmatrix}.
\end{equation}
\end{subequations}

%For any admissible switching signal $\cR(\cdot)$ (see Definition \ref{de2}), define the evolution on $\cI$ by the following equation:
%%\begin{subequations}
%%\label{eq2500}
%\begin{equation}
%\label{eq250}
%\hspace{-.25cm}
%\dot x_p = f(x_p(t),\lambda_{[p]_{\cR(t)}}(t)),
%\hspace{.2cm}
%\lambda_{[p]_{\cR(t)}}(t) = \lambda_{[p]_{\cR(t)}} \left(
%x_{[p]_{\cR(t)}}(t)
%\right),
%\hspace{.2cm}
%x_{[p]_{\cR(t)}} := \hspace{-.2cm} \sum_{p'\in [p]_{\cR(t)}} \hspace{-.2cm} x_{p'},
%\hspace{.2cm}
%t\in\cI,\quad p\in\cP.
%\end{equation}
%{\PAB with the continuity condition}
%%and the two different types of switching conditions
%\begin{equation}
%\label{eq251}
%x_p(t^+) = x_p(t^-),
%\quad t\in\Rset_+\setminus\cI,
%\quad p\in\cP.
%\end{equation}
%%and
%%\begin{equation}
%%\label{eq252}
%%x_p(t^+) = \frac{\bfo_n\t x_p(t^-)}{\bfo_n\t x_{[p]_{\cR(t^-)}}(t^-)}x_{[p]_{\cR(t^-)}}(t^-),
%%\quad t\in\Rset_+\setminus\cI,
%%\quad p\in\cP.
%%\end{equation}
%%\end{subequations}

%\begin{itemize}
%\item {\bf epidemiological model in commuting population} the system \eqref{eq250}-\eqref{eq251};
%\item {\bf epidemiological model in commuting population with random draw} the system \eqref{eq250}-\eqref{eq252}.
%\end{itemize}

When needed, we will use the following decomposition, deduced from~\eqref{eq300} and quite similar to the one in~\cite{Driessche:2002aa}:
\begin{subequations}
\label{eq177}
\begin{equation}
\label{eq177a}
\cF^I(x,\lambda) := F^I_+ (\lambda) x^I + F_+ (\lambda) x^U,\quad
\cV^I(x,\lambda) := V^I x^I -  F^I_- (\lambda) x^I,\quad
\cV^U(x,\lambda) := V^U x - F_- (\lambda) x^U,
\end{equation}
where {\em by definition}
\begin{equation}
\label{eq177b}
F^I_- (\lambda) \text{ is the main diagonal of } - F_0(\lambda),\qquad
F^I_+ (\lambda) := F_0(\lambda)+F^I_- (\lambda).
\end{equation}
\end{subequations}
The diagonal of the matrix $F^I_+ (\lambda)$ is thus zero by construction.

\subsection{Disease-free states and disease-free equilibrium points}
\label{se315}

First, as in \cite{Driessche:2002aa}, we define $\Xbf_s$ to be the set of all {\em disease-free states}:
\begin{equation}
\label{eq40}
\Xbf_s := \left\{
\begin{pmatrix} 0_{n_I} \\ x^U \end{pmatrix}\ :\ x^U\in\Rset_+^{n_U}
\right\}.
\end{equation}

Another important subset of the state space is the set of disease-free equilibria (DFE), a subset of the previous one.
Here, contrary to \cite{Driessche:2002aa}, we put in this category any equilibrium point, {\em whatever its stability properties.}
In absence of infection, one has $\lambda=0_{n_\lambda}$, and the set $\Xbf_\eq\subset\Xbf_s$ of the DFEs of system \eqref{eq2} %and of the system \eqref{eq250}-\eqref{eq252}
is given by
\begin{equation}
\label{eq20}
\Xbf_\eq := \{x\in\Xbf_s\ :\ f(x,0_{n_\lambda}) = 0_n\}.
\end{equation}

Due to the assumed property of linearity of $f$ with respect to $x$, the set $\{x\in\Rset^n\ :\ f(x,0_{n_\lambda}) = 0_n\}$ is a vectorial space, and the set $\Xbf_\eq$ is therefore a {\em cone}.
Assume $\Xbf_\eq$ is not empty, so that there exists (at least) one DFE $y^*>0$ of \eqref{eq0}.
%Of course, in order to compute some basic  assume in the sequel that the latter is not restricted to $0_n$.
%This implies that, for any given total population number of $N>0$ individuals, there
Then, for any given total population number of $N>0$ (non-infected) individuals, there exists a DFE $x^*>0$ such that $\bfo_n\t x^* = N$, namely
\begin{equation*}
%\label{eq24}
x^* := \frac{N}{\bfo_n\t y^*}y^*,
%\bfo\t x^* = N,
\end{equation*}
or equivalently, considering only the non-infected compartments,
\[
x^{U*} := \frac{N}{\bfo_{n_U}\t y^{U*}}y^{U*}.
\]
%where by definition $\bfo\in\Rset^n$ is such that $\bfo\t :=\begin{pmatrix} 1 & \dots & 1 \end{pmatrix}$.
As a matter of fact, for $x^*$ defined as before, $f(x^*,0_{n_\lambda}) = \frac{N}{\bfo_n\t y^*}f(y^*,0_{n_\lambda}) = 0_n$ by linearity, so that $x^*$ is a DFE; and $\bfo_n\t x^* := \frac{N}{\bfo_n\t y^*}\bfo_n\t y^* = N$.

% let $y^*$ be a nonzero DFE of \eqref{eq0}, then $x^* := \frac{N}{\bfo\t y^*}y^*$ is a DFE fulfilling the previous condition, as it fulfils \eqref{eq24} and moreover $f(x^*,0) = \frac{N}{\bfo\t y^*}f(y^*,0) = 0$ by linearity.
Let us now characterize in a more precise manner the set $\Xbf_\eq$ and the condition $\Xbf_\eq\neq\emptyset$. 
For any $x\in\Xbf_s$, one has $x^I=0_{n_I}$ and $\lambda(x) = 0_{n_\lambda}$, so that
\[
f(x,0_{n_\lambda}) := \begin{pmatrix}
(F_0 (0_{n_\lambda}) + V^I) x^I + F_+ (0_{n_\lambda}) x^U\\ V^U x - F_- (0_{n_\lambda}) x^U \end{pmatrix}
= \begin{pmatrix} 0_{n_I} \\ V^U x \end{pmatrix}.
\]
Therefore, the cone $\Xbf_\eq$ is characterized by the following property:
\begin{equation}
\label{eq42}
\Xbf_\eq = \Xbf_s \cap \ker V^U.
\end{equation}

%With these notations, the set of the DFEs is exactly
%\[
%\left\{
%\begin{pmatrix} 0_{n_I} \\ x^{U*} \end{pmatrix}\ :\ x^{U*}\in\Rset_+^{n_U}
%\right\} \cap \ker V^U.
%\]

\subsection{Assumptions on the models}
\label{se32}

We make the following assumptions on the objects defined in Section~\ref{se31}.
%\eqref{eq300}.

\begin{assumption}[On the operators $F_0, F_\pm$] 
\label{as1}
The operators $F_0 (\cdot), F_+ (\cdot), F_- (\cdot)$ defined in \eqref{eq300b} are linear.
% and map the domain positive cone to the codomain positive cone.
Moreover, for any $\lambda\in\Rset^{n_\lambda}_+$, the matrix $F_0 (\lambda)$ is a Metzler matrix and $F_+ (\lambda), F_- (\lambda)$ are nonnegative, with $F_- (\lambda)$ diagonal.
Last,
\begin{equation}
\label{eq189}
\bfo_{n_I}\t F_0 (\cdot) = 0_{n_I}\t,\qquad
\bfo_{n_I}\t F_+ (\cdot) = \bfo_{n_U}\t F_-(\cdot).
\end{equation}
\end{assumption}

\begin{assumption}[On the matrices $V^I, V^U$]
\label{as3}
For the matrices $V^I, V^U$ defined in \eqref{eq300a},
\begin{itemize}
\item
$V^U \begin{pmatrix} 0_{n_I \times n_U} \\ I_{n_U} \end{pmatrix}$ is a Metzler matrix of which $0$ is a simple eigenvalue associated to the left-eigenvector $\bfo_{n_U}\t$;
\item $V^I$ is a Metzler matrix,
$V^U \begin{pmatrix} I_{n_I} \\ 0_{n_U\times n_I} \end{pmatrix} > 0_{n_U\times n_I}$ and
\begin{equation}
\label{eq188}
-\bfo_{n_I}\t V^I \geq \bfo_{n_U}\t V^U  \begin{pmatrix}
I_{n_I} \\ 0_{n_U \times n_I}
\end{pmatrix} \succ 0_{1\times n_I}.
\end{equation}
\end{itemize}
\end{assumption}

\begin{assumption}[On the map $\lambda$]
\label{as4}
The map $\lambda$ takes on nonnegative values, and is zero on $\Xbf_s$.
\end{assumption}

Explanations of the meaning of Assumptions~\ref{as1} and~\ref{as3} have been given in Section~\ref{se50}, and Assumption~\ref{as4} requires no particular comment.

One may check without major difficulties that the examples displayed in Section~\ref{se2} fit Assumptions~\ref{as1} to~\ref{as4}. 
(See however an important nuance in Remark~\ref{re40} on the simplicity of the eigenvalue $0$ (Assumption~\ref{as3}), in relation with the dimension of the set of disease-free equilibrium points.)

Notice that~\eqref{eq189} repeats~\eqref{eq18a} without modification.
Also, the link between~\eqref{eq188} and \eqref{eq18b}/\eqref{eq18c} is evident, while \eqref{eq18d} is contained in the first point of Assumption~\ref{as3}.
Before going further, we enlighten the relation of the previous assumptions with the framework invented by van den Driessche and Watmough~\cite{Driessche:2002aa}.
\begin{proposition}
\label{pr1}
Assume $f$ defined by~\eqref{eq300} fulfils Assumptions~\ref{as1},~\ref{as3},~\ref{as4}.
Then the same is true for Assumptions {\bf (A.1)}, {\bf (A.2)}, {\bf (A.3)}, {\bf (A.4)} in {\rm \cite{Driessche:2002aa}}.
%\color{red}
%Moreover, there exist Disease-Free Equilibria and at each point $x^*\in\Xbf_\eq$, the Jacobian of $\begin{pmatrix}
%\cV^I(x,\lambda(x)) \\ \cV^U(x,\lambda(x))
%\end{pmatrix}$ admits 0 as simple eigenvalue, and apart from this all its eigenvalues have negative real parts.
\end{proposition}

As will be seen below, the present setting does not contain an hypothesis on asymptotic stability as strong as~\cite[Assumption~{\bf (A.5)}]{Driessche:2002aa}.
Indeed, Assumption~\ref{as4} allows for a unique zero eigenvalue in the spectrum of the Jacobian matrix of the system in disease-free evolution.
This comes from the fact that, due to the second formula in~\eqref{eq188}, the population is constant along any disease-free evolution, see the discussion in the beginning of Section~\ref{se2}.
(Recall that extensive quantities are considered here; see the related discussion in Section \ref{se0}.)
This insensitivity to the total population level generates a zero eigenvalue in the Jacobian spectrum at any equilibrium point.

%{\color{red} The function whose Jacobian is considered in Proposition~\ref{pr1} is, as in~\cite[Assumption~{\bf (A.5)}]{Driessche:2002aa}, equal to the right-and side $f(x,\lambda(x))$ `when the function $\cF^I(x,\lambda(x))$ is set to zero'.}

\begin{proof}[Proof of Proposition~\rm\ref{pr1}]
\mbox{}

\noindent $\bullet$ We first retrieve from~\eqref{eq300} the setting of~\cite{Driessche:2002aa}.
For this, let $\cF^I(x,\lambda)$, $\cV^I(x,\lambda)$, $\cV^U(x,\lambda)$ be deduced from model~\eqref{eq300} thanks to~\eqref{eq177}.
From Assumption~\ref{as3}, one deduces that for any $\lambda\in\Rset_+^{n_\lambda}$, $F_0(\lambda)$ is a Metzler matrix and its off-diagonal elements are nonnegative.
By construction, the latter are also the off-diagonal elements of $F^I_+ (\lambda)$, which possesses a zero diagonal.
Thus $F^I_+ (\lambda)\geq 0_{n_I\times n_I}$.
One the other hand, $-F^I_- (\lambda)$ is a diagonal matrix that bears the diagonal elements of $F_0(\lambda)$.
Due to Assumption~\ref{as1}, one has
\[
0_{n_I}\t = \bfo_{n_I}\t F_0 (\lambda) = \bfo_{n_I}\t (F^I_+ (\lambda)-F^I_- (\lambda)),
\]
so that $\bfo_{n_I}\t F^I_- (\lambda) = \bfo_{n_I}\t F^I_+ (\lambda) \geq 0$, and the diagonal matrix $F^I_- (\lambda)$ is nonnegative.
Now, the decomposition of $f$ according to
\[
f(x,\lambda) = \begin{pmatrix} \cF^I(x,\lambda) + \cV^I(x,\lambda) \\ \cV^U(x,\lambda) \end{pmatrix}
\]
is identical to the one given in~\cite{Driessche:2002aa}.
We will now establish that it fulfils the Assumptions~{\bf (A.1)} to~{\bf (A.4)} of the present paper.

\noindent $\bullet$
Due to Assumption~\ref{as1}, the function $\cF^I(x,\lambda(x))$ in the decomposition \eqref{eq177}, which gathers the infection terms, takes on nonnegative values.
This is essentially the content of~\cite[Assumption~{\bf (A.1)}]{Driessche:2002aa}.

Assumption~{\bf (A.3)} in~\cite{Driessche:2002aa} states that no infection term feeds the non-infected compartments.
This property is already fulfilled by the choice of the structure of model~\eqref{eq300}.

Translated in the present setting, Assumption~{\bf (A.2)} in~\cite{Driessche:2002aa} states that whenever $x_i=0$ for some $i=1,\dots, n$, then the corresponding component $\cV^I_i(x,\lambda(x))\geq 0$.
First let $i\in\{1,\dots, n\}$ be the index of an infected compartment, that is $i\in\{1,\dots, n_I\}$. 
On the one hand, due to the fact that $V^I$ is a Metzler matrix by Assumption~\ref{as3}, and that $x\geq 0_n$, one has $(V^Ix^I)_i \geq 0$ whenever $x_i=0$.
On the other hand, the operator $F^I_- (\lambda)$ being diagonal, the $i$-th component of $F^I_- (\lambda) x^I$ is zero when $x_i=0$.
Therefore, as $\cV^I(x,\lambda) := V^I x^I -  F^I_- (\lambda) x^I$, one obtains that $\cV^I_i(x,\lambda(x))\geq 0$ whenever $x_i=0$.

The same argument allows to treat the case of a non-infected compartment.
Let $i\in\{1,\dots, n\}$ be the index of such a compartment, that is $i\in\{n_I+1,\dots, n\}$, and $x\in\Rset_+^n$ such that $x_i=0$.
One has
\begin{equation}
\label{eq333}
V^U x
= V^U \left(
\begin{pmatrix} I_{n_I} \\ 0_{n_U\times n_I} \end{pmatrix} x^I
+ \begin{pmatrix} 0_{n_I \times n_U} \\ I_{n_U} \end{pmatrix} x^U
\right)
= V^U \begin{pmatrix} I_{n_I} \\ 0_{n_U\times n_I} \end{pmatrix} x^I
+ V^U \begin{pmatrix} 0_{n_I \times n_U} \\ I_{n_U} \end{pmatrix} x^U.
\end{equation}
The $i$-th component of the first term is nonnegative due to the fact that $x^I\geq 0_{n_I}$ and $V^U \begin{pmatrix} I_{n_I} \\ 0_{n_U\times n_I} \end{pmatrix} > 0_{n_U \times n_I}$ (Assumption~\ref{as3});
while the $i$-th component of the second term is also nonnegative, due to the fact that $x_i=0$ and $x^U\geq 0_{n_U}$, and that, by Assumption~\ref{as3}, $V^U \begin{pmatrix} 0_{n_I \times n_U} \\ I_{n_U} \end{pmatrix}$ is a Metzler matrix.
%The $i$-th component of the first term is nonnegative, due to the fact that $x_i=0$ and $x^U\geq 0_{n_U}$, and that, by Assumption~\ref{as3}, $V^U \begin{pmatrix} 0_{n_I \times n_U} \\ I_{n_U} \end{pmatrix}$ is a Metzler matrix;
%while the $i$-th component of the second term is also nonnegative due to the fact that $x^I\geq 0_{n_I}$ and $V^U \begin{pmatrix} I_{n_I} \\ 0_{n_U\times n_I} \end{pmatrix} > 0_{n_U \times n_I}$ (Assumption~\ref{as3}).
Overall, $(V^Ux)_i \geq 0$.
On the other hand, the operator $F_- (\lambda)$ is diagonal (Assumption~\ref{as1}), so that the $i$-th component of $F_- (\lambda) x^U$ is zero when $x_i=0$.
Therefore, $\cV^U_i(x,\lambda)  \geq 0$ %= V^U x - F_- (\lambda) x^U\geq 0_{n_U}$
whenever $x_i=0$.
This achieves the proof of Assumption~{\bf (A.2)} in~\cite{Driessche:2002aa}. 

Last, Assumption~{\bf (A.4)} in~\cite{Driessche:2002aa} states that at any disease-free state $x\in \Xbf_s$, the infectious term $\cF^I(x,\lambda(x)) = F^I_+ (\lambda) x^I + F_+ (\lambda) x^U$ is equal to $0_{n_I}$ and that $\cV^I(x,\lambda(x)) = V^I x^I -  F^I_- (\lambda) x^I \leq 0_{n_I}$.
Let $x\in \Xbf_s$, then $x^I=0_{n_I}$ and all terms $F^I_+ (\lambda)x^I$, $F^I_- (\lambda)x^I$ and $V^I x^I$ are thus $0_{n_I}$.
On the other hand, $\lambda(x) = 0_{n_\lambda}$ due to Assumption~\ref{as4}, so by linearity the matrix $F_+(\lambda)$ is null, and thus $F_+ (\lambda) x^U = 0_{n_I}$, so that $\cV^I(x,\lambda(x)) = 0_{n_I}$ as well.
This achieves the proof of Proposition~\ref{pr1}.
\end{proof}

%\begin{remark}
%\label{re2}
%Notice that the second identity in \eqref{eq188} implies that $0$ is an eigenvalue of the Metzler matrix $V^U \begin{pmatrix} 0_{n_I \times n_U} \\ I_{n_U} \end{pmatrix}$, and in fact that it is its eigenvalue with the largest real part.
%%Notice that the second identity in \eqref{eq188} implies that the dominant eigenvalue of the Metzler matrix $V^U \begin{pmatrix} 0_{n_I \times n_U} \\ I_{n_U} \end{pmatrix}$ is zero.
%However, it does not imply that this eigenvalue is simple.
%\end{remark}

\subsection{Basic reproduction number for a single-population model}
\label{se34}

The following result shows essentially that the previous setting contains the stationary, single-population, case treated in~\cite{Driessche:2002aa}.
It gives the value of the basic reproduction number of model \eqref{eq0} with a dynamic defined by $f$ in \eqref{eq300}.
The proof consists mainly in adapting the framework of \cite[Theorem 2]{Driessche:2002aa}.
%The proof consists essentially in showing that the framework of \cite[Theorem 2]{Driessche:2002aa} may be adapted to this new situation.

For any sufficiently regular function $H(x)\ :\ \Rset^n \to \Rset^r$, we decompose the Jacobian matrix of $H$ in $\Rset^{r\times n}$ as
\begin{subequations}
\label{eq195}
\begin{equation}
\frac{\partial H}{\partial x}(x)
= \begin{pmatrix}
\displaystyle \frac{\partial H}{\partial x^I}(x)  & \displaystyle \frac{\partial H}{\partial x^U} (x)
\end{pmatrix},\qquad
\frac{\partial H}{\partial x^I}(x) \in\Rset^{r\times n_I},\ \frac{\partial H}{\partial x^U}(x) \in\Rset^{r\times n_U}.
\end{equation}
For simplicity, this will be written in the sequel with the more compact notation
\begin{equation}
\frac{\partial H}{\partial x}(x)
= \partial_x H(x)
= \begin{pmatrix}
\displaystyle \partial_{x^I} H(x)  & \displaystyle \partial_{x^U} H(x)
\end{pmatrix}.
\end{equation}
\end{subequations}

\begin{theorem}[Basic reproduction number and stability for single-subpopulation epidemic models]
\label{th1}

Assume $f$ %(decomposed according to~\eqref{eq177})
and $\lambda$ fulfil Assumptions \ref{as1}, \ref{as3}, \ref{as4}.
Then the following properties hold.

%Assume Assumptions \ref{as1}, \ref{as3}, \ref{as4} are fulfilled.

\begin{itemize}
\item
System \eqref{eq0}-\eqref{eq300} admits nonzero disease-free equilibrium points.

\item
For any nonzero disease-free equilibrium point $x^*\in\Xbf_{\eq}$ of system \eqref{eq0}-\eqref{eq300}, let
\begin{subequations}
\label{eq355}
\begin{equation}
\label{eq355a}
\cR_0 := \rho \left(F
 V^{-1}
\right),
\end{equation}
where $F, V\in\Rset^{n_I\times n_I}$ are given by:
\begin{equation}
\label{eq355b}
F := \left.
\partial_{x^I} \left(
F_+(\lambda(x))x^{U*}
\right)
\right|_{x=x^*},
%= \frac{\partial \left[
%F_+(\lambda)x^{U*}
%\right]}{\partial\lambda}(0_{n_\lambda}) \frac{\partial \lambda}{\partial x^I}(x^{U*})
\qquad
V := V^I.
\end{equation}
\end{subequations}
The following assertions are true:
\begin{itemize}
\item
if $\cR_0>1$, then $x^*$ is unstable;
\item
if $\cR_0<1$, then %$\Xbf_{\eq}$ is attracting and
$x^*$ is marginally stable.
Moreover, if the first identity in \eqref{eq188} holds {\em with an equality}, then $x^*$ is locally asymptotically stable.

% and $\Xbf_{\eq}$ is attracting.
\end{itemize}
\end{itemize}
\end{theorem}

In agreement with the notation defined in~\eqref{eq195}, the matrix $F$ of Theorem~\ref{th1} is such that
\[
F = \left.
\partial_{x^I} \left(
F_+(\lambda(x))x^{U*}
\right)
\right|_{x=x^*}
= \left.
\frac{\partial \left(
F_+(\lambda(x))x^{U*}
\right)}{\partial x^I}
\right|_{x=x^*}
= \sum_{i=1}^{n_\lambda} \left.
\frac{\partial (F_+(\lambda)x^{U*})}{\partial \lambda_i}
\right|_{\lambda= 0_{n_\lambda}}
\left.
\frac{\partial \lambda_i(x)}{\partial x^I}
\right|_{x=x^*}.
\]

\begin{proof}[Proof of Theorem~\rm\ref{th1}]
\mbox{}

\noindent $\bullet$
From Assumption \ref{as3}, the matrix $M := V^U \begin{pmatrix} 0_{n_I \times n_U} \\ I_{n_U} \end{pmatrix} \in\Rset^{n_U\times n_U}$ is a Metzler matrix, of which the positive vector $\bfo_{n_U}\t$ is a left-eigenvector associated to the eigenvalue 0.
For large enough value of $\alpha$, $M+\alpha I_{n_U}$ is non-negative, and $\bfo_{n_U}\t$ is a left-eigenvector associated to the eigenvalue $\alpha$.
The latter is thus the spectral radius of $M+\alpha I_{n_U}$~\cite[Corollary 1.1.12]{Berman:1994aa}, so that $0$ is the stability modulus of $M$ and no other eigenvalue is located on the imaginary axis.
%From~\eqref{eq188}, one also gets that, for any $\alpha>0$, $\bfo_{n_U}\t (M- \alpha I_{n_U}) = -\alpha \bfo_{n_U}\t \prec 0_{n_U}\t$, so that $M-\alpha I_{n_U}$ is Hurwitz~\cite[Theorem 2.1]{Mason:2007aa} and $s(M) < \alpha$ for any $\alpha>0$, so that $s(M)\leq 0$ and finally $s(M)=0$.
%, so that $s(M) \geq 0$.
%From~\eqref{eq188}, one also gets that,
%for any $\alpha>0$, $\bfo_{n_U}\t (M- \alpha I_{n_U}) = -\alpha \bfo_{n_U}\t \prec 0_{n_U}\t$, so that $M-\alpha I_{n_U}$ is Hurwitz~\cite[Theorem 2.1]{Mason:2007aa} and $s(M) < \alpha$ for any $\alpha>0$, so that $s(M)\leq 0$ and finally $s(M)=0$.

%Now, for large enough value of $\alpha$, $M+\alpha I_{n_U}$ is non-negative, and its dominant eigenvalue $\alpha$ is thus equal to its spectral radius~\cite[chap.~XIII, Theorem 3]{Gantmakher:2000aa}.
%%the dominant eigenvalue is equal to $\alpha$.}
%The matrix $M+\alpha I_{n_U}$ thus admits an associated {\em nonnegative right-}eigenvector $w\in\Rset_+^{n_U}$.
The matrix $M$ also admits a {\em nonnegative} right-eigenvector $w\in\Rset_+^{n_U}$, see~\cite[Theorem 1.1.1]{Berman:1994aa}.
Now, the nonzero, nonnegative, vector $\begin{pmatrix}
0_{n_I} \\ w
\end{pmatrix} \in \Rset_+^n$ is such that
\[
V^U \begin{pmatrix}
0_{n_I} \\ w
\end{pmatrix}
= V^U \begin{pmatrix} 0_{n_I \times n_U} \\ I_{n_U} \end{pmatrix} w
= Mw = 0_{n_U}.
\]
One concludes, using \eqref{eq42}, that
\[
\begin{pmatrix}
0_{n_I} \\ w
\end{pmatrix} \in \Xbf_s \cap \ker V^U = \Xbf_\eq,
\]
which demonstrates the existence of a nonzero disease-free equilibrium of system \eqref{eq0}-\eqref{eq300}.

\noindent $\bullet$
Generally speaking, the Jacobian matrix of the function $f(x,\lambda(x))$ for $f$ defined in \eqref{eq300} at any point $x\in\Rset_+^n$ is decomposed as
\begin{equation}
\label{eq45}
J(x) := \begin{pmatrix} J_{II}(x) & J_{IU}(x) \\ J_{UI}(x) & J_{UU}(x) \end{pmatrix},
\end{equation}
$J_{II}(x)\in\Rset^{n_I \times n_I }$, $J_{IU}(x) \in\Rset^{n_I \times n_U }$, $J_{UI}(x) \in\Rset^{n_U \times n_I }$, $J_{UU}(x)\in\Rset^{n_U \times n_U }$,
where
\begin{gather*}
J_{II}(x) :=
\partial_{x^I} \left(
(F_0 (\lambda(x)) + V^I) x^I + F_+ (\lambda(x)) x^U
\right), \\
J_{IU}(x) :=
\partial_{x^U} \left(
(F_0 (\lambda(x)) + V^I) x^I + F_+ (\lambda(x)) x^U
\right), \\
J_{UI}(x) :=
\partial_{x^I} \left(
V^U x - F_- (\lambda(x)) x^U
\right),\qquad
J_{UU}(x) :=
\partial_{x^U} \left(
V^U x - F_- (\lambda(x)) x^U
\right).
\end{gather*}

Let $x^*\in\Xbf_\eq$ be a (nonzero) DFE.
We will first show that at such a point, the block $J_{IU}(x^*)$ is null.
As a matter of fact, one has $x^{I*} = 0_{n_I}$, $\lambda(x^*) = 0_{n_\lambda}$, and thus the matrices $F_0(\lambda(x^*)), F_\pm(\lambda(x^*))$ are zero, due to the linearity of these operators, see Assumption~\ref{as1}.
Therefore,
\[
J_{IU}(x^*)
= \left.
\partial_{x^U} \left(
F_+ (\lambda(x)) x^U
\right) \right|_{x=x^*}
= F_+ (\lambda(x^*)) + \left.
\partial_{x^U} \left(
F_+ (\lambda(x)) x^{U*}
\right) \right|_{x=x^*}
= \frac{\partial F_+}{\partial \lambda}(0_{n_\lambda}) \frac{\partial \lambda}{\partial x^U}(x^*) x^{U*}.
\]
One has $\lambda(x)=0_{n_\lambda}$ at any point $x\in\Xbf_s$, so that
\[
\frac{\partial \lambda}{\partial x^U}(x^*) = 0_{n_\lambda\times n_U},
\]
and one concludes that at any DFE $x^*$,
\begin{equation}
\label{eq79}
J_{IU}(x^*) = 0_{n_I\times n_U}.
\end{equation}

Due to the previous property, the stability of any DFE $x^*$ is determined by the spectra of the two blocks $J_{II}(x^*)$ and $J_{UU}(x^*)$ of $J(x^*)$.
Let us compute these quantities.

On the one hand, $\left.
\frac{\partial}{\partial x^U} \left(
F_- (\lambda(x))
\right) \right|_{x=x^*}$ is shown to be zero, just as $\left.
\frac{\partial}{\partial x^U} \left(
F_+ (\lambda(x))
\right) \right|_{x=x^*}$ previously.
Thus,
\begin{equation}
\label{eq790}
J_{UU}(x^*)
=
\left.
\partial_{x^U}\left(V^U x - F_- (\lambda(x)) x^U\right)
\right|_{x=x^*}
= V^U \frac{\partial x}{\partial x^U} - F_- (\lambda(x^*))
= V^U \begin{pmatrix} 0_{n_I \times n_U} \\ I_{n_U} \end{pmatrix},
\end{equation}
which by Assumption~\ref{as3} has zero as simple eigenvalue and possibly other eigenvalues with negative real parts.

On the other hand, using the nullity of the terms $F_0 (\lambda(x))$ and $F_+ (\lambda(x))$, one has
\[
J_{II}(x^*)
=
\left.
\partial_{x^I} \left((F_0 (\lambda(x)) + V^I) x^I + F_+ (\lambda(x)) x^U \right)
\right|_{x=x^*}
= V^I + \left.
\partial_{x^I} \left[
F_+(\lambda(x))x^{U*}
\right]
\right|_{x=x^*}.
\]
One checks directly that
\[
\left.
\partial_{x^I} \left(
F_+(\lambda(x))x^{U*}
\right)
\right|_{x=x^*}
= \left.
\partial_\lambda \left(
F_+(\lambda)x^{U*}
\right) \right|_{\lambda = 0_{n_\lambda}}
\left.
\partial_{x^I} \lambda
\right|_{x=x^*},
\]
where, due to the linearity of $F_+(\cdot)$, the quantity $\partial_\lambda \left(
F_+(\lambda)x^{U*}
\right)$ does not depend upon $\lambda$.
Comparing with the notations introduced in \eqref{eq355b}, one thus obtains
\begin{equation}
J_{II}(x^*) = F + V.
\end{equation}

Notice first that the matrix $V=V^I$ is Hurwitz.
As a matter of fact, by hypotheses in~Assumption~\ref{as3}, one has
\[
\bfo_{n_I}\t V^I \leq - \bfo_{n_U}\t V^U  \begin{pmatrix}
I_{n_I} \\ 0_{n_U \times n_I}
\end{pmatrix} 
\prec 0.
\]
%due to the fact that $\bfo_{n_U}\t V^U  \begin{pmatrix}
%I_{n_I} \\ 0_{n_U \times n_I}
%\end{pmatrix} \succ 0_{1\times n_I}$ (see~Assumption~\ref{as3}).
This, together with the fact that this matrix is Metzler, implies that $V^I$ is Hurwitz, see~\cite[Theorem 2.1]{Mason:2007aa}, and in particular {\em invertible}.

On the other hand, by Assumption~\ref{as1}, the operator $F_+(\lambda)$ is positive.
This implies that the terms $\left.
\partial_\lambda \left(
F_+(\lambda)x^{U*}
\right) \right|_{\lambda = 0_{n_\lambda}}$
%$\frac{\partial (F_+(\cdot) x^{U*})}{\partial \lambda_i}$
are (constant) nonnegative matrices.
Also, the terms $\left.
\partial_{x^I} \lambda
\right|_{x=x^*}$
%$\frac{\partial \lambda_i(\cdot)}{\partial x^I}(x^*)$
are nonnegative, as $\lambda(x)$ is zero at any point $x\in\Xbf_s$ and takes on nonnegative values.
Thus, $F$ is a {\em nonnegative matrix}.
Arguing as in~\cite[Proof of Theorem 2]{Driessche:2002aa}, one then obtains that
\begin{equation}
\label{eq400}
s(J_{II}(x^*)) < 0 \Leftrightarrow \cR_0 < 1 \qquad \text{ and }\qquad
s(J_{II}(x^*)) > 0 \Leftrightarrow \cR_0 > 1,
\end{equation}
for the value $\cR_0$ defined in \eqref{eq355}, and where $s$ denotes the stability modulus, see the Notations in the end of the Introduction.

\noindent $\bullet$
Consider first the case where $\cR_0 > 1$, so that $s(J_{II}(x^*)) > 0$.
One has
\[
s(J(x^*)) = \max\{ s(J_{II}(x^*)), s(J_{UU}(x^*))\} = \max\{ s(J_{II}(x^*)), 0 \} = s(J_{II}(x^*)) > 0.
\]
The DFE $x^*$ is then unstable.

\noindent $\bullet$
Consider now the case where $\cR_0 < 1$.
One has here $s(J_{II}(x^*)) < 0$, but $s(J_{UU}(x^*)) = 0$ (contrary to what happens in~\cite{Driessche:2002aa}), and $J_{UU}(x^*)$ has only 0 as an imaginary eigenvalue, which is simple.
%Suppose now on the contrary that the first formula in \eqref{eq188} holds without 
%\[
%\bfo_{n_I}\t V^I + \bfo_{n_U}\t V^U  \begin{pmatrix}
%I_{n_I} \\ 0_{n_U \times n_I}
%\end{pmatrix}
%< 0_{1\times n_I}.
%\]
Let us first establish the (simple) stability of the DFE under the general condition \eqref{eq188}.
Let $\delta x := x-x^*$, $\delta x^I := x^I - x^{*I}$, $\delta x^U := x^U - x^{*U}$.
By definition, the tangent linear equation writes
\[
\delta \dot x^I = J_{II}(x^*) \delta x^I,\qquad
\delta \dot x^U = J_{UI}(x^*) \delta x^I + J_{UU}(x^*) \delta x^U.
\]
Then exploiting the asymptotic stability of the block $J_{II}(x^*)$, for sufficiently small $\varepsilon >0$, there exist $\zeta>0$, $c\geq 1$, $k>0$ such that, for any trajectory of~\eqref{eq0}-\eqref{eq300} such that $\|\delta x (0)\| < \zeta$,
\begin{equation}
\label{eq721}
\| \delta x^I (t)\|, \| \delta \dot x^I (t)\|  \leq ce^{-kt} \| \delta x (0)\| \qquad \text{ for any $t\geq 0$, as long as } \|\delta x (t)\| < \varepsilon.
\end{equation}
On the other hand, as seen in Section~\ref{se50}, one has for any trajectory
\begin{multline*}
\bfo_{n_I}\t \dot x^I + \bfo_{n_U}\t \dot x^U
=
\left[
\bfo_{n_I}\t \left(
F_0 (\lambda) + V^I
\right)
+ \bfo_{n_U}\t V^U  \begin{pmatrix}
I_{n_I} \\ 0_{n_U \times n_I}
\end{pmatrix}
\right]x^I\\
+ \left[
\bfo_{n_I}\t F_+ (\lambda)- \bfo_{n_U}\t F_- (\lambda)
+ \bfo_{n_U}\t V^U \begin{pmatrix}
0_{n_I \times n_U} \\ I_{n_U}
\end{pmatrix}
\right] x^U,
\end{multline*}
so that in the vicinity of the DFE $x^*$, one has $\lambda(x^*)=0_{n_\lambda}$ and
\[
\bfo_{n_I}\t \delta \dot x^I + \bfo_{n_U}\t \delta \dot x^U
= \left(
\bfo_{n_I}\t V^I + \bfo_{n_U}\t V^U  \begin{pmatrix}
I_{n_I} \\ 0_{n_U \times n_I}
\end{pmatrix}
\right) \delta x^I
+ \bfo_{n_U}\t V^U \begin{pmatrix}
0_{n_I \times n_U} \\ I_{n_U}
\end{pmatrix} \delta x^U.
\]
Due to the first point in Assumption~\ref{as3}, the second term is zero, yielding
\begin{eqnarray*}
\bfo_{n_U}\t \delta \dot x^U
& = &
- \bfo_{n_I}\t \delta \dot x^I + \left(
\bfo_{n_I}\t V^I + \bfo_{n_U}\t V^U  \begin{pmatrix}
I_{n_I} \\ 0_{n_U \times n_I}
\end{pmatrix}
\right) \delta x^I\\
& = &
\left(
-\bfo_{n_I}\t J_{II}(x^*) + \bfo_{n_I}\t V^I + \bfo_{n_U}\t V^U  \begin{pmatrix}
I_{n_I} \\ 0_{n_U \times n_I}
\end{pmatrix}
\right) \delta x^I.
\end{eqnarray*}

As $\delta x^I$ decreases exponentially, see~\eqref{eq721}, one thus gets that,
\[
| \bfo_{n_U}\t \delta x^U(t)| \leq c' \| \delta x (0)\| \qquad \text{ as long as } \|\delta x (t)\| < \varepsilon,
\]
for some $c'\geq 1$.
Denoting $\tx^U := \begin{pmatrix} x_{n_I+1} & \dots & x_{n-1} \end{pmatrix}\t$
and $\delta \tx^U := \tx^U - \tx^{*U}$, the evolution of $\delta \tx^U$ is then governed by the remaining eigenvalues of the matrix $V^U \begin{pmatrix} 0_{n_I \times n_U} \\ I_{n_U} \end{pmatrix}$ mentioned in Assumption~\ref{as3}.
The latter have negative real parts, as noticed in the beginning of the present demonstration.
Therefore, for some $c''\geq 1$,
\[
\| \delta \tx^U (t)\| \leq c''e^{-kt} \| \delta x (0)\| \qquad \text{ as long as } \|\delta x (t)\| < \varepsilon.
\]

Take $\|\delta x (0)\| < \min \left\{ \zeta; \frac{\varepsilon}{5c}; \frac{\varepsilon}{5c'}; \frac{\varepsilon}{5c''} \right\}$, then as long as $ \|\delta x (t)\| < \varepsilon$, one deduces
\[
\| \delta x^I (t)\| < \frac{\varepsilon}{5},\qquad | \bfo_{n_U}\t \delta x^U(t)| < \frac{\varepsilon}{5}, \qquad \| \delta \tx^U (t)\| < \frac{\varepsilon}{5}.
\]
Using for example for $\|\cdot\|$ the $l^1$-norm, one has
\[
|  \delta x_n (t)| = | \bfo_{n_U}\t \delta x^U(t) - \bfo_{n_U-1}\t \delta \tx^U(t)| \leq | \bfo_{n_U}\t \delta x^U(t)| + |\bfo_{n_U-1}\t \delta \tx^U(t)| \leq | \bfo_{n_U}\t \delta x^U(t)| + \| \delta \tx^U (t)\| \leq \frac{2\varepsilon}{5},
\]
so that
\[
\|\delta x (t)\| = \| \delta x^I (t)\| + \| \delta x^U (t)\| \leq \| \delta x^I (t)\| +\| \delta \tx^U (t)\| + \| \delta x_n (t)\| < \frac{4\varepsilon}{5} < \varepsilon.
\]
Therefore $\|\delta x (t)\|$ never reaches the value $\varepsilon$, and the simple stability of the DFE $x^*$ is proved.

Assume now more specifically that the first formula in \eqref{eq188} holds with an equality, i.e.
\[
\bfo_{n_I}\t V^I + \bfo_{n_U}\t V^U  \begin{pmatrix}
I_{n_I} \\ 0_{n_U \times n_I}
\end{pmatrix}
= 0_{1\times n_I}.
\]
Then along any trajectory of system~\eqref{eq0}-\eqref{eq300}, one has (see Section~\ref{se2}) $\bfo_n\t \dot x \equiv 0$, so that $\bfo_n\t x(t) \equiv \bfo_n\t x(0)$.
Replacing e.g.~the coordinate $x_n(\cdot)$ by its value $\sum_{j=1}^n x_j(0)-\sum_{j=1}^{n-1} x_j(\cdot)$ yields an $(n-1)$-dimensional dynamical system, of which the vector $\tx^* := \begin{pmatrix} x_1^* & \dots & x_{n-1}^* \end{pmatrix}$ is an equilibrium.
The spectrum of the Jacobian matrix of this reduced dynamical system at this point consists of the $n-1$ eigenvalues of $J(x^*)$ with negative real parts.
The point $\tx^*$ is thus a locally asymptotically stable equilibrium of the reduced system, and $x^*$ is therefore a locally asymptotically stable equilibrium of the initial system.
This achieves the demonstration of Theorem~\ref{th1}.
\end{proof}

\begin{remark}
\label{re30}
The function whose Jacobian is considered in Proposition~{\rm\ref{pr1}} is, as in~{\rm \cite[Assumption~{\bf (A.5)}]{Driessche:2002aa}}, equal to the right-and side $f(x,\lambda(x))$ `when the function $\cF^I(x,\lambda(x))$ is set to zero'.
\end{remark}

\begin{remark}
\label{re40}
Consider the following system with {\em two} susceptible compartments, extending the SIR model~\eqref{eq555} by introduction of~inherited heterogeneity of susceptibility and infectivity:
\begin{equation*}
%\label{eq556}
\dot S_i = \mu N_i - S_i \left(
\beta_{ii} \frac{I_i}{S_i} + \beta_{ii'} \frac{I_{i'}}{N_{i'}} 
\right) -\mu S_i,\qquad
\dot I_i = S_i \left(
\beta_{ii} \frac{I_i}{S_i} + \beta_{ii'} \frac{I_{i'}}{N_{i'}} 
\right) -(\gamma+\mu) I_i,\qquad
\dot R_i = \gamma I_i -\mu R_i,
\end{equation*}
with $N_i=S_i+I_i+R_i$, for any $i,i'$ such that $\{i,i'\}=\{1,2\}$.
This model may be written in the preceding setting with $n_I=2$, $n_U=4$.
One verifies that, with the state vector $\begin{pmatrix} I_1 & I_2 & S_1 & S_2 & R_1 & R_2 \end{pmatrix}\t$, one has (compare with~\eqref{eq613})
\[
V^U
= \begin{pmatrix}
\mu & 0 & 0 & 0 & \mu & 0\\
0 & \mu & 0 & 0 & 0 & \mu\\
\gamma & 0 & 0 & 0 & -\mu & 0\\
0 & \gamma & 0 & 0 & 0 & -\mu
\end{pmatrix}.
\]
Corresponding to the conservation of the two quantities $S_i+I_i+R_i$, respectively within the compartments $S_i/I_i/R_i$ for $i=1,2$, one has here
\[
\begin{pmatrix} 1 & 0 & 1 & 0 \end{pmatrix} V^U \begin{pmatrix} 0_{2 \times 4} \\ I_4 \end{pmatrix}
= \begin{pmatrix} 0 &1 & 0 & 1 \end{pmatrix} V^U \begin{pmatrix} 0_{2 \times 4} \\ I_4 \end{pmatrix}
= 0_4\t.
\]
These identities imply, but are stronger than the eigenvector property in the first point of Assumption~\ref{as3}, namely
%the second formula in~\eqref{eq188}, namely
\[
\bfo_4\t V^U \begin{pmatrix} 0_{2 \times 4} \\ I_4 \end{pmatrix}
= 0_4\t.
\]
Here the eigenvalue $0$ of matrix $V^U \begin{pmatrix} 0_{2 \times 4} \\ I_4 \end{pmatrix}$ has thus multiplicity (at least) $2$, so that Assumption~\ref{as3} is {\em not} fulfilled and Theorem~{\rm\ref{th1}} does not apply directly.
However, using the two conservation formulas of the quantities $N_i$, one adapts easily the demonstration, showing that two components of $\delta x^U$ are bounded instead of one, while the remaining components converge to zero.
One then recovers the marginal stability property, and indeed the conclusions of Theorem~{\rm\ref{th1}} are still valid.
This provides a substantial extension of this result.

Notice that the same phenomenon appears for the multi-group model presented in Section~{\rm\ref{se52}} and for the vector-host model in Section~{\rm\ref{se55}}.
For the first one, the disease-free equilibrium points may or may not present individuals in every non-infected compartments $S_i, R_i$; and in the second one, they may or may not contain both uninfected hosts $S$ and uninfected vectors $M$.
Generally speaking this subtlety requires special attention when the dimension of the set  $\Xbf_\eq$ of all DFEs is {\em at least equal to $2$.}
\end{remark}

\section{Basic reproduction number of epidemiological models in commuting populations}
\label{se6}

We gather in this section the main results of the paper.
In Section~\ref{se61} is given the basic reproduction number for a partitioned-population model (that is, without commutations).
This serves as an introduction to Section~\ref{se62}, where the general case of commuting populations is treated.

\subsection{Basic reproduction number of stationary partitioned-population models}
%\subsection{Several subpopulations in a single location}
\label{se61}

Consider first the model of partitioned population given in equation~\eqref{eq1}, with a function $f$ as in~\eqref{eq300}.
The corresponding model is stationary, and composed of $|\cP\backslash\cR|$ uncoupled `locations' (classes).

Assuming the Assumptions~\ref{as1} and~\ref{as3} fulfilled, as well as~\ref{as4} for every $\lambda_q$, $q\in\cQ$, it is easy to see that there exist disease-free equilibrium points to this system.
As a matter of fact, equation~\eqref{eq3} provides a description at the level of each location $q\in\cQ=\cP\backslash\cR$, and due to the previous assumptions, the latter admits disease-free equilibrium points $x_q^*$, $q\in\cQ$ (see Theorem~\ref{th1}).
It is then easy to split the global population at this point and allocate proportions $\alpha_px_q^*$, $0 \leq \alpha_p \leq 1$, $\sum_{p\in q} \alpha_p = 1$ of the latter to each subpopulation $x_p$, $p\in q$, in each class $q$ in order to obtain a DFE $\cX^*$.

We will first introduce in Section~\ref{se511} a representation of the state variable for a partitioned-population model.
A difficulty is that, as we aim at obtaining tractable criteria, we can no longer use the intrinsic formulas in~\eqref{eq1} expressed in terms of all subpopulations indistinctly: instead, one must consider the subpopulation classes and define a global ordering of the subpopulations that keep track of their contacts, in order to ease the representation of the global state variable.
Once this done, one may compute the Jacobian matrix of the system (Section~\ref{se512}), and finally obtain in Section~\ref{se513} the analysis results related to the stability of the DFE $\cX^*$.

\subsubsection{Ad hoc representation of the state variable for partitioned-population model}
\label{se511}

We will now proceed to the computation of the Jacobian matrix of equation~\eqref{eq1} at a disease-free equilibrium point.
For this, we will choose a specific representation of the global state variable $\cX$, compatible with the previous setting.
In order to %end up with a matrix representation allowing to
picture the analysis results in terms of {\em matrix} properties, we will have to drop the intrinsic representation~\eqref{eq1}, expressed in terms of {\em subpopulations}, and argue in terms of {\em classes}, ordered in an order having some specified properties.
This complicates notably the writing of the conditions.
However, this complication is more of a notational nature, than of a conceptual one.

Before displaying the new representation, notice that, in any case, the class $q_j\in\cQ$ contains $|q_j|$ distinct subpopulations, so that the total number of subpopulations fulfils the identity
\begin{equation}
\label{eq375}
|\cP| = \sum \left\{
|q|\ :\ q\in \cQ
\right\}.
\end{equation}
Also, due to the fact that $x_q = \sum_{p\in q} x_p$ for any class $q\in\cQ$, one has for any $p,p'\in\cP$;
\begin{equation}
\frac{\partial x_{[p']}}{\partial x_p} = I_n\ \text{ if }\ [p]=[p'],\qquad \frac{\partial x_{[p']}}{\partial x_p} = 0_{n\times n}\ \text{ if }\ [p]\neq [p'],
\end{equation}
independently of the ordering of the components of the vector $\cX$.

For simplicity, we will assume that the $n$-dimensional vectors $x_p$ are ordered within the $n|\cP|$-dimensional vector $\cX$ according to the following principles: all vectors $x^I_p$ are grouped first, then all vectors $x^U_p$; and within the vectors $x^I_p$, resp.~$x^U_p$, the information of the subpopulation of a first class $q_1$ are put first, then of a second class $q_2$, and so on until $q_{|\cQ|}$.
Schematically, this corresponds to writing the state variable $\cX\in\Rset_+^{n|\cP|}$ as
\begin{subequations}
\label{eq519}
\begin{equation}
\label{eq519a}
\cX = \begin{pmatrix}
\cX^I \\ \cX^U
\end{pmatrix},\qquad
\cX^I\in\Rset_+^{n_I|\cP|},\quad \cX^U\in\Rset_+^{n_U|\cP|},
\end{equation}
with
\begin{equation}
\label{eq519b}
\cX^{I\mbox{\tiny\sf T}}
%= \begin{pmatrix}
%x^I_{p_{1, q_1}}& x^I_{p_{2, q_1}} & \dots & x^I_{p_{|q_1|, q_1}} & x^I_{p_{1, q_2}} & x^I_{p_{2, q_2}} & \dots & x^I_{p_{|q_2|, q_2}} & \dots & x^I_{p_{1, q_1}} & x^I_{p_{2, q_1}} & \dots & x^I_{p_{|q_1|, q_1}}
%\end{pmatrix}
= \begin{pmatrix}
x^{I \mbox{\tiny\sf T}}_{q_1,p_1} & \dots & x^{I \mbox{\tiny\sf T}}_{q_1, p_{|q_1|}} & x^{I \mbox{\tiny\sf T}}_{q_2,p_1} & \dots & x^{I \mbox{\tiny\sf T}}_{q_2,p_{|q_2|}} & \dots & x^{I \mbox{\tiny\sf T}}_{q_{|\cQ|},p_1} & \dots  & x^{I \mbox{\tiny\sf T}}_{q_{|\cQ|},p_{| q_{|\cQ|} |}}
%& x^I_{p_{1, q_2}} & x^I_{p_{2, q_2}} & \dots & x^I_{p_{|q_2|, q_2}} & \dots & x^I_{p_{1, q_1}} & %x^I_{p_{2, q_2}} & \dots & x^I_{p_{|q_2|, q_2}} %& x^I_{p_{2, q_1}} & \dots & x^I_{p_{|q_1|, q_1}}
\end{pmatrix},
\end{equation}
\end{subequations}
where the vector $x^I_{q_j, p_i}$, $j=1,\dots, |\cQ|$, $i=1,\dots, |q_j|$, represents the $n_I$ infected compartments of the subpopulation $p_i$ in the class $q_j\in\cQ$.
The vector $\cX^U$ is defined similarly.
Notice that such a decomposition is {\em not} unique, as the ordering of the classes on the one hand, and the ordering of the subpopulations in every class on the other, are not unique.
Once the ordering of the classes and the ordering of each subpopulation in its class have been chosen, then the ordering of the components in vector $\cX$ is the corresponding lexicographic order: for any $q,q'\in\cQ$, $p\in q$, $p'\in q'$, the components $x_{q,p}$ are located before the components $x_{q',p'}$ if and only if $q$ precedes $q'$, or if $q=q'$ and $p$ precedes $p'$.

\begin{remark}
\label{re4}
With the convention exposed above, one may rewrite identity~\eqref{eq375} as
\begin{equation}
\label{eq376}
|\cP| = \sum_{j=1}^{|\cQ|} |q_j|.
%\left\{ |q_j|\ :\ j=1,\dots, |\cQ| \right\}.
\end{equation}
\end{remark}

\subsubsection{Computation of the Jacobian matrix}
\label{se512}

With the convention of representation of the state variable $\cX$ previously exposed, the Jacobian matrix $\cJ(\cX^*)$ at a DFE $\cX^*$ is written by blocks as
%Let us decompose the Jacobian matrix as in~\eqref{eq45}, that is
\begin{equation}
\label{eq47}
\cJ(\cX^*) := \begin{pmatrix} \cJ_{II}(\cX^*) & \cJ_{IU}(\cX^*) \\ \cJ_{UI}(\cX^*) & \cJ_{UU}(\cX^*) \end{pmatrix},
\end{equation}
for some $\cJ_{II}(\cX^*)\in\Rset^{n_I |\cP|\times n_I |\cP|}$, $\cJ_{IU}(\cX^*) \in\Rset^{n_I |\cP|\times n_U |\cP|}$, $\cJ_{UI}(\cX^*) \in\Rset^{n_U |\cP|\times n_I |\cP|}$, $\cJ_{UU}(\cX^*)\in\Rset^{n_U |\cP|\times n_U |\cP|}$.
Let us assess the value of these blocks.

\medskip \noindent $\bullet$
Let $p\in\cP$, then taking inspiration from the computations in the proof of Theorem~\ref{th1}, one has that, at any DFE $\cX^*$, for any $p\in\cP$,
\begin{multline*}
\lefteqn{\partial_{x^I_p} \left. \left[
(F_0 (\lambda_{[p]}(x_{[p]})) + V^I) x^I_p + F_+ (\lambda_{[p]}(x_{[p]})) x^U_p
\right] \right|_{\cX=\cX^*}}\\
= V^I + \left.
\partial_{x^I_p} \left[
F_+(\lambda_{[p]}(x))x^{U*}_p
\right]
\right|_{x=x^*_{[p]}}
= V^I + \left.
\partial_\lambda \left[
F_+(\lambda)x^{U*}_p
\right]
\right|_{\lambda=0_{n_\lambda}}
\left.
\partial_{x^I_p} \lambda_{[p]}(x)
\right|_{x=x^*_{[p]}}.
\end{multline*}
In coherence with the previous notations, $x^{U*}_p$ represents the uninfected components of the state vector of the subpopulation $p$ at equilibrium; and $x^*_{[p]} = \sum_{p'\in[p]} x_{p'}^*$.

\begin{remark}
\label{re3}
As noticed before, the function $\lambda \mapsto \partial_\lambda \left[
F_+(\lambda)x^{U*}_p
\right]$
%\frac{\partial \left[
%F_+(\lambda)x^{U*}_p
%\right]}{\partial\lambda}(\lambda)$
is {\em constant}, as it is the derivative of an affine map.
Accordingly, here and in the sequel we will simply write
\[
\partial_\lambda \left[
F_+(\lambda)x^{U*}_p
\right]
%\frac{\partial \left[
%F_+(\lambda)x^{U*}_p
%\right]}{\partial\lambda}
\qquad \text{ instead of } \qquad
\partial_\lambda \left[
F_+(\lambda)x^{U*}_p
\right]
%\frac{\partial \left[
%F_+(\lambda)x^{U*}_p
%\right]}{\partial\lambda}
(0_{n_\lambda}).
\]
Also, $x^*_{[p]} = \sum_{p'\in[p]} x_{p'}^*$, so that the preceding derivative does not depend upon the subpopulation of the class.
One may thus write
\[
\partial_{x^I_p} \lambda_{[p]}(x^*_{[p]})
= \partial_{x^I} \lambda_{[p]}(x^*_{[p]}),
%\frac{\partial \lambda_{[p]}}{\partial x^I_p}(x^*_{[p]})
%= \frac{\partial \lambda_{[p]}}{\partial x^I}(x^*_{[p]}),
\]
where the variable $x^I_p$ has been replaced by the dummy variable $x^I$.
\end{remark}

On the other hand, for $p,p'\in\cP$ such that $[p]=[p']$ but $p\neq p'$, one has
\begin{multline*}
\lefteqn{\partial_{x^I_{p'}} \left. \left[
(F_0 (\lambda_{[p]}(x_{[p]})) + V^I) x^I_p + F_+ (\lambda_{[p]}(x_{[p]})) x^U_p
\right] \right|_{\cX=\cX^*}}\\
= \left.
\partial_{x^I_{p'}} \left[
F_+(\lambda_{[p]}(x))x^{U*}_p
\right]
\right|_{x=x^*_{[p]}}
= \left.
\partial_{\lambda} \left[
F_+(\lambda)x^{U*}_p
\right]
\partial_{x^I} \lambda_{[p]}(x)
\right|_{x=x^*_{[p]}}
\end{multline*}
(using again Remark~\ref{re3} to introduce a dummy variable).

Last, the partial derivative $\partial_{x^I_{p'}} \left. \left[
(F_0 (\lambda_{[p]}(x_{[p]})) + V^I) x^I_p + F_+ (\lambda_{[p]}(x_{[p]})) x^U_p
\right] \right|_{\cX=\cX^*}$ is equal to $0_{n_I\times n_I}$ whenever $[p]\neq [p']$.

One shows with the same techniques that, for any $p,p'\in\cP$,
\[
\partial_{x^U_{p'}} \left. \left[
(F_0 (\lambda_{[p]}(x_{[p]})) + V^I) x^I_p + F_+ (\lambda_{[p]}(x_{[p]})) x^U_p
\right] \right|_{\cX=\cX^*}
= 0_{n_I\times n_U}.
\]
Last, one has for any $p\in\cP$,
\[
\left.
\partial_{x^U_p} \left(
V^U x_p - F_- (\lambda_{[p]}(x_{[p]})) x^U_p
\right) \right|_{\cX=\cX^*}
= V^U \begin{pmatrix} 0_{n_I \times n_U} \\ I_{n_U} \end{pmatrix},
\]
while for any $p,p'\in\cP$ such that $p\neq p'$,
\[
\left.
\partial_{x^U_{p'}} \left(
V^U x_p - F_- (\lambda_{[p]}(x_{[p]})) x^U_p
\right) \right|_{\cX=\cX^*}
= 0_{n_U \times n_U}.
\]

\medskip \noindent $\bullet$
With the help of the previous computations and adopting now the component ordering defined in~\eqref{eq519}, one may proceed to compute the three blocks $\cJ_{II}(\cX^*)$, $\cJ_{IU}(\cX^*)$ and $\cJ_{UU}(\cX^*)$ in the decomposition~\eqref{eq47}.
The matrix $\cJ_{II}(\cX^*)$ is defined as
\begin{subequations}
\label{eq78}
\begin{equation}
\label{eq78a}
\cJ_{II}(\cX^*) = \diag\left\{ \cB_{q_1}(\cX^*); \dots; \cB_{q_{|\cQ|}}(\cX^*)\right\},
\end{equation}
where, for any $q\in\cQ$, the block  $\cB_q\in\Rset^{n_I|q|\times n_I|q|}$, is given by
%$j=1,\dots, |\cQ|$, $\cB_j\in\Rset^{n_I|q_j|\times n_I|q_j|}$, is given by
\begin{equation}
\label{eq78b}
\cB_q := (I_{|q|}\otimes V^I)
+ \begin{pmatrix}
\partial_{x^I} \left[
F_+(\lambda(x))x^{U*}_{q,1}
\right](x^*_q)\\
\vdots \\
\partial_{x^I} \left[
F_+(\lambda(x))x^{U*}_{q,|q|}
\right](x^*_q)
\end{pmatrix} (\bfo_{|q|}\t \otimes I_{n_I}).
%\cB_j := (I_{|q_j|}\otimes V^I) + \begin{pmatrix}
%\frac{\partial \left[
%F_+(\lambda)x^{U*}_{q_j,1}
%\right]}{\partial\lambda} \frac{\partial \lambda_{q_j}}{\partial x^I}(x^*_{q_j})\\
%\vdots\\
%\frac{\partial \left[
%F_+(\lambda)x^{U*}_{q_j,|q_j|}
%\right]}{\partial\lambda} \frac{\partial \lambda_{q_j}}{\partial x^I}(x^*_{q_j})
%\end{pmatrix} (\bfo_{|q_j|}\t \otimes I_{n_I}).
\end{equation}
In coherence with the previous notations,~in \eqref{eq78b} $x^{U*}_{q,i}$, $i=1,\dots, |q|$, is the equilibrium value of the subpopulation numbered $i$ in the $q$ class/location; and $x^*_q = \sum \left\{  x^*_p\ : p\in\cP,\ [p] = q \right\} = \sum_{i=1}^{|q|} x^*_{q,i}$ is the total population present in that class.
One also has:
\begin{gather}
\label{eq78c}
\cJ_{IU}(\cX^*) = 0_{n_I |\cP|\times n_U |\cP|},\\
\label{eq78d}
\cJ_{UU}(\cX^*) = I_{|\cP|} \otimes V^U \begin{pmatrix} 0_{n_I \times n_U} \\ I_{n_U} \end{pmatrix},
\end{gather}
where $\otimes$ represents the Kronecker product.
It is important to notice at this point that the 4th block $\cJ_{UI}(\cX^*)$ has a block diagonal structure, with $|\cQ|$ blocks, respectively of size $|q_1|, \dots, |q_{|\cQ|}|$.
\end{subequations}

Let us comment on that structure of the Jacobian matrix.
First, formula~\eqref{eq78a} shows a block-diagonal structure that corresponds to the class structure: there may be no influence whatsoever between two subpopulations whose classes (or locations) are different.
Then~\eqref{eq78b} unveils the interinfluence between the infected populations at the same location.
The first term therein, namely $(I_{|q|}\otimes V^I)$, originates from the passage from a compartment to another within every given subpopulation $x_p$; while the second one --- the product of two matrices of respective dimensions $|q| n_I\times n_I$ and $n_I \times |q| n_I$ --- comes from the interinfluence between the different subpopulations in the same class $q$, which results from the fact that the force of infection $\lambda$ depends upon the total population $x_q = \sum_{p\in q} x_p$ at this location.
The $n_I$ lines of this second term that correspond to a given subpopulation $p_{q,i}$, $i\in\{1,\dots, |q|\}$, are
\[
\partial_{x^I} \left[
F_+(\lambda(x))x^{U*}_{q,i}
\right](x^*_q) (\bfo_{|q|}\t \otimes I_{n_I}).
\]
They depend {\em only} upon two quantities: on the one hand upon the equilibrium value $x^{U*}_{q,i}$ of this subpopulation; and on the other hand upon the global sum $x^{U*}_q-x^{U*}_{q,i}$ of the other subpopulations in the class $q$ (and not individually of every other subpopulation in the class $q$).

Formula~\eqref{eq79} is the analogue of~\eqref{eq78c} in the simple case of a unique subpopulation.
It is a consequence of the fact that the equilibrium considered is disease-free.
Last, formula~\eqref{eq78d} testifies of the same type of diagonal structure than the first term in~\eqref{eq78b}, which  gives in $\cJ_{II}(\cX^*)$ a term $(I_{|\cP|} \otimes V^I)$ similar to~\eqref{eq78d}.

%{\PABD As can be read from the formulas~\eqref{eq78}, \fbox{???}}

\subsubsection{DFE stability for a partioned-population model}
\label{se513}

The following result treats the case of several subpopulations distributed in different locations in a permanent manner.
This extension of Theorem~\ref{th1} is essentially a rewriting, due to the fact that the partition of the subpopulations is unchanged along the time.
It is put here for didactic reasons, as an intermediate towards the fully general case where the partition changes with respect to time, treated below in Theorem~\ref{th5}.

%analogous to Theorem~\ref{th1} in the case of several subpopulations in the same location.
%Of course, this case is indeed {\em equivalent to} the case of a single population.
%Therefore, it is rather a rewriting of Theorem~\ref{th1} than a real extension.
%It is put here only for didactic reasons, as an intermediate towards the fully general case treated below in Theorem~\ref{th5}.

\begin{theorem}[Basic reproduction number and stability for epidemic models with fixed population partition]
\label{th2}

Let $\cP$ be a finite set, $\cR$ a partition of $\cP$ and $\cQ := \cP\backslash\cR$ the corresponding set of equivalence classes.
Let $f$ %(decomposed according to~\eqref{eq177})
and $\lambda_q$, $q\in\cQ$, fulfilling Assumptions \ref{as1}, \ref{as3}, \ref{as4}.
Then the following properties hold.

\begin{itemize}
\item
System \eqref{eq1}-\eqref{eq300} admits nonzero disease-free equilibrium points.
The set of equilibrium points of this system is exactly the Cartesian product $\Xbf_{\eq}^{|\cP|}$, where $\Xbf_{\eq}$, defined in~\eqref{eq20}, is the set of equilibrium points for an isolated subpopulation.

\item
For any nonzero disease-free equilibrium point $\cX^*\in\Xbf_{\eq}^{|\cP|}$ of system \eqref{eq1}-\eqref{eq300}, let
\begin{subequations}
\label{eq305}
\begin{equation}
\label{eq305a}
\cR_{0,q} := \rho \left(F_q
 V^{-1}
\right),\qquad q\in\cQ,
\end{equation}
where $F_q, V\in\Rset^{n_I\times n_I}$ are given by:
\begin{gather}
\label{eq305b}
\hspace{-1cm}
F_q := \left.
\partial_{x^I} \left(
F_+(\lambda_q(x))x^{U*}_q
\right)
\right|_{x=x^*_q},\qquad
%= \left.
%\frac{\partial \left[
%F_+(\lambda)x^{U*}_q
%\right]}{\partial\lambda} \frac{\partial \lambda_q(x)}{\partial x^I}
%\right|_{x=x^*_q}
%= \sum_{i=1}^{n_\lambda} %\left.
%\frac{\partial [F_+(\lambda)x^{U*}_q]}{\partial \lambda_i}
%%\right|_{\lambda= 0_{n_\lambda}}
%\left.
%\frac{\partial \lambda_{q,i}(x)}{\partial x^I}
%\right|_{x=x^*_q},\\
%\label{eq305c}
V := V^I.
\end{gather}
\end{subequations}

Let
\begin{equation}
\label{eq59}
\cR_0 := \max_{q\in\cQ} \cR_{0,q}.
\end{equation}

The following assertions are true:
\begin{itemize}
\item
if $\cR_0>1$, then $\cX^*$ is unstable;
\item
if $\cR_0<1$, then %$\Xbf_{\eq}$ is attracting and
$\cX^*$ is marginally stable.
Moreover, if the first identity in \eqref{eq188} holds {\em with an equality}, then $\cX^*$ is locally asymptotically stable.
\end{itemize}
\end{itemize}
\end{theorem}

Recall that the notation $\partial_{x^I}$ is defined in~\eqref{eq195}.

While apparently complicated, the result enunciated in Theorem~\ref{th2} is in fact quite natural.
On the one hand, formula~\eqref{eq59} states that the basic reproduction number of the system is the largest of the basic reproduction numbers of the (sub-)systems describing the infection spread in each independent class (location) $q\in\cQ$.
On the other hand, due to perfect mixing within every class, formula~\eqref{eq305} states that this value may be computed for any class $q\in\cQ$ as the basic reproduction number corresponding to a {\em unique} population at the equilibrium point $x_q^* = \sum_{p\in q} x_p^*$.
In this respect, it is important to notice that  $F_q$ in~\eqref{eq305b} is an $n_I\times n_I$ matrix, while the block $\cB_q$ corresponding to the class $q$ of the Jacobian matrix defined in~\eqref{eq78b} is a $n_I|q| \times n_I|q|$ matrix.

Notice that when $\cQ$ contains a unique class gathering all subpopulations, then Theorem~\ref{th2} reduces to the case, treated in Theorem~\ref{th1}, of a unique subpopulation.
This is a consequence of the fact that what matters to portray the asymptotic evolution at the scale of a class $q\in\cQ$, is exclusively the total population $x_q$ present therein.

\begin{proof}[Proof of Theorem~\rm\ref{th2}]
\mbox{}

\noindent $\bullet$
Due to~\eqref{eq1}, at equilibrium, every subpopulation $p\in\cP$ evolves uncoupled, according to
\[
\dot x_p = f(x_p,0_{n_\lambda}).
\]
It is therefore clear that the set of equilibrium points of~\eqref{eq1}-\eqref{eq300} is exactly the Cartesian product $\Xbf_{\eq}^{|\cP|}$.

\noindent $\bullet$
Let us now consider the Jacobian matrix.
As $\cJ_{IU}(\cX^*)$ is null, one  has
\[
s(\cJ(\cX^*)) = \max\{s(\cJ_{II}(\cX^*)),s(\cJ_{UU}(\cX^*)) \},
\]
and one is led to consider the two diagonal blocks $\cJ_{II}(\cX^*)$ and $\cJ_{UU}(\cX^*)$.
We will exploit the block diagonal structure of these two matrices.

\noindent $\bullet$
Let us compute first the value of the dominant eigenvalues of the diagonal blocks $\cB_q$ of $\cJ_{II}(\cX^*)$, $q\in\cQ$.
% basic reproduction numbers of the subsystems (that is, formula~\eqref{eq59}) comes from the block diagonal structure of all four blocks $\cJ_{II}(\cX^*), \cJ_{IU}(\cX^*), \cJ_{UI}(\cX^*), \cJ_{UU}(\cX^*)$.
For this, let us consider a subpopulation class $q$.
We will assess the dominant eigenvalue of the corresponding matrix $\cB_q$ and show that it is equal to $\cR_{0,q}$ defined in~\eqref{eq305a}.
%Let us establish the second part, formula~\eqref{eq355}.
%For some $v_i\in\Rset_+^{n_I}$, $i=1,\dots,|q_j|$, l
Let $\nu\in\Rset$ and the vector
\[
\begin{pmatrix}
v_1 \\ \vdots \\ v_{|q|}
\end{pmatrix} \in \Rset_+^{|q|n_I}
\]
be respectively the dominant eigenvalue and a positive right-eigenvector of the matrix $\cB_q$ defined in~\eqref{eq78b}.
In other words,
%
%, such that the vector
%\[
%\begin{pmatrix}
%v_1 \\ \vdots \\ v_{|q_j|}
%\end{pmatrix}
%\]
%
%and $\nu\in\Rset$ be respectively a positive left-eigenvector and the dominant eigenvalue of the matrix
%\[
%\begin{pmatrix}
%\frac{\partial \left[
%F_+(\lambda)x^{U*}_1
%\right]}{\partial\lambda} \frac{\partial \lambda_{q}}{\partial x^I}(x^*_{q})\\
%\vdots\\
%\frac{\partial \left[
%F_+(\lambda)x^{U*}_{|q|}
%\right]}{\partial\lambda} \frac{\partial \lambda_{q}}{\partial x^I}(x^*_{q})
%\end{pmatrix} (\bfo_{|q|}\t \otimes V^{I -1}).
%\]
\begin{equation}
\label{eq381}
\left(
(I_{|q|} \otimes V^I )
+ \begin{pmatrix}
\partial_{x^I} \left[
F_+(\lambda_q(x))x^{U*}_{q,1}
\right](x^*_q)\\
\vdots \\
\partial_{x^I} \left[
F_+(\lambda_q(x))x^{U*}_{q,|q|}
\right](x^*_q)
\end{pmatrix} (\bfo_{|q|}\t \otimes I_{n_I})
%\begin{pmatrix}
%\frac{\partial \left[
%F_+(\lambda)x^{U*}_{q,1}
%\right]}{\partial\lambda} \frac{\partial \lambda_q}{\partial x^I}(x^*_q)\\
%\vdots\\
%\frac{\partial \left[
%F_+(\lambda)x^{U*}_{q,|q|}
%\right]}{\partial\lambda} \frac{\partial \lambda_q}{\partial x^I}(x^*_q)
%\end{pmatrix}
%(\bfo_{|q|}\t \otimes I_{n_I})
\right)
\begin{pmatrix}
v_1 \\ \vdots \\ v_{|q|}
\end{pmatrix}
= \nu \begin{pmatrix}
v_1 \\ \vdots \\ v_{|q|}
\end{pmatrix}.
\end{equation}

The map $x^U\mapsto F_+(\lambda)x^U$ is linear.
Therefore, due to the fact that $x_q^* = x_{q,1}^*+\dots+x_{q,|q|}^*$, so that $x_q^{U*} = x_{q,1}^{U*}+\dots+x_{q,|q|}^{U*}$, one has
\[
(\bfo_{|q|}\t \otimes I_{n_I} )
\begin{pmatrix}
\partial_{x^I} \left[
F_+(\lambda_q(x))x^{U*}_{q,1}
\right](x^*_q)\\
\vdots \\
\partial_{x^I} \left[
F_+(\lambda_q(x))x^{U*}_{q,|q|}
\right](x^*_q)
\end{pmatrix}
%\begin{pmatrix}
%\frac{\partial \left[
%F_+(\lambda)x^{U*}_{q,1}
%\right]}{\partial\lambda} \frac{\partial \lambda_q}{\partial x^I}(x^*_q)\\
%\vdots\\
%\frac{\partial \left[
%F_+(\lambda)x^{U*}_{q,|q|}
%\right]}{\partial\lambda} \frac{\partial \lambda_q}{\partial x^I}(x^*_q)
%\end{pmatrix}
= \sum_{i=1}^{|q|}
\partial_{x^I} \left[
F_+(\lambda_q(x))x^{U*}_{q,i}
\right](x^*_q)
= \partial_{x^I} \left[
F_+(\lambda_q(x))x^{U*}_q
\right](x^*_q).
\]

The emergence of an epidemic in the class $q$ depends only upon the total population present therein and upon the local conditions of the infection, but not upon the division of this class in several subpopulations: the stability of the DFE is indeed a property of the class itself, independent of its constituting subpopulations.
In coherence with this intuition, we now exhibit a dominant eigenvector at the level of the whole population.
Using the previous computations, left-multiplying formula~\eqref{eq381} by the matrix $(\bfo_{|q|}\t \otimes I_{n_I} ) \in\Rset^{n_I\times |q|n_I}$ yields the identity:
\[
\left(
(\bfo_{|q|}\t \otimes V^I )
+
\partial_{x^I} \left[
F_+(\lambda_q(x))x^{U*}_q
\right](x^*_q)
(\bfo_{|q|}\t \otimes I_{n_I})
\right)
\begin{pmatrix}
v_1 \\ \vdots \\ v_{|q|}
\end{pmatrix}
= \nu (\bfo_{|q|}\t \otimes I_{n_I} ) \begin{pmatrix}
v_1 \\ \vdots \\ v_{|q|}
\end{pmatrix},
\]
that is, for any $q\in\cQ$,
\[
\left(
V^I +
\partial_{x^I} \left[
F_+(\lambda_q(x))x^{U*}_q
\right](x^*_q)
%\frac{\partial \left[
%F_+(\lambda)x^{U*}_q
%\right]}{\partial\lambda} \frac{\partial \lambda_q}{\partial x^I}(x^*_q)
\right)
v_q
= \nu\, v_q,\qquad
\text{ with }\
v_q := v_1 + \dots + v_{|q|}.
\]
The vector $v_q$ is a positive eigenvector, it is therefore associated to the dominant eigenvalue $\nu$.
The matrix involved is the same than the one appearing in Theorem~\ref{th1}, and the same handling yields the formulas in~\eqref{eq305}, similar to~\eqref{eq355}.
The basic reproduction number of the system then appears as the largest of the dominant eigenvalues of the diagonal blocks $\cB_q$ of $\cJ_{II}(\cX^*)$, $q\in\cQ$.

\noindent $\bullet$
Let us now consider the matrix $\cJ_{UU}(\cX^*)$, expressed in~\eqref{eq78d}.
Each of the $|\cP|$ blocks possesses a unique eigenvector associated to the eigenvalue $0$, therefore $0$ is an eigenvalue of $\cJ_{UU}(\cX^*)$ whose algebraic and geometrical multiplicities are both equal to $|\cP|$.
The evolution in each class of $\cQ$ is independent of what happens in the other ones, and conducting the same analysis than in the proof of Theorem~\ref{th1} gives the same stability result, depending on whether $\cR_0$ given in~\eqref{eq59} is smaller or larger than 1.
This achieves the proof of Theorem~\ref{th2}.
\end{proof}

\subsection{Basic reproduction number of general epidemiological models in commuting populations}
\label{se62}

We study now the behaviour of the solutions of an epidemiological model in commuting population given by~\eqref{eq2}-\eqref{eq300}, for a $T$-periodic admissible switching signal $\cR(\cdot)$ with interswitch durations in a given set $\Lambda_m$ defined in~\eqref{eq58}, assuming Assumptions~\ref{as1},~\ref{as3},~\ref{as4} fulfilled.

We first show in Section~\ref{se621} that the disease-free periodic solutions are indeed equilibrium points.
%, as in Section~\ref{se61}.
We then extend in Section~\ref{se622} the framework introduced in Section~\ref{se511} for representing the state variable, and state the main result in Section~\ref{se523}, together with an important corollary.
Section~\ref{se5235} provides a method to compute more easily the quantities involved.
Finally a simple example is displayed and analyzed in Section~\ref{se524}, in order to make visible the different steps.

\subsubsection{Equilibrium states of a periodic commuting population model}
\label{se621}

We are interested in modelling and analysing the effects of periodic `commutations'.
The latter are specific variations in time of the considered system.
Having a more acute look at the formulas~\eqref{eq2} with $T-$periodic switching signal, one sees that {\em the time-variation only affects the definition of the (piecewise constant) force of infection}, namely the subpopulation mixing on each time interval.

Consider a disease-free evolution of system~\eqref{eq2}-\eqref{eq300}, that is one for which $x^I(t) = 0_{n_I}$ and $\lambda_{[p]_{\cR(t)}}(t) = 0_{n_\lambda}$ for any $t\in\cI$.
Seen from the subpopulation $p\in\cP$, this implies (see formula~\eqref{eq333})
\[
\dot x_p = f(x_p(t),0_{n_\lambda}) = \begin{pmatrix}  0_{n_I}\\ V^U x_p \end{pmatrix},\qquad t\in\cI.
\]
The only changes considered here are the commutations.
The latter modify the mixing of the subpopulations, and possibly the value of the components of the force of infection.
However, the structure of the model is not touched, and in particular the non-infective passages from a compartment to another are not modified along time.
In consequence, contrary e.g.~to what happens with the settings developed to analyse seasonal phenomena (see for example~the papers~\cite{Bacaer:2007aa,Bacaer:2007ab,Wang:2008aa}), the disease-free attractors are all the equilibrium points in the space $\Rset_+^{|\cP|n}$.
%Bacaer:2009aa
%, and not more general periodic solutions.
See more details in the first part of Theorem~\ref{th5}.

%\cite{Bacaer:2007aa,Bacaer:2007ab,Bacaer:2009aa,Berg:2011aa,Bacaer:2012aa}

%= \begin{pmatrix} 0_{n_I \times n_U} \\ I_{n_U} \end{pmatrix} V^U \begin{pmatrix} 0_{n_I \times n_U} \\ I_{n_U} \end{pmatrix} x_p^U.

\subsubsection{State variable and Jacobian matrix representation for commuting population model}
\label{se622}

%\subsubsection{\PABB Ad hoc representation of the state variable for a commuting population model}

We extend here the notational setting introduced in Section~\ref{se511}, in order to accommodate the changes of partition from a maximal subinterval of $\cI$ to the next one that model the commutations.
Denote $\cR(\cI_k)$ the partition in force during any maximal subinterval of $\cI_k$.
The class of subpopulation $p$ during this period of time will be denoted accordingly $[p]_{\cR(\cI_k)}$.

For simplicity we assume the existence of a `reference ordering of the subpopulations', pre-existing the infection; and, for any union of intervals $\cI_k$, $k=1,\dots,m$, of another lexical ordering such as the ones presented in Section~\ref{se511}, having the property of attributing contiguous ranks to the subpopulations located in the same class.
Denote $\cM_k\in\Rset^{|\cP|\times|\cP|}$ the permutation matrix allowing to move from the reference ordering to the ordering attached to $\cI_k$.
In other words, the subpopulation with rank $i$ in the initial ordering, is ranked $i'$ in the new one, such that $\cM_k e_i = e_{i'}$ during $\cI_k$, where $e_i, e_{i'}$ are respectively the $i$-th and $i'$-th vectors of the canonical basis of $\Rset^{|\cP|}$.

Consider a given DFE $\cX^*\in\Rset_+^{n|\cP|}$.
The Jacobian matrix of system~\eqref{eq2}-\eqref{eq300} is constant on any interval in $\cI_k$.
Referring to formula~\eqref{eq47}, its value when expressed according to the subpopulation ordering adapted to the partition $\cR(\cI_k)$, will be denoted $\cJ_k(\cX^*)$, and the corresponding nonzero blocks $\cJ_{II,k}(\cX^*), \cJ_{UI,k}(\cX^*), \cJ_{UU,k}(\cX^*)$.
{\em In the reference ordering,} the Jacobian matrix $\hcJ_k(\cX^*)$ is then equal to
\begin{subequations}
\label{eq599}
\begin{gather}
\label{eq599a}
\hcJ_k(\cX^*)
= \begin{pmatrix} \hcJ_{II,k}(\cX^*) & \hcJ_{IU,k}(\cX^*) \\ \hcJ_{UI,k}(\cX^*) & \hcJ_{UU,k}(\cX^*) \end{pmatrix},\\
\label{eq599b}
\hcJ_{\eta\eta',k}(\cX^*) := (\cM_k\t \otimes I_{n_\eta}) \cJ_{\eta\eta',k}(\cX^*) (\cM_k\otimes I_{n_{\eta'}}),\qquad
\eta,\eta'\in\{I,U\}.
%\begin{pmatrix} (\cM_k\t \otimes I_{n_I}) \cJ_{II,k}(\cX^*) (\cM_k\otimes I_{n_I}) & 0_{n_I|\cP|\times n_U|\cP|} \\ (\cM_k\t\otimes I_{n_U}) \cJ_{UI,k}(\cX^*)(\cM_k\otimes I_{n_I})  & (\cM_k\t\otimes I_{n_U}) \cJ_{UU,k}(\cX^*) (\cM_k\otimes I_{n_U}) \end{pmatrix}.
\end{gather}
\end{subequations}
In particular, one has $\hcJ_{IU,k}(\cX^*)=\cJ_{IU,k}(\cX^*)=0_{n_I|\cP|\times n_U|\cP|}$ at any equilibrium point $\cX^*$.

\subsubsection{DFE stability for a commuting population model}
\label{se523}

With the previous notations, it is now possible to state the main result of the paper.

\begin{theorem}[Basic reproduction number and stability for epidemic models in commuting populations]
%\begin{theorem}[Basic reproduction number and stability for epidemic models with several subpopulations in a single location]
\label{th5}
Let $\cP$ be a finite set, $\cR(\cdot)$ a $T$-periodic admissible switching signal with set of switching times $\Rset_+\setminus\cI$ according to pattern \eqref{eq58}, and $\cQ(t) := \cP\backslash\cR(t)$ the corresponding set of equivalence classes at time $t\in\cI$.
%Assume $(\cP,\cR(\cdot))$ fulfils Assumption~\ref{as8}.
%Assume $\cP$ is strongly sequentially connected by $\cR(\cdot)$.
Let $f$ be given as in~\eqref{eq300}
%(decomposed according to~\eqref{eq177})
and $\lambda_q$, $q\in\cQ(\cI_k)$, $k=1,\dots,m$, fulfilling Assumptions \ref{as1}, \ref{as3}, \ref{as4}.
Then the following properties hold.

\begin{itemize}
\item
The disease-free $T$-periodic solutions of system \eqref{eq2}-\eqref{eq300} are exactly the equilibrium points of the Cartesian product $\Xbf_{\eq}^{|\cP|}$, where $\Xbf_{\eq}$, defined in~\eqref{eq20}, is the set of equilibrium points for an isolated subpopulation.
%System \eqref{eq2}-\eqref{eq300} admits nonzero disease-free $T$-periodic solutions.
%All of them are equilibrium points.
%The set of equilibrium points of this system is exactly the Cartesian product $\Xbf_{\eq}^{|\cP|}$, where $\Xbf_{\eq}$, defined in~\eqref{eq20}, is the set of equilibrium points for an isolated subpopulation.
\item
For any nonzero disease-free equilibrium point $\cX^*\in\Xbf_{\eq}^{|\cP|}$ of system \eqref{eq2}-\eqref{eq300}, let
\begin{equation}
\label{eq559}
\cR_0 := \rho\left(
e^{\hcJ_{II,m}(\cX^*)(\tau_m-\tau_{m-1})} \dots\ e^{\hcJ_{II,2}(\cX^*)(\tau_2-\tau_1)}\ e^{\hcJ_{II,1}(\cX^*)(\tau_1-\tau_0)}
\right),
\end{equation}
for the $|\cP|n_I\times |\cP|n_I$ Jacobian matrices $\hcJ_{II,k}(\cX^*)$, $k=1,\dots,m$, defined in~\eqref{eq599}.
The following assertions are true:
\begin{itemize}
\item
if $\cR_0>1$, then $\cX^*$ is unstable;
\item
if $\cR_0<1$, then %$\Xbf_{\eq}$ is attracting and
$\cX^*$ is marginally stable.
Moreover, if the first identity in \eqref{eq188} holds {\em with an equality}, then $\cX^*$ is locally asymptotically stable.
\end{itemize}
\end{itemize}
\end{theorem}

Recall that the notion of admissible switching signal and the notation $\cI$ have been defined in Definition~\ref{de2}.
Computing exactly $\cR_0$ involves computing the spectral radius of a matrix of size $|\cP|n_I\times |\cP|n_I$.

\begin{remark}

The computation of $\cR_0$ by~\eqref{eq559} is made easier in the case where, in spite of the commutations, the infection evolves independently in completely disjoint subgroups of the population.
This is indeed the case if for some subpopulations $p,p'\in\cP$, there is no positive integer $l$ and finite sequence $p_0, p_1, \dots, p_l$ of elements of $\cP$ such that $p_0 := p$, $p_l:=p'$ and for any $i=0,\dots, l-1$, there exists $k\in\{1,\dots,m\}$ such that
\[
[p_i]_{\cR(\cI_k)} = [p_{i+1}]_{\cR(\cI_k)}.
\]
(Alternatively, this means that the transitive closure of the union of the relations $\cR(\cI_k)$, $k=1,\dots,m$, has at least two distinct classes of equivalence.)
When such disconnection occurs, the infection initially present in the subpopulation $p$ cannot move to the subpopulation $p'$ and vice versa, and the matrix product in~\eqref{eq559} involves several diagonal blocks.
As is the case for Theorem~{\rm\ref{th2}}, the value of $\cR_0$ then appears as the maximal value of the corresponding quantities computed in each connected component of $\cP$ (that is, on each class of equivalence of the transitive closure of the union of the relations $\cR(\cI_k)$, $k=1,\dots,m$).
%In order to simplify the statement of the stability result, we will introduce the following connectedness assumption of the switching pattern.
%
%\begin{assumption}[On the $T$-periodic admissible switching signal $\cR(\cdot)$]
%\label{as8}
%The set $\cP$ is said {\em sequentially connected} by a $T$-periodic admissible switching signal $\cR(\cdot)$ on $\cP$  if for any subpopulations $p,p'\in\cP$, there exist a positive integer $l$ and a finite sequence $p_0, p_1, \dots, p_l$ of elements of $\cP$ such that $p_0 = p$, $p_l:=p'$ and for any $i=0,\dots, l-1$, there exists $k\in\{0,\dots,m\}$ such that
%\[
%[p_i]_{\cR(\cI_k)} = [p_{i+1}]_{\cR(\cI_k)}.
%\]
%\end{assumption}
%It is straightforward to show that when Assumption~\ref{as8} is not fulfilled in the results to come, then the infection cannot move from some subpopulations to some others.
%In such cases, the stability analysis has to be achieved independently on every strongly sequentially connected component of $\cP$.
\end{remark}

\begin{proof}[Proof of Theorem~\rm\ref{th5}]
\mbox{}

\noindent $\bullet$
Let us first consider a disease-free trajectory for~\eqref{eq2}-\eqref{eq300}.
As observed previously, the coupling between equations describing the evolution of different subpopulations is manifested only through the infection process: due to~\eqref{eq2}, every subpopulation $p\in\cP$ evolves according to
\[
\dot x_p = f(x_p,0_{n_\lambda}).
\]
It is therefore clear that the set of equilibrium points of~\eqref{eq2}-\eqref{eq300} is exactly the Cartesian product $\Xbf_{\eq}^{|\cP|}$.

More precisely, for any disease-free periodic trajectory, any subpopulation $p\in\cP$ undergoes a stationary evolution, governed (see~\eqref{eq300c}) by the evolution
\[
\dot x^U_p = V^U \begin{pmatrix} 0_{n_I \times n_U} \\ I_{n_U} \end{pmatrix} x^U_p.
\]
By Assumption~\ref{as3}, the matrix present in this formula has 0 as simple eigenvalue, and all other eigenvalues have negative real parts.
Therefore, every trajectory converges towards an equilibrium point, and no other periodic solution exists.

\noindent $\bullet$
As established in Section~\ref{se622}, the Jacobian matrix of the system is constant on every maximal subinterval of $\cI$, and equal to  $\hcJ_k(\cX^*)$ on any $\cI_k$, $k=1,\dots,m$.
Using the block-diagonal structure and arguments identical to the demonstration of Theorem~\ref{th2}, one shows the marginal stability of the uninfected components; while the stability of the infected components depends upon the position  with respect to 1 of the spectral radius of the monodromy matrix.
The latter is exactly $\cR_0$.
This achieves the proof of Theorem~\ref{th5}.
\end{proof}

We now present an important corollary of Theorem~{\rm\ref{th5}}.
The following result shows that, when the conditions of transmission are identical in the whole population (in a sense made precise in the statement), then the disease evolution is identical to that of a perfectly mixed population gathering all the subpopulations.

\begin{corollary}[Homogeneous density-independent transmission conditions]
\label{co0}
Under the conditions of Theorem~{\rm\ref{th5}}, let $\cX^*\in\Xbf_{\eq}^{|\cP|}$ be a nonzero disease-free equilibrium point of system \eqref{eq2}-\eqref{eq300} such that
\begin{equation}
\label{eq789}
\forall k\in\{1,\dots, m\},\quad \forall q,q'\in\cQ(\cI_k),\quad
\partial_{x^I} \left[
F_+(\lambda_q(x))x^{U*}_q
\right](x^*_q)
= \partial_{x^I} \left[
F_+(\lambda_{q'}(x))x^{U*}_{q'}
\right](x^*_{q'})
%= \frac{\partial \left[
%F_+(\lambda)x^{U*}_{q'}
%\right]}{\partial\lambda}
%\frac{\partial \lambda_{q'}}{\partial x^I}(x^*_{q'})
= : G_k.
\end{equation}
%\begin{itemize}
%\item
%the forces of infection are homogeneous in the population, in the sense that
%\[
%\forall k=1,\dots, m,\quad \forall p,p'\in\cP,\qquad
%\lambda_{[p]_{\cR(\cI_k)}} = \lambda_{[p']_{\cR(\cI_k)}} =: \lambda_k;
%\]
%
%\item
%the forces of infection $\lambda_k$ are positively homogeneous of degree $0$, that is $\lambda_k(\alpha x)=\lambda_k(x)$ for any $x\in\Rset_+^n$ and any positive scalar $\alpha$.
%\end{itemize}
Then, the basic reproduction number given in~\eqref{eq559} is equal to
\begin{equation}
\label{eq560}
\cR_0 = \rho\left(
e^{\left(
V^I+G_m
\right)
(\tau_m-\tau_{m-1})} \dots\ e^{\left(
V^I+G_1
\right)
(\tau_1-\tau_0)}
\right).
\end{equation}
\end{corollary}

A major interest of Corollary~\ref{co0}, when it applies, is of course to reduce the determination of the basic offspring number to the computation of the spectral radius of a square matrix of size $n_I$.

To apply this result,
%Corollary~\ref{co0},
the conditions of transmission of the infection should be identical in all classes and at all times, in the sense precisely defined by~\eqref{eq789}.
The function $F_+$ being the same for the whole system (independent of the class and time), this condition relates only to the forces of infection $\lambda_q$.
Except in very special conditions where the sizes of the classes match exactly, one has to assume that, for any $k\in\{1,\dots, m\}$, $q\in\cQ(\cI_k)$, the function
\[
x^* \mapsto
\partial_{x^I} \left[
F_+(\lambda_q(x))x^{U*}
\right](x^*)
%\frac{\partial \left[
%F_+(\lambda)x^{U*}
%\right]}{\partial\lambda}
%\frac{\partial \lambda_q}{\partial x^I}(x^*)
\]
is {\em constant on the cone $\Xbf_\eq$}.
In particular, due to the linearity of the first factor with respect to $x^{U*}$, one should have, for any scalar $\alpha > 0$,
 \[
\frac{\partial \left[
F_+(\lambda)\alpha x^{U*}
\right]}{\partial\lambda}
\frac{\partial \lambda_q}{\partial x^I}(\alpha x^*)
= \frac{\partial \left[
F_+(\lambda) x^{U*}
\right]}{\partial\lambda}
\frac{\partial \lambda_q}{\partial x^I}(x^*)
= \frac{\partial \left[
F_+(\lambda) \alpha x^{U*}
\right]}{\partial\lambda}
\frac{1}{\alpha}
\frac{\partial \lambda_q}{\partial x^I}(x^*).
\]
A natural condition is therefore that the partial derivative with respect to $x^I$ of the functions $\lambda_q$ defining the forces of infection, are homogeneous of degree $-1$.
It is sufficient for this that the functions $\lambda_q$ themselves are {\em homogeneous of degree $0$}, a characteristic of frequency-dependent transmission.

One shows without difficulty that~\eqref{eq789} is verified e.g.~when the functions $\lambda_q$ are all identical in every class $q\in\cQ(\cI_k)$, and constant on the cone $\Xbf_\eq$.
(This  constant function is homogeneous of degree $0$.)
%As a matter of fact, in such a case, the function $\frac{\partial \lambda}{\partial x^I}$ is positively homogeneous of degree $-1$.
%This implies that, for any scalar $\alpha\geq 0$,
% \[
%\frac{\partial \left[
%F_+(\lambda)\alpha x^{U*}_q
%\right]}{\partial\lambda}
%\frac{\partial \lambda}{\partial x^I}(\alpha x^*_q)
%= \frac{\partial \left[
%F_+(\lambda)\alpha x^{U*}_q
%\right]}{\partial\lambda}
%\frac{1}{\alpha} \frac{\partial \lambda}{\partial x^I}(x^*_q)
%= \frac{\partial \left[
%F_+(\lambda) x^{U*}_q
%\right]}{\partial\lambda}
%\frac{\partial \lambda}{\partial x^I}(x^*_q),
%\]
%and the collinearity assumption allows to obtain~\eqref{eq789}, whatever the equilibrium values $x_q^*$.
Thus, Corollary~\ref{co0} establishes in particular that, if the conditions of the transmission are identical in all locations (classes), then the value of the basic reproduction number is identical {\em whatever the mixing conditions (that is, whatever the commutation pattern and the number and composition of the classes between two successive switchings)}.
Up to our knowledge, this property of invariance relatively to the mixing is new.
See an application to the SIR model in Section~\ref{se524}.

As a last remark, notice that when Corollary~\ref{co0} holds, the evolution of the infected compartments of any subpopulation $p\in\cP$ in the vicinity of $x^*_{I,p}$ is identical  on $\cI_k$ to that of the uncoupled autonomous linear equation $\dot x_{I,p} = (V^I+G_k) (x_{I,p}-x^*_{I,p})$.

\begin{proof}[Proof of Corollary~\rm\ref{co0}]
Let $v\in\Rset^{n_I}$, and $k\in\{1,\dots,m\}$.
For any $q\in\cQ(\cI_k)$, one has
\[
(\bfo_{|q|} \otimes v)\t \cB_q
= (\bfo_{|q|}\t \otimes v\t (V^I + F_q))
\]
for $\cB_q$ in~\eqref{eq78b}.

Then, left-multiplying the matrix $e^{\cJ_{II,k}(\cX^*)(\tau_k-\tau_{k-1})}$ by the vector $(\bfo_{|\cP|} \otimes v)\t$ yields
\begin{multline*}
(\bfo_{|\cP|} \otimes v)\t e^{\cJ_{II,k}(\cX^*)(\tau_k-\tau_{k-1})}\\
= \begin{pmatrix}
(\bfo_{|q_1|}\t \otimes v\t e^{(V^I + F_{q_1}) (\tau_k-\tau_{k-1})})
& \dots &
(\bfo_{|q_{|\cQ(\cI_k)|}|}\t \otimes v\t e^{(V^I + F_{q_{|\cQ(\cI_k)|}}) (\tau_k-\tau_{k-1})})
\end{pmatrix}.
\end{multline*}

We now use assumption~\eqref{eq789}, which states that $F_{q_1} = \dots = F_{q_{|\cQ(\cI_k)|}} = G_k$, to deduce
\[
(\bfo_{|\cP|} \otimes v)\t e^{\cJ_{II,k}(\cX^*)(\tau_k-\tau_{k-1})}
= (\bfo_{|\cP|}\t \otimes v\t e^{(V^I + G_k) (\tau_k-\tau_{k-1})}),
\]
as the number of elements in all classes is equal to the total number of subpopulations, see~\eqref{eq376}.
One thus obtains
\begin{equation}
\label{eq791}
(\bfo_{|\cP|} \otimes v)\t e^{\cJ_{II,k}(\cX^*)(\tau_k-\tau_{k-1})}
= (\bfo_{|\cP|} \otimes v)\t (I_{|\cP|} \otimes e^{(V^I + G_k) (\tau_k-\tau_{k-1})}).
\end{equation}

From the previous formula, one deduces the same property for $e^{\hcJ_{II,k}(\cX^*)}$.
In effect,~\eqref{eq599b} yields
\begin{eqnarray*}
(\bfo_{|\cP|} \otimes v)\t e^{\hcJ_{II,k}(\cX^*)(\tau_k-\tau_{k-1})}
& = &
(\bfo_{|\cP|} \otimes v)\t (\cM_k\t\otimes I_{n_I}) e^{\cJ_{II,k}(\cX^*)(\tau_k-\tau_{k-1})} (\cM_k\otimes I_{n_I}) \\
& = &
(\bfo_{|\cP|} \otimes v)\t e^{\cJ_{II,k}(\cX^*)(\tau_k-\tau_{k-1})} (\cM_k\otimes I_{n_I}) \\
& &
\qquad \text{(because $\bfo_{|\cP|} \t \cM_k =\bfo_{|\cP|} \t $, as $\cM_k$ is a permutation matrix)}\\
& = &
(\bfo_{|\cP|} \otimes v)\t (I_{|\cP|} \otimes e^{(V^I + G_k) (\tau_k-\tau_{k-1})}) (\cM_k\otimes I_{n_I})\\
& &
\qquad \text{(thanks to~\eqref{eq791})}\\
& = &
(\bfo_{|\cP|} \otimes v)\t (\cM_k \otimes e^{(V^I + G_k) (\tau_k-\tau_{k-1})})\\
& = &
(\bfo_{|\cP|} \otimes v)\t (I_{|\cP|} \otimes e^{(V^I + G_k) (\tau_k-\tau_{k-1})})\\
& &
\qquad \text{(because $\bfo_{|\cP|} \t \cM_k =\bfo_{|\cP|} \t $)}\\
& = &
(\bfo_{|\cP|}\t \otimes v\t e^{(V^I + G_k) (\tau_k-\tau_{k-1})})\\
& = &
(\bfo_{|\cP|} \otimes e^{(V^I + G_k)\t (\tau_k-\tau_{k-1})} v)\t.
\end{eqnarray*}

In order to finish the proof, one applies this property recursively for $k=m,\dots,1$:
\begin{eqnarray*}
\lefteqn{(\bfo_{|\cP|} \otimes v)\t e^{\hcJ_{II,m}(\cX^*)(\tau_m-\tau_{m-1})} \dots\ e^{\hcJ_{II,2}(\cX^*)(\tau_2-\tau_1)}\ e^{\hcJ_{II,1}(\cX^*)(\tau_1-\tau_0)}}\\
& = &
(\bfo_{|\cP|} \otimes e^{\left(
V^I+G_m
\right)\t
(\tau_m-\tau_{m-1})}v)\t
e^{\hcJ_{II,m-1}(\cX^*)(\tau_{m-1}-\tau_{m-2})} \dots\ e^{\hcJ_{II,2}(\cX^*)(\tau_2-\tau_1)}\ e^{\hcJ_{II,1}(\cX^*)(\tau_1-\tau_0)}\\
& = &
(\bfo_{|\cP|} \otimes e^{\left(
V^I+G_1
\right)\t
(\tau_1-\tau_0)}
\dots
e^{\left(
V^I+G_m
\right)\t
(\tau_m-\tau_{m-1})}
v)\t\\
& = &
(\bfo_{|\cP|} \otimes v)\t (I_{|\cP|} \otimes e^{\left(
V^I+G_m
\right)
(\tau_m-\tau_{m-1})} \dots\ e^{\left(
V^I+G_1
\right)
(\tau_1-\tau_0)}).
\end{eqnarray*}
Choose now for $v\in\Rset^{n_I}$ the positive eigenvector corresponding to the dominant eigenvalue of the positive matrix $e^{\left(
V^I+G_m
\right)
(\tau_m-\tau_{m-1})} \dots\ e^{\left(
V^I+G_1
\right)
(\tau_1-\tau_0)}$.
The vector $(\bfo_{|\cP|} \otimes v)\t$ is then a left-eigenvector of the matrix $e^{\hcJ_{II,m}(\cX^*)(\tau_m-\tau_{m-1})} \dots\ e^{\hcJ_{II,2}(\cX^*)(\tau_2-\tau_1)}\ e^{\hcJ_{II,1}(\cX^*)(\tau_1-\tau_0)}$.
Being positive, it corresponds to its dominant eigenvalue.
This shows~\eqref{eq560} and achieves the demonstration of Corollary~\ref{co0}.
\end{proof}

\subsubsection{Computing the exponentials of the blocks $\cB_q$ in the Jacobian matrix $\cJ_{II,k}(\cX^*)$}
\label{se5235}

The following result is useful to reduce the dimension of the matrix product that appears in the expression of $\cR_0$ in~\eqref{eq559}.

\begin{lemma}
\label{le87}
For $\cB_q$ defined in~\eqref{eq78b}, one has for any $t\geq 0$,
\begin{multline}
\label{eq288}
e^{\cB_qt}
= \left(
I_{|q|}\otimes e^{V^It}
\right)\\
+ \int_0^t
\begin{pmatrix}
e^{V^I(t-s)} \left.
\partial_{x^I}
[F_+(\lambda_q(x))x^{U*}_{q,1}]
\right|_{x=x^*_q}\\
\vdots\\
e^{V^I(t-s)} \left.
\partial_{x^I}
[F_+(\lambda_q(x))x^{U*}_{q,|q|}]
\right|_{x=x^*_q}
\end{pmatrix}
%\hspace{-.5cm}
e^{\left(
V^I+\left.
\partial_{x^I} \left[
F_+(\lambda_q(x))x^{U*}_q
\right]
\right|_{x=x^*_q}
\right)
s} ds\
(\bfo_{|q|} \otimes I_{n_I})\t.
\end{multline}
In particular, for any infection model~\eqref{eq300} presenting a {\em unique} infected compartment ($n_I=1$), one has
\begin{equation}
\label{eq535}
e^{\cB_q t}
= e^{V^It} I_{|q|}
+ e^{V^It} \left(
e^{\left.
\partial_{x^I} \left[
F_+(\lambda_q(x))x^{U*}_q
\right]
\right|_{x=x^*_q} t}
-1
\right)
\begin{pmatrix}
 \frac{\left.
\partial_{x^I}
[F_+(\lambda_q(x))x^{U*}_{q,1}]
\right|_{x=x^*_q}}
{\left.
\partial_{x^I} \left[
F_+(\lambda_q(x))x^{U*}_q
\right]
\right|_{x=x^*_q}}\\
\vdots\\
 \frac{\left.
\partial_{x^I}
[F_+(\lambda_q(x))x^{U*}_{q,|q|}]
\right|_{x=x^*_q}}
{\left.
\partial_{x^I} \left[
F_+(\lambda_q(x))x^{U*}_q
\right]
\right|_{x=x^*_q}}
\end{pmatrix} \bfo_{|q|}\t.
\end{equation}
\end{lemma}

%Thanks to Lemma~\ref{le87}, u
Using~\eqref{eq288} to compute $e^{\cB_q t}$ amounts essentially to compute the $|q|+1$ matrices of size $n_I\times n_I$
\[
e^{V^It} \qquad \text{and} \qquad
\int_0^t
e^{-V^Is} \left.
\partial_{x^I}
[F_+(\lambda_q(x))x^{U*}_{q,i}]
\right|_{x=x^*_q}
e^{\left(
V^I+\left.
\partial_{x^I} \left[
F_+(\lambda_q(x))x^{U*}_q
\right]
\right|_{x=x^*_q}
\right)
s} ds,\quad i=1,\dots,|q|.
\]

Notice also that formula~\eqref{eq535}, which holds true when $n_I=1$, may be generalised for $n_I>1$ when the matrix $V^I$ commutes with every $\left.
\partial_{x^I}
[F_+(\lambda_q(x))x^{U*}_{q,i}]
\right|_{x=x^*_q}$, $i=1,\dots, |q|$.

\begin{proof}[Proof of Lemma~\rm\ref{le87}]

By definition, $e^{\cB_qt}$ is characterised by the fact that: for any $y(0)\in\Rset^{n_I|q|}$,
\[
y(t) = e^{\cB_qt}y(0),\quad t\geq 0 \qquad \Leftrightarrow \qquad \dot y = \cB_qy,\quad t\geq 0.
\]
Let $\by := (\bfo_{|q|} \otimes I_{n_I})\t y = y_1+\cdots+y_q$.
From the structure of $\cB_q$ in~\eqref{eq78b} and the fact that $x^{U*}_{q,1} +\cdots+ x^{U*}_{q,|q|}= x^{U*}_q$, one deduces, for any $i=1,\dots, |q|$,
\[
\dot \by = \left(
V^I + \left.
\partial_{x^I} \left[
F_+(\lambda_q(x))x^{U*}_q
\right]
\right|_{x=x^*_q}
\right) \by,\qquad
\dot y_i = V^Iy_i + \left.
\partial_{x^I} \left[
F_+(\lambda_q(x))x^{U*}_{q,i}
\right]
\right|_{x=x^*_q}
\by,
\]
which writes as the following triangular system
\[
\begin{pmatrix}
\dot \by \\ \dot y_i
\end{pmatrix}
= \begin{pmatrix}
V^I + \left.
\partial_{x^I} \left[
F_+(\lambda_q(x))x^{U*}_q
\right]
\right|_{x=x^*_q}
& 0_{n_I\times n_I} \\
 \left.
\partial_{x^I} \left[
F_+(\lambda_q(x))x^{U*}_{q,i}
\right]
\right|_{x=x^*_q}
& V^I
\end{pmatrix}
\begin{pmatrix}
\by \\ y_i
\end{pmatrix}.
\]
By integration of the first equation one gets
\[
\by(t) = e^{\left(
V^I+\left.
\partial_{x^I} \left[
F_+(\lambda_q(x))x^{U*}_q
\right]
\right|_{x=x^*_q}
\right)
t} \by(0)
= e^{\left(
V^I+\left.
\partial_{x^I} \left[
F_+(\lambda_q(x))x^{U*}_q
\right]
\right|_{x=x^*_q}
\right)
t} (\bfo_{|q|} \otimes I_{n_I})\t y(0);
\]
and the second equation thus yields, for any $i=1,\dots, |q|$,
\[
y_i(t) = e^{V^It} y_i(0)
+ \int_0^t e^{V^I(t-s)} \left.
\partial_{x^I} \left[
F_+(\lambda_q(x))x^{U*}_{q,i}
\right]
\right|_{x=x^*_q}
e^{\left(
V^I+\left.
\partial_{x^I} \left[
F_+(\lambda_q(x))x^{U*}_q
\right]
\right|_{x=x^*_q}
\right)
s} ds\
(\bfo_{|q|} \otimes I_{n_I})\t y(0).
\]
This yields
\begin{multline*}
%\label{eq288}
y(t)
= \left(
I_{|q|}\otimes e^{V^It}
\right) y(0)\\
+ \int_0^t
\left(
I_{|q|}\otimes e^{V^I(t-s)}
\right)
\left.
\begin{pmatrix}
\partial_{x^I}
[F_+(\lambda_q(x))x^{U*}_{q,1}]\\
\vdots\\
\partial_{x^I}
[F_+(\lambda_q(x))x^{U*}_{q,|q|}]
\end{pmatrix}
\right|_{x=x^*_q}
\hspace{-.5cm}
e^{\left(
V^I+\left.
\partial_{x^I} \left[
F_+(\lambda_q(x))x^{U*}_q
\right]
\right|_{x=x^*_q}
\right)
s} ds\
(\bfo_{|q|} \otimes I_{n_I})\t y(0),
\end{multline*}
and finally, as $y(t) = e^{\cB_qt} y(0)$, one identifies
\begin{multline*}
%\label{eq288}
e^{\cB_qt}
= \left(
I_{|q|}\otimes e^{V^It}
\right)\\
+ \int_0^t
\left(
I_{|q|}\otimes e^{V^I(t-s)}
\right)
\left.
\begin{pmatrix}
\partial_{x^I}
[F_+(\lambda_q(x))x^{U*}_{q,1}]\\
\vdots\\
\partial_{x^I}
[F_+(\lambda_q(x))x^{U*}_{q,|q|}]
\end{pmatrix}
\right|_{x=x^*_q}
\hspace{-.5cm}
e^{\left(
V^I+\left.
\partial_{x^I} \left[
F_+(\lambda_q(x))x^{U*}_q
\right]
\right|_{x=x^*_q}
\right)
s} ds\
(\bfo_{|q|} \otimes I_{n_I})\t.
\end{multline*}
which yields~\eqref{eq288}.

Deduction of formula~\eqref{eq535} in case where $n_I=1$ and all matrix components are then scalar, is then straightforward.
\end{proof}

\subsubsection{An illustrative example}
\label{se524}

$\bullet$
In order to illustrate the previous result, we consider the simple example shown in Figures~\ref{fi1} and~\ref{fi2}.
It consists in five subpopulations, $\cP=\{p_1,p_2,p_3,p_4,p_5\}$, with $T$-periodic commutations between two different modes.
The {\em reference ordering} mentioned in Section~\ref{se622} is for example the order of the enumeration of the elements of $\cP$ in Figure~\ref{fi1}.
To fix the ideas, we put $\cI := \Rset_+\setminus\Nset$, and $\cI = \cI_1\cup\cI_2$, with
\[
\cI_1 := (0,\tau) \mod T,\qquad \cI_2 := (\tau,T) \mod T,
\]
for some $0< \tau < T$.
This is a switching pattern from the set $\Lambda_2$ defined in~\eqref{eq58}.
We assume that the corresponding partitions are given by the two partitions $\cR_1, \cR_2$ presented in Figure~\ref{fi2}, that is
\begin{equation}
\label{eq601}
\cR(\cI_1)=\cR_1 := \{ \{p_1,p_2,p_3\}, \{p_4,p_5\} \},\qquad
\cR(\cI_2)=\cR_2 := \{ \{p_1\}, \{p_2,p_5\}, \{p_3,p_4\} \},
\end{equation}
in such a way that $m=2$, $|\cQ(\cI_1)| = 2, |\cQ(\cI_2)| = 3$ for the quotient sets $\cQ(\cI_1)=\cQ_1=\cP\setminus\cR_1$, $\cQ(\cI_2)=\cQ_2=\cP\setminus\cR_2$.

Choosing for example to order the classes of subpopulations, and the subpopulations within each class, in the order in which they appear in~\eqref{eq601}, yields the correspondence
\[
p_1 = p_{1,1},\ 
p_2 = p_{1,2},\ 
p_3 = p_{1,3},\ 
p_4 = p_{2,1},\ 
p_5 = p_{2,2}
\]
during $\cI_1$, and
\[
p_1 = p_{1,1},\ 
p_2 = p_{2,1},\ 
p_3 = p_{3,1},\ 
p_4 = p_{3,2},\ 
p_5 = p_{2,2}
\]
during $\cI_2$.
With this choice, the permutation matrices $\cM_k$, $k=1,2$, defined in Section~\ref{se622} are:
\begin{equation}
\label{eq794}
\cM_1 := I_5,\qquad
\cM_2 :=
\begin{pmatrix}
1 & 0 & 0 & 0 & 0\\
0 & 1 & 0 & 0 & 0\\
0 & 0 & 0 & 0 & 1\\
0 & 0 & 1 & 0 & 0\\
0 & 0 & 0 & 1 & 0
\end{pmatrix}.
\end{equation}

Theorem~\ref{th5} provides the exact value of $\cR_0$, namely
\begin{equation}
\label{eq793}
\cR_0
= \rho\left(
e^{\hcJ_{II,2}(\cX^*)(T-\tau)} e^{\hcJ_{II,1}(\cX^*)\tau}
\right)
= \rho\left(
\cM_2\t e^{\cJ_{II,2}(\cX^*)(T-\tau)} \cM_2 \ e^{\cJ_{II,1}(\cX^*)\tau}
\right).
\end{equation}
When $\cR_0 < 1$, Theorem~\ref{th5} yields {\em local asymptotic stability} of the equilibrium point $\cX^*$, due to the fact that the first identity in \eqref{eq188} holds here with an equality, see Remark~\ref{ex4}.
When $\cR_0 > 1$, Theorem~\ref{th5} yields instability.

\noindent $\bullet$
We consider e.g.~the SIR model~\eqref{eq5}, for which $n=3$, $n_I=1$ and $n_\lambda=1$.
The global evolution is governed by equation~\eqref{eq2}, that is here
\begin{subequations}
\begin{align}
\dot I_i &= \beta_1(\cI_1) \frac{I_1+I_2+I_3}{N_1+N_2+N_3} S_i -(\gamma+\mu)I_i,\qquad i=1,2,3\\
\dot S_i &= -\beta_1(\cI_1) \frac{I_1+I_2+I_3}{N_1+N_2+N_3} S_i + \mu(I_i+R_i),\qquad i=1,2,3,\\
\dot I_i &= \beta_2(\cI_1) \frac{I_4+I_5}{N_4+N_5} S_i -(\gamma+\mu)I_i,\qquad i=4,5,\\
\dot S_i &= -\beta_2(\cI_1) \frac{I_4+I_5}{N_4+N_5} S_i + \mu(I_i+R_i),\qquad i=4,5,\\
\dot R_i &= \gamma I_i -\mu R_i,\qquad i=1,2,3,4,5,
\end{align}
\end{subequations}
on $\cI_1$, and
\begin{subequations}
\begin{align}
\dot I_1 &= \beta_1(\cI_2) \frac{I_1}{N_1} S_1 -(\gamma+\mu)I_1,\\
\dot S_1 &= -\beta_1(\cI_2) \frac{I_1}{N_1} S_1 + \mu(I_1+R_1),\\
\dot I_i &= \beta_2(\cI_2) \frac{I_2+I_5}{N_2+N_5} S_i -(\gamma+\mu)I_i,\qquad i=2,5,\\
\dot S_i &= -\beta_2(\cI_2) \frac{I_2+I_5}{N_2+N_5} S_i + \mu(I_i+R_i),\qquad i=2,5,\\
\dot I_i &= \beta_3(\cI_2) \frac{I_3+I_4}{N_3+N_4} S_i -(\gamma+\mu)I_i,\qquad i=3,4,\\
\dot S_i &= -\beta_3(\cI_2) \frac{I_3+I_4}{N_3+N_4} S_i + \mu(I_i+R_i),\qquad i=3,4,\\
\dot R_i &= \gamma I_i -\mu R_i,\qquad i=1,2,3,4,5,
\end{align}
\end{subequations}
on $\cI_2$.

The decomposition of the right-hand side $f$ for the SIR model~\eqref{eq5} has been given in the formulas~\eqref{eq613}.
In particular,
%\begin{subequations}
%\label{eq603}
\begin{equation*}
%\label{eq603a}
V^I = -(\gamma+\mu),\qquad
\lambda(x) = \beta\frac{I}{S+I+R},\qquad
F_+(\lambda) = \lambda \begin{pmatrix} 1 & 0\end{pmatrix},
\end{equation*}
and for any $x^*\in\Xbf_\eq$,
\begin{equation}
\label{eq603b}
\frac{\partial \left[
F_+(\lambda)x^{U*}
\right]}{\partial\lambda}
\frac{\partial \lambda}{\partial x^I}(x^*)
= \begin{pmatrix} 1 & 0 \end{pmatrix} x^{U*}
\frac{\beta}{\bfo_{n_U}\t x^{U*}}
= \beta \frac{S^*}{S^*} = \beta.
\end{equation}
%\end{subequations}

\noindent $\bullet$
With these computations, let us now apply the preceding results to analyze the behaviour of this system, in the vicinity of a given equilibrium point $\cX^*$.
The latter is uniquely defined by the numbers of susceptible individuals $S_i^*$, $i=1,\dots, 5$, in the five subpopulations.
At equilibrium, all individuals are susceptible and one has necessarily $S_i^*=|p_i|$.
For the subpopulations chosen in Figure~\ref{fi1} this means 
\begin{equation}
\label{eq604}
S_1^*= 6,\qquad S_2^*= 1,\qquad S_3^*= 5,\qquad S_4^*= 3,\qquad S_5^*= 7.
\end{equation}

In the general case, applying~\eqref{eq793} necessitates first to assess the value of the exponential of the two matrices $\hcJ_{II,1}(\cX^*)\tau$ and $\hcJ_{II,2}(\cX^*)(T-\tau)$.
The first matrix may be computed from the fact that $\hcJ_{II,1}(\cX^*)+(\gamma+\mu)I_5 = \cJ_{II,1}(\cX^*)+(\gamma+\mu)I_5$ is equal to 
\[
\begin{pmatrix}
\beta_1(\cI_1) \frac{S_1^*}{S_1^*+S_2^*+S_3^*} &  \beta_1(\cI_1) \frac{S_1^*}{S_1^*+S_2^*+S_3^*} &  \beta_1(\cI_1) \frac{S_1^*}{S_1^*+S_2^*+S_3^*} & 0 & 0\\
\beta_1(\cI_1) \frac{S_2^*}{S_1^*+S_2^*+S_3^*} & \beta_1(\cI_1) \frac{S_2^*}{S_1^*+S_2^*+S_3^*} & \beta_1(\cI_1) \frac{S_2^*}{S_1^*+S_2^*+S_3^*} & 0 & 0\\
\beta_1(\cI_1) \frac{S_3^*}{S_1^*+S_2^*+S_3^*} & \beta_1(\cI_1) \frac{S_3^*}{S_1^*+S_2^*+S_3^*} & \beta_1(\cI_1) \frac{S_3^*}{S_1^*+S_2^*+S_3^*} & 0 & 0\\
0 & 0 & 0 & \beta_2(\cI_1) \frac{S_4^*}{S_4^*+S_5^*} & \beta_2(\cI_1) \frac{S_4^*}{S_4^*+S_5^*}\\
0 & 0 & 0 & \beta_2(\cI_1) \frac{S_5^*}{S_4^*+S_5^*} & \beta_2(\cI_1) \frac{S_5^*}{S_4^*+S_5^*}
\end{pmatrix}.
\]
The ad hoc ordering introduced in Section~\ref{se622} provides the second matrix $\cJ_{II,2}(\cX^*)+(\gamma+\mu)I_5$ with a block diagonal structure, for example:
\[
\begin{pmatrix}
\beta_1(\cI_2) &  0 & 0 & 0 & 0\\
0 & \beta_2(\cI_2) \frac{S_2^*}{S_2^*+S_5^*} & \beta_2(\cI_2) \frac{S_2^*}{S_2^*+S_5^*} & 0 & 0\\
0 & \beta_2(\cI_2) \frac{S_5^*}{S_2^*+S_5^*} & \beta_2(\cI_2) \frac{S_5^*}{S_2^*+S_5^*} & 0 & 0\\
0 & 0 & 0 & \beta_3(\cI_2) \frac{S_3^*}{S_3^*+S_4^*} & \beta_3(\cI_2) \frac{S_3^*}{S_3^*+S_4^*}\\
0 & 0 &0 & \beta_3(\cI_2) \frac{S_4^*}{S_3^*+S_4^*} & \beta_3(\cI_2) \frac{S_4^*}{S_3^*+S_4^*}
\end{pmatrix}.
\]
The value back in the reference ordering is such that $\hcJ_{II,2}(\cX^*)+(\gamma+\mu)I_5$ equals
\[
\begin{pmatrix}
\beta_1(\cI_2) &  0 & 0 & 0 & 0\\
0 & \beta_2(\cI_2) \frac{S_2^*}{S_2^*+S_5^*} & 0 & 0 & \beta_2(\cI_2) \frac{S_2^*}{S_2^*+S_5^*}\\
0 & 0 & \beta_3(\cI_2) \frac{S_3^*}{S_3^*+S_4^*} & \beta_3(\cI_2) \frac{S_3^*}{S_3^*+S_4^*} & 0\\
0 & 0 & \beta_3(\cI_2) \frac{S_4^*}{S_3^*+S_4^*} & \beta_3(\cI_2) \frac{S_4^*}{S_3^*+S_4^*} & 0\\
0 & \beta_2(\cI_2) \frac{S_5^*}{S_2^*+S_5^*} & 0 & 0 & \beta_2(\cI_2) \frac{S_5^*}{S_2^*+S_5^*}
\end{pmatrix}.
\]

Notice that in general the two matrices $\hcJ_{II,1}(\cX^*)$ and $\hcJ_{II,2}(\cX^*)$ do {\em not} commute; so that the product of the exponentials is not the exponential of the sum.

\noindent $\bullet$
Two degenerate cases that present a constant contact pattern are easy to study, as they actually boil down to the situation of Theorem~\ref{th2}.
These are the following.
\begin{itemize}
\item[$\circ$]
When $\tau=0$, then~\eqref{eq793} yields
\[
\cR_0
= \rho\left(
\cM_2\t e^{\cJ_{II,2}(\cX^*)T} \cM_2
\right)
= \rho\left(
e^{\cJ_{II,2}(\cX^*)T}
\right)
= e^{-(\gamma+\mu)T} \max
\left\{
e^{\beta_1(\cI_2)T}; e^{\beta_2(\cI_2)T}; e^{\beta_3(\cI_2)T}
\right\}.
\]
\item[$\circ$]
 When $\tau=T$, then~\eqref{eq793} gives
\[
\cR_0
= \rho\left(
e^{\cJ_{II,1}(\cX^*)T}
\right)
= e^{-(\gamma+\mu)T} \max
\left\{
e^{\beta_1(\cI_1)T}; e^{\beta_2(\cI_1)T}
\right\}.
\]
\end{itemize}

\noindent $\bullet$
Let us compute $e^{\cJ_{II,1}(\cX^*)\tau}$ and $e^{\cJ_{II,2}(\cX^*)(T-\tau)}$ in the general case, thanks to formula~\eqref{eq535}.
In view of~\eqref{eq604}, one has
\begin{gather*}
\frac{S_1^*}{S_1^*+S_2^*+S_3^*} = \frac{1}{2},\quad
\frac{S_2^*}{S_1^*+S_2^*+S_3^*} = \frac{1}{12},\quad
\frac{S_3^*}{S_1^*+S_2^*+S_3^*} = \frac{5}{12},\quad
\frac{S_4^*}{S_4^*+S_5^*} = \frac{3}{10},\quad
\frac{S_5^*}{S_4^*+S_5^*} = \frac{7}{10},\\
\frac{S_2^*}{S_2^*+S_5^*} = \frac{1}{8},\quad
\frac{S_5^*}{S_2^*+S_5^*} = \frac{7}{8},\quad
\frac{S_3^*}{S_3^*+S_4^*} = \frac{5}{8},\quad
\frac{S_4^*}{S_3^*+S_4^*} = \frac{3}{8}.
\end{gather*}
One deduces that
\[
e^{\cJ_{II,1}(\cX^*)\tau}
= e^{-(\gamma+\mu)\tau}
\left(
I_5 +
\diag\left\{
\frac{1}{12}(e^{\beta_1(\cI_1)\tau}-1)
\begin{pmatrix}
6 & 6 & 6\\
1 & 1 & 1\\
5 & 5 & 5
\end{pmatrix};
\frac{1}{10}
(e^{\beta_2(\cI_1)\tau}-1)
\begin{pmatrix}
3 & 3\\7 & 7
\end{pmatrix}
\right\}
\right)
\]
and
\begin{multline*}
e^{\cJ_{II,2}(\cX^*)(T-\tau)}
= e^{-(\gamma+\mu)(T-\tau)}\\
\times
\left(
I_5 +
\diag\left\{
(e^{\beta_1(\cI_2)(T-\tau)}-1);
\frac{1}{8}
(e^{\beta_2(\cI_2)(T-\tau)}-1)
\begin{pmatrix}
1 & 1\\ 7 & 7
\end{pmatrix};
\frac{1}{8}
(e^{\beta_3(\cI_2)(T-\tau)}-1)
\begin{pmatrix}
5 & 5\\3 & 3
\end{pmatrix}
\right\}
\right).
\end{multline*}
Denoting for simplicity
\begin{gather*}
\phi_{11} := \frac{1}{12}(e^{\beta_1(\cI_1)\tau}-1),\qquad
\phi_{21} := \frac{1}{10}(e^{\beta_2(\cI_1)\tau}-1),\\
\phi_{12} := (e^{\beta_1(\cI_2)(T-\tau)}-1),\qquad
\phi_{22} := \frac{1}{8}(e^{\beta_2(\cI_2)(T-\tau)}-1),\qquad
\phi_{32} := \frac{1}{8}(e^{\beta_3(\cI_2)(T-\tau)}-1),
\end{gather*}
with the expression of $\cM_2$ in~\eqref{eq794}, one finally obtains that
\begin{eqnarray}
%\lefteqn{
\hspace{-.6cm}
e^{(\gamma+\mu)T} e^{\hcJ_{II,2}(\cX^*)(T-\tau)} e^{\hcJ_{II,1}(\cX^*)\tau}
%}\\
& = &
\nonumber
\begin{pmatrix}
 1+ \phi_{12} & 0 & 0 & 0 & 0\\
0 & 1+ \phi_{22} & 0 & 0 & \phi_{22}\\
0 & 0 & 1+ 5\phi_{32} & 5\phi_{32} & 0\\
0 & 0 & 3\phi_{32} & 1+ 3\phi_{32} & 0\\
0 & 7\phi_{22} & 0 & 0 & 1+ 7\phi_{22}
\end{pmatrix}
\\
& &
\label{eq998}
\times
\begin{pmatrix}
1+ 6\phi_{11} & 6\phi_{11} & 6\phi_{11} & 0 & 0\\
\phi_{11} & 1+ \phi_{11} & \phi_{11} & 0 & 0\\
5\phi_{11} & 5\phi_{11} & 1+ 5\phi_{11} & 0 & 0\\
0 & 0 & 0 & 1+ 3\phi_{21} & 3\phi_{21}\\
0 & 0 & 0 & 7\phi_{21} & 1+ 7\phi_{21}
\end{pmatrix}.
%& = &
%\begin{pmatrix}
%6\phi_{11}(1+\phi_{21}) & 6\phi_{11}(1+\phi_{22}) & 6\phi_{11}(1+5\phi_{23}) & 30\phi_{11}\phi_{23} & 6\phi_{11}\phi_{22}\\
%\phi_{11}(1+\phi_{21}) & \phi_{11}(1+\phi_{22}) & \phi_{11}(1+5\phi_{23}) & 5\phi_{11}\phi_{23} & (1+\phi_{11})\phi_{22}\\
%5\phi_{11}(1+\phi_{21}) & 5\phi_{11}(1+\phi_{22}) & (1+5\phi_{11})(1+5\phi_{23}) & 5(1+5\phi_{11})\phi_{23} & 5\phi_{11}\phi_{22}\\
%0 & 21\phi_{21}\phi_{22} &  3(1+3\phi_{21})\phi_{23} & (1+3\phi_{21})(1+3\phi_{23}) & 3\phi_{21}(1+7\phi_{22})\\
%0 & 7(1+7\phi_{21})\phi_{22} & 21\phi_{21}\phi_{23} & 7\phi_{21}(1+3\phi_{23}) & (1+7\phi_{21})(1+7\phi_{22})
%\end{pmatrix}.
\end{eqnarray}

\noindent $\bullet$
By comparison, the stationary system usually considered when assuming fast commuter movements in metapopulation models~\cite[Section 7.2.1.4]{Keeling:2008aa} is obtained by considering the weighted mean of the transmission terms.
Accordingly, the blocks of the Jacobian matrix of the corresponding autonomous system are obtained as weighted means, yielding an approximation of $\rho\left(
e^{\hcJ_{II,2}(\cX^*)(T-\tau)} e^{\hcJ_{II,1}(\cX^*)\tau}
\right)$ in~\eqref{eq793} by the expression
\[
\rho\left(
e^{\hcJ_{II,2}(\cX^*)(T-\tau)+\hcJ_{II,1}(\cX^*)\tau}
\right).
\]
Here,
\begin{eqnarray}
\hspace{-.6cm}
\tau \hcJ_{II,1}(\cX^*) + (T-\tau) \hcJ_{II,2}(\cX^*)
& = &
\nonumber
-T(\gamma+\mu)I_5\\
& &
\nonumber
+ \tau \begin{pmatrix}
\frac{1}{2} \beta_1(\cI_1) & \frac{1}{2} \beta_1(\cI_1) & \frac{1}{2} \beta_1(\cI_1) & 0 & 0\\
\frac{1}{12} \beta_1(\cI_1) & \frac{1}{12} \beta_1(\cI_1) & \frac{1}{12} \beta_1(\cI_1) & 0 & 0\\
\frac{5}{12} \beta_1(\cI_1) & \frac{5}{12} \beta_1(\cI_1) & \frac{5}{12} \beta_1(\cI_1) & 0 & 0\\
0 & 0 & 0 & \frac{3}{10} \beta_2(\cI_1) & \frac{3}{10} \beta_2(\cI_1)\\
0 & 0 & 0 & \frac{7}{10} \beta_2(\cI_1) & \frac{7}{10} \beta_2(\cI_1)
\end{pmatrix}\\
& &
\label{eq999}
+ (T-\tau) \begin{pmatrix}
\beta_1(\cI_2) & 0 & 0 & 0 & 0\\
0 & \frac{1}{8} \beta_2(\cI_2) & 0  & 0 & \frac{1}{8} \beta_2(\cI_2)\\
0 & 0 & \frac{5}{8} \beta_3(\cI_2) & \frac{5}{8} \beta_3(\cI_2) & 0\\
0 & 0 & \frac{3}{8} \beta_3(\cI_2) & \frac{3}{8} \beta_3(\cI_2) & 0\\
0 & \frac{7}{8} \beta_2(\cI_2) & 0 & 0 & \frac{7}{8} \beta_2(\cI_2)
\end{pmatrix}.
\end{eqnarray}

\noindent $\bullet$
As a last remark and an illustration of Corollary~\ref{co0}, consider now the situation where the spreading is uniform in the whole population during each of the two phases, that is
\[
\beta_1(\cI_1) = \beta_2(\cI_1) =: \beta(\cI_1),\qquad
\beta_1(\cI_2) = \beta_2(\cI_2) = \beta_3(\cI_2) =: \beta(\cI_2).
\]
One checks that the vector $\bfo_5\t$ is a left eigenvector of the matrix in~\eqref{eq998}, associated to the eigenvalue
\begin{equation}
\label{eq1000}
e^{(\gamma+\mu)T} e^{\beta(\cI_2) (T-\tau)} e^{\beta(\cI_1) \tau}.
\end{equation}
The values of $G_1, G_2$ provided by~\eqref{eq789} are respectively $\beta(\cI_1), \beta(\cI_2)$, see~\eqref{eq603b}.
One then verifies that the value in~\eqref{eq1000} is identical to the one provided by formula~\eqref{eq560}, as expected by Corollary~\ref{co0}.

As a side comment, this value is also equal to the one obtained through use of the stationary approximation, namely the spectral radius of the matrix in~\eqref{eq999}.

\section{Conclusion}
\label{se7}

A new framework for the modelling and analysis of the spread of epidemic diseases in complex commuting populations was proposed in this paper.
It consists in a general class of compartmental ODE models, where the population is divided into homogeneous subpopulations of different sizes.
%This is carried out in the following manner.
At any time is defined a partition of the set of these subpopulations, whose classes are comprised of those currently at a common physical location.
These partitions are piecewise constant and periodic in time, 
% with discrete commutations,
in order to model repetitive contact patterns.
% presenting regularly repeated mixing.
Each subpopulation thus represents here a maximal set of individuals perpetually and homogeneously in contact from the point of view of epidemic spread.
%aggregate (by definition) all subpopulations currently in a given location population is
%At any time, the latter is described by a partition of the set of subpopulations.
%By definition, the classes of this partition aggregate the subpopulations currently in the same place.
%In addition,
On top of this, every subpopulation %in the heterogeneous population
is further subdivided, as is usually the case, into several compartments accounting for the epidemiological status, and more generally for any heterogeneity trait pertinent to describe the disease spread with all the desired precision~\cite{Hethcote:1987aa}.
%, in the spirit of~\cite{Hethcote:1987aa}.

In order to allow for adequate representation of the mixtures and separations of the subpopulations, it was necessary to revisit the epidemiological model.
To this end, we introduced a new general setting showcasing the special role of the force of infection.
It led to consider compartmental models that are jointly linear with respect to the state variable, and affine with respect to the force of infection (itself a function of the state variable).

Based on this formulation, a limited set of assumptions was adopted, which essentially transposes the hypotheses introduced by van den Driessche and Watmough~\cite{Driessche:2002aa} into this larger setting.
The behaviour around the disease-free equilibrium points of such systems, and especially the possibility of epidemic bursts, has then been characterised in terms of spectral radius of matrix products.
In a didactic purpose, we gradually considered the case of a unique subpopulation (Theorem~\ref{th1}); then systems with several subpopulations in stationary contact pattern (Theorem~\ref{th2}); and finally in full generality so-called `epidemiological models in commuting population', that is, systems with several subpopulations interacting through periodically alternating mixing (Theorem~\ref{th5}).
%From a computational point of view, this operation requires the assessment of some spectral radius.
As an interesting consequence, this allowed to show (Corollary~\ref{co0}) that, when the transmission conditions are identical in the whole population, then the value of the basic reproduction number is actually independent of the contact pattern, and is the same than what it would be if all subpopulations were brought together.
Up to our knowledge, such result is original.
An elementary example was provided to illustrate the computations underlying the results.\\

The framework proposed in this paper is believed to be original, and may serve as a basis for studying general model properties.
First of all, the proposed class of compartmental models provides a significant extension and unification of many models previously published in the literature.
This larger setting should help identify the key characteristics of epidemiological models and analyse their behaviour.
Second, the modelling of the spatial aspect in terms of {\em partitions of the population} rather than in terms of {\em fixed physical locations} seems well adapted to the description of the spread of an epidemic in commuting populations.
% keeping track of the group identity is essential.
% point of view adopted here
This option allows to keep track of the identity of each group, %with however a reduction of
while reducing the conceptual and computational complexity necessary to represent the system evolution compared to the classical modelling choices.
The number of variables is here proportional to the number of subpopulation classes defined by the partition, which represent the `active' locations, and not of the number of locations or patches, or a power of this quantity.
This is particularly suitable for an exact description of time-varying conditions of disease propagation, by opposition with the stationary, averaged, settings usually considered.
It allows to analyse theoretically and assess numerically the effects of such approximations.
%more economical representations than the vision in terms of locations (which may be empty at some moments).

The proposed modelling framework should allow for extensive computational investigation.
It permits for example to address the issues of robustness with respect to the parametric uncertainties of the model --- be they related to the infection process (as e.g.~the contact rates, or the latency, recovery and loss of immunity rates) or to the commuting process and the mixing times.
In this framework, the question of the stability robustness of a DFE boils down to the study of the extremal value of the spectral radius over a given set of matrices, or of the generalized spectral radius~\cite{Wirth:2002aa,Jungers:2009aa} of a specific set of matrices.
Also, it should permit comprehensive exploration of the final epidemic size for such periodic commuting systems.
This relevant quantity represents the total number of individuals infected during an outbreak (see e.g.~\cite{Andreasen:2011aa} and references therein), and its computation is quite complicated.
%Last, we believe that independently of the commuting framework, the general class of models proposed here to depict epidemiological processes may have its own interest.
These topics will be the subject of forthcoming research.

\section*{Acknowledgements}
Pierre-Alexandre Bliman is grateful to Carlos Canudas de Wit (GIPSA-lab) for having the opportunity to work together during COVID pandemic in the Inria working group on Healthy mobility control.
The discussion and reflection initiated at that time nurtured the inspiration for the present paper.
He is also grateful to Abderrahman Iggidr (Inria), Alain Rapaport (Inrae) and Gauthier Sallet (Institut \'Elie Cartan de Lorraine) for valuable discussions on the topics considered here.
Boureima Sangar\'e is deeply thankful to the Laboratoire de Mathématiques de Besançon for hosting him for three months and supporting his trip to Sorbonne Universit\'e.
Assane Savadogo is very much thankful to Prof.~Christophe Cambier (Sorbonne Universit\'e), coordinator of the PDI-MSC program of IRD which supported this work. He is also thankful to Inria Paris centre for financial support.
Last, the funding of ANR is acknowledged, under the project ANR-23-CE48-0004.

%{\PABB A.~Savadogo --- IRD and Inria.}

%This paper contains preliminary results.
%Forthcoming research will focus on the use of the proposed setting to establish properties of the epidemiological models in commuting populations.
%We are particularly interested in the properties of the basic reproduction number and of the final epidemic size of such systems.
%
%As said before, in the proposed framework, not only each subpopulation is homogeneous, but the different subpopulations differ only through the schedule of their successive locations.
%However, diversity traits such as differential susceptibility or infectivity may be recovered by considering different compartments to account for such phenomena, more precisely one susceptible compartment for the individuals with the same susceptibility characteristics and so on.

\addcontentsline{toc}{section}{References}
\bibliographystyle{siam}
\bibliography{Biblio-commuting}

%\appendix
%
%\renewcommand{\theequation}{A.\arabic{equation}}
%\setcounter{equation}{0}
%\renewcommand{\thefigure}{A.\arabic{figure}}
%\setcounter{figure}{0}
%\renewcommand{\thetable}{A.\arabic{table}}
%\setcounter{table}{0}
%\renewcommand{\thelemma}{A.\arabic{lemma}}
%\setcounter{lemma}{0}
%\renewcommand{\thetheorem}{A.\arabic{theorem}}
%\setcounter{theorem}{0}
%\renewcommand{\theremark}{A.\arabic{{remark}}}
%\setcounter{remark}{0}
%\renewcommand{\thedefinition}{A.\arabic{definition}}
%\setcounter{definition}{0}

\end{document}